%% file: main.tex
\documentclass[11pt]{article}

\usepackage[T1]{fontenc}
\usepackage[utf8]{inputenc}
\usepackage{amsmath,amssymb,amsthm,mathtools}
\usepackage[margin=1in]{geometry}
\usepackage{microtype}
\usepackage{xspace}
\usepackage{booktabs,tabularx,array}
\usepackage{graphicx}
\usepackage{placeins}
\usepackage{caption}
\usepackage{subcaption}
\usepackage{algorithm}
\usepackage{algpseudocode}
\usepackage{tikz}
\usetikzlibrary{arrows.meta,calc,fit,positioning,shapes.geometric,decorations.pathreplacing}
\usepackage{xcolor}
\usepackage{colortbl}
\usepackage[numbers,square,sort&compress]{natbib}
\usepackage{hyperref}
\usepackage[nameinlink,capitalise,noabbrev]{cleveref}

\hypersetup{
  colorlinks=true,
  linkcolor=black,
  citecolor=black,
  urlcolor=black,
  pdftitle={Massively Parallel Algorithms for Huffman Coding}
}

\newtheorem{theorem}{Theorem}[section]
\newtheorem{lemma}[theorem]{Lemma}

\newtheorem{corollary}[theorem]{Corollary}

\newtheorem*{mainresultrestatement}{Theorem~\ref{thm:main-result}}
\newtheorem*{mergingproblemrestatement}{Problem~\ref{prob:kraft-budget-merging}}

\theoremstyle{definition}
\newtheorem{definition}[theorem]{Definition}

\newtheorem{example}[theorem]{Example}

\newtheorem{problem}[theorem]{Problem}
\newcommand{\Text}{T}
\newcommand{\TextLength}{n}
\newcommand{\Alphabet}{\Sigma}
\newcommand{\AlphabetSize}{m}
\newcommand{\Character}[1]{a_{#1}}
\newcommand{\Frequency}[1]{f_{#1}}
\newcommand{\Code}{C}
\newcommand{\CodeLength}[1]{h_{#1}}
\newcommand{\BinaryStrings}{\{0,1\}^{*}}
\newcommand{\MPC}{\textsf{MPC}\xspace}
\newcommand{\PRAM}{\textsf{PRAM}\xspace}
\newcommand{\PrefixSums}[1]{\textcolor{black}{#1^{+}}}
\newcommand{\PrefixSum}[2]{\textcolor{black}{#1^{+}[#2]}}
\newcommand{\PRAMFanout}{\eta}
\newcommand{\KraftScaleBound}{d_{\max}}
\newcommand{\DepthGroupCount}{g}
\newcommand{\RecordCount}{\textcolor{black}{\hat{n}}}
\newcommand{\LocalMemorySize}{\textcolor{black}{s_{\mathrm{loc}}}}
\newcommand{\FrequencyMultiplicity}[1]{\textcolor{black}{n_{#1}}}
\newcommand{\TotalFrequency}[1]{\textcolor{black}{n_{#1}}}
\newcommand{\ResidualBudgetBound}{b_{\mathrm{res}}}

\newcommand{\Naturals}{\mathbb{N}}

\newcommand{\MinimumEncodedLength}{\mu}
\newcommand{\ExactBlockCost}{\bar{\mu}}
\newcommand{\KnapsackValue}[2]{\textcolor{black}{\psi(#1,#2)}}
\newcommand{\AccumulatedMergeCost}{c}
\newcommand{\RecoveredMergeCost}{c'}
\newcommand{\RatioBound}{\lambda}
\newcommand{\LeafDepth}[1]{h_{#1}}

\newcommand{\MaximumLeafDepth}{h_{\max}}
\newcommand{\AveragingDepth}{\hat{h}}
\newcommand{\BaseWindowDepth}{\textcolor{black}{h}}
\newcommand{\DepthSpread}[1]{\Delta_{#1}}
\newcommand{\GoldenRatio}{\varphi}

\newcommand{\BlockBudgetMax}{b_{\max}}
\newcommand{\FibonacciNumber}[1]{\mathsf{fib}(#1)}
\newcommand{\Representative}[1]{r_{#1}}
\newcommand{\FrequencyBlock}[1]{B_{#1}}
\newcommand{\softO}{\widetilde{O}}

\allowdisplaybreaks

\title{Massively Parallel Algorithms for Huffman Coding}
\author{Masoud Seddighin\footnote{Tehran Institute of Advanced Studies} \and Saeed Seddighin}
\date{}

\begin{document}

\maketitle

\input{Sections/abstract}

\input{Sections/introduction-materials}
\input{Sections/preliminaries}
\input{Sections/our-results}

\input{Sections/bounded-ratio-huffman}
\input{Sections/fast-layered-merging}
\input{Sections/iterative-state-reconstruction}

\input{Sections/putting-pieces-together}
\input{Sections/recovering-codewords}
\input{Sections/pram}

\FloatBarrier
\bibliographystyle{plainnat}
\bibliography{references}

\end{document}

%% file: Sections/abstract.tex
\begin{abstract}
Huffman coding is one of the oldest and most fundamental problems in computer science. Given a string of length $\TextLength$ over a general alphabet, the goal is to assign a binary codeword to each character so that no codeword is a prefix of another and the total encoded length of the string is minimized. Huffman coding is widely used in practical compression systems, including file compression. As modern datasets continue to grow, it is natural to study whether a Huffman code can be constructed efficiently in the massively parallel computation (\MPC) model. The celebrated Huffman coding algorithm admits two straightforward \MPC implementations: for any constant $\epsilon\in(0,1)$, one uses $O(\TextLength^\epsilon)$ memory per machine but requires $\Theta(\log \TextLength)$ rounds, while the other runs in $O(1)$ rounds but requires $\Theta(\sqrt{\TextLength})$ memory per machine. We give the first nontrivial \MPC algorithm for Huffman coding that bypasses both limitations. For every constant $\epsilon>0$, our algorithm uses $O_\epsilon(\log\log \TextLength)$ rounds and $\softO(\TextLength^\epsilon)$ memory per machine, while its total memory and total computation are $\softO(\TextLength)$. 

{\color{black}Although our framework is tailored to the \MPC model, we show that it has applications in the \PRAM setting as well. Specifically, we show that the same algorithm can be implemented in the \PRAM model with total work close to linear. In particular, whenever the text length is at most $2^{O\left(\left(\frac{\log \AlphabetSize}{\log\log \AlphabetSize}\right)^2\right)}$ (slightly superpolynomial), where $\AlphabetSize$ is the alphabet size, our parallel running time is asymptotically better than that of all previously known \PRAM algorithms and is essentially optimal.}
\end{abstract}
\thispagestyle{empty}   

%% file: Sections/introduction-materials.tex
\section{Introduction}
Huffman coding is one of the oldest and most fundamental problems in computer science~\cite{huffman1952, cormen2009,teng1987,kirkpatrick1996}. Given a string of length $\TextLength$ over a general alphabet, the goal is to assign a binary codeword to each character so that no codeword is a prefix of another and the total encoded length of the string is minimized. Huffman coding is widely used in practical compression systems, including file compression.

The celebrated Huffman coding algorithm solves this problem optimally in linear time via a remarkably simple greedy idea. It first represents the text as a sequence of character frequencies $f_1 \leq f_2 \leq \dots \leq f_m$, where $m$ is the size of the alphabet. It then iteratively merges the two smallest frequencies into a single new node with frequency equal to their sum, inserts this merged node back into the (sorted) collection, and repeats this process until a single tree remains. The resulting tree defines the codewords: each leaf corresponds to a character, and its codeword is given by the sequence of left/right branches on the root-to-leaf path.

The study of Huffman coding in the Massively Parallel Computation (\MPC) model is motivated by both practical and theoretical considerations. Practically, \MPC has become the standard paradigm for processing data at scale, underlying frameworks such as MapReduce, Hadoop, and Spark, precisely the setting in which Huffman coding is needed. Theoretically, Huffman coding belongs to a broader class of combinatorial problems with simple, efficient sequential algorithms that become surprisingly challenging to implement in the \MPC model. This class includes the Longest Increasing Subsequence, Edit Distance, the Longest Common Subsequence problems which have all been subject to \MPC studies~\cite{im2017} (STOC'17),~\cite{boroujeni2021} (JACM), ~\cite{hajiaghayi2019} (SODA'19).

The massively parallel computation (\MPC) model consists of a number of machines, each with memory sublinear in the input size. Computation proceeds in rounds: within a round, each machine performs local computation on the data it holds, and between rounds, machines exchange data with one another, again subject to their memory limit. The goal is to minimize the round complexity while keeping the memory per machine, the total memory, and the total work as small as possible.

In this work, we give the first non-trivial \MPC algorithm for Huffman coding. Before stating our result, it is worth first pointing out the trivial ways in which Huffman's algorithm can already be turned into an \MPC algorithm, so as to make clear where the real difficulty lies. We assume in this paper that the input of the Huffman coding problem is a non-binary text of size $n$ and the output is a single integer value which is the shortest length of a binary string that can represent the input via prefix-free codewords.
\setcounter{page}{1}    

\paragraph{Trivial Approach 1.}
The first approach achieves \(O(1)\) rounds using
\(O(\sqrt{\TextLength})\) local memory. The idea is to represent all
 frequencies in a compressed format using a threshold \(\tau=\sqrt{\TextLength}\). For every
\(j\leq\tau\), we store only the number of  frequencies
equal to \(j\), while frequencies larger than \(\tau\) are
stored explicitly. Since the  frequencies always sum to
\(\TextLength\), there are at most \(\sqrt{\TextLength}\) such large
frequencies. Hence, the entire representation has size
\(O(\sqrt{\TextLength})\) and can be gathered on one machine.
That machine then simulates the classical Huffman algorithm using this
compressed representation, maintaining the counters and the explicit
list as new frequencies are created. This gives an \(O(1)\)-round \MPC
algorithm with \(O(\TextLength)\) total memory and \(O(\TextLength)\)
total work.

\paragraph{Trivial Approach 2.}
A second approach uses \(O(\TextLength^{1-\epsilon})\) machines
with \(O(\TextLength^\epsilon)\) memory each, for any fixed
\(\epsilon\in(0,1)\), and runs in \(O(\log\TextLength)\) rounds. The
 frequencies are kept sorted and processed in phases
\(r=0,1,2,\ldots\). Phase \(r\) handles the frequencies in dyadic range
$
I_r=[2^r,2^{r+1}).
$
At the beginning of phase \(r\), at most one  frequency is smaller
than \(2^r\); we call it the \emph{carry}. If the carry and another
 frequency exist, the carry is first merged with the smallest other
frequency. The frequencies in \(I_r\) are
then paired consecutively and all pairs are merged in parallel. If their
number is odd, the last frequency becomes the carry for the next phase.

Every parent created from two frequencies in \(I_r\) has frequency at least
\(2^{r+1}\), so these parallel merges are consistent with a sequential
execution of Huffman's algorithm. Each phase takes \(O(1)\) rounds, and
there are \(O(\log\TextLength)\) phases since every  frequency is at
most \(\TextLength\). Thus, the algorithm uses \(O(\log\TextLength)\)
rounds, \(O(\TextLength)\) total memory, and
\(\softO(\TextLength)\) total computation.

\paragraph{Our Approach.} The $\Omega(\sqrt{n})$ memory requirement of the first approach and the logarithmic round complexity of the second approach are fundamental barriers of these methods that cannot be overcome with modifications. Our algorithm bypasses both limitations by utilizing an entirely different framework. That is, our algorithm requires $O_\epsilon(\log\log \TextLength)$ rounds of computation and requires $\softO(n^{\epsilon})$ memory per machine. The total work and total memory of our algorithm are $\softO(n)$ as desired.

A natural question is how these approaches compare when the input is not the raw text but the frequencies $f_1, f_2, \ldots, f_m$ themselves that are polynomially bounded in $m$ (in this case the input length is $m$ instead of $n$). Approach~2 then requires $O(\log m)$ rounds, $\softO(m^\epsilon)$ memory per machine, and $\softO(m)$ total memory and total work. Approach~1, by contrast, breaks down: its memory per machine grows to $O(\sqrt{\sum_i f_i})$, which becomes impractical as soon as $\sum_i f_i$ grows to quadratic. Our approach is almost unaffected by this change of input model: its round complexity remains $O_\epsilon(\log\log m)$, {\color{black}its memory per machine remains $\softO_\epsilon(m^\epsilon)$, and its total memory and total work are bounded by $\softO_\epsilon(m^{1+\epsilon})$.}

%% file: Sections/preliminaries.tex
\section{Preliminaries}

Let $\Text$ be the input text and let $\TextLength$ be its length. The text is written over a finite alphabet
$\Alphabet=\{\Character{1},\ldots,\Character{\AlphabetSize}\}$, where $\Character{1},\ldots,\Character{\AlphabetSize}$ are the distinct characters that appear in $\Text$. For each $i\in[\AlphabetSize]=\{1,\ldots,\AlphabetSize\}$, let $\Frequency{i}$ be the number of occurrences of $\Character{i}$ in $\Text$. Thus, every frequency is positive and
$
\sum_{i=1}^{\AlphabetSize}\Frequency{i}=\TextLength.
$

A \emph{binary string} is a finite sequence of zeros and ones. A binary code assigns one such string, called a \emph{codeword}, to each character: formally, it is a map $\Code:\Alphabet\to\BinaryStrings$. A binary string $x$ is a \emph{prefix} of a binary string $y$ if $y$ begins with $x$. The code is \emph{prefix-free} if no codeword is a prefix of another. For example, $\{0,10,11\}$ is prefix-free, whereas $\{0,01\}$ is not, because $0$ is a prefix of $01$.

The code transforms the input text into a binary string by replacing each character with its codeword and concatenating the results. Prefix-freeness makes this transformation unambiguous: the encoded text can be read from left to right without separators between consecutive codewords. If $\CodeLength{i}=|\Code(\Character{i})|$ is the length of the codeword assigned to $\Character{i}$, then the encoded text contains
$
  \sum_{i=1}^{\AlphabetSize}\Frequency{i}\CodeLength{i}
$
bits. Huffman coding asks for a prefix-free binary code that minimizes this quantity~\cite{huffman1952}. We denote the minimum encoded length by
\[
  \MinimumEncodedLength(\Text)
  =\min_{\Code\text{ prefix-free}}
    \sum_{i=1}^{\AlphabetSize}\Frequency{i}|\Code(\Character{i})|.
\]
\textcolor{black}{We handle $\AlphabetSize=1$ separately by assigning the empty codeword, giving optimum encoded length zero, and assume $\AlphabetSize\geq2$ in the subsequent analysis.} Our main algorithm focuses on computing the value $\MinimumEncodedLength(\Text)$. We discuss how to recover the codewords themselves later in the paper. \textcolor{black}{For an array $A=(a_1,\ldots,a_k)$, let $\PrefixSums{A}$ denote its prefix-sum array, with $\PrefixSum{A}{0}=0$ and $\PrefixSum{A}{r}=\sum_{i=1}^r a_i$ for $1\leq r\leq k$.}

%% file: Sections/our-results.tex
\section{Our Results and Techniques}
 In this paper, we propose the first nontrivial \MPC algorithm for Huffman coding that goes beyond the limitations of the trivial implementations. For every constant $\epsilon>0$, our algorithm uses $\softO(\TextLength^\epsilon)$ memory per machine and runs in $O_\epsilon(\log\log \TextLength)$ rounds.
Its total memory and total computation are both $\softO(\TextLength)$.

\begin{theorem}
	\label{thm:main-result}
	For every constant $\epsilon>0$, there is an \MPC algorithm that receives a text $\Text$ of length $n$ and computes $\MinimumEncodedLength(\Text)$ in $O_\epsilon(\log\log \TextLength)$ rounds using $\softO(\TextLength^\epsilon)$ memory per machine. Its total memory and total computation are both bounded by $\softO(\TextLength)$.
\end{theorem}

We show later in the paper that the codewords can also be constructed via the same algorithm but for simplicity, in this section we focus on computing the value of $\MinimumEncodedLength(\Text)$. Up to this point, we have viewed Huffman coding through its iterative merge process, which builds a Huffman tree from the bottom up. Our algorithm uses a second, equivalent formulation that instead optimizes directly over the lengths of the codewords. In a prefix-free binary code, the total encoded length of the text depends only on how deep each character sits in the tree, that is, the number of bits in its codeword. The actual pattern of zeros and ones does not affect the cost. Thus, if symbol $a_i$ has frequency $\Frequency{i}$ and codeword length $\CodeLength{i}$, its contribution to the total cost is $\Frequency{i}\CodeLength{i}$, and the overall cost becomes
$
\sum_{i=1}^{\AlphabetSize}\Frequency{i}\CodeLength{i}.
$

This raises a fundamental question: which length vectors can actually be realized by a prefix-free binary code? The classical Kraft--McMillan theorem gives the exact characterization~\cite{kraft1949,mcmillan1956}. It states that positive integers $\CodeLength{1},\ldots,\CodeLength{\AlphabetSize}$ can be the lengths of a prefix-free binary code if and only if
$
\sum_{i=1}^{\AlphabetSize}2^{-\CodeLength{i}}\leq 1.
$
The Kraft--McMillan theorem gives an equivalent formulation of Huffman coding:
\begin{equation}
\label{eq:huffman-length-formulation}
\MinimumEncodedLength(\Text)=\min\left\{\sum_{i=1}^{\AlphabetSize}\Frequency{i}\CodeLength{i}\;:\;\sum_{i=1}^{\AlphabetSize}2^{-\CodeLength{i}}\leq 1,\ \ \CodeLength{i}\in\Naturals\text{ for every }i\in[\AlphabetSize]\right\}.
\end{equation}
Hence, the optimum in Equation~\eqref{eq:huffman-length-formulation} is exactly $\MinimumEncodedLength(\Text)$.

The advantage of this formulation is that it describes the entire solution at once. The classical Huffman algorithm constructs the solution through a sequence of dependent merges, which makes a direct low-round \MPC implementation difficult. In contrast, Equation~\eqref{eq:huffman-length-formulation} asks for all codeword lengths simultaneously through a single optimization problem. 


Our algorithm exploits both the optimization view and the iterative view of
Huffman coding. The optimization view works particularly well when the ratio
between the frequencies is bounded. In this case, the leaves can appear only
within a short range of depths. The Kraft formulation then gives enough global
structure to determine many leaf depths in parallel, without simulating the
Huffman merges one by one.

For a general instance, however, the frequencies may differ by a factor
as large as \(\TextLength\), and the corresponding range of depths may
be large. The optimization view alone is not well suited to such an instance,
because a direct implementation would require too much local memory. The
iterative view alone is also insufficient: processing the entire frequency
sequence through dependent Huffman merges requires too many rounds. To keep
both the local memory and the number of rounds small, we process the sorted
frequencies in batches in which the frequency ratio is bounded. We use the optimization view to process each batch and the iterative
view to connect its solution to the remainder of the instance. These two perspectives yield the following two subtasks.

\begin{itemize}
	\item \textbf{Bounded-Ratio Huffman Coding.}
Bounded-Ratio Huffman Coding takes as input an ordered multiset of frequencies $\Frequency{1} \le \cdots \le \Frequency{m}$, and solves the instance under the condition that the ratio of the largest frequency to the second-smallest is bounded by $\RatioBound$, i.e.,
$
{\Frequency{m}}/{\Frequency{2}} \le \RatioBound.
$
($\lambda$ appears as a parameter of the algorithm's complexity). The subtask outputs the optimal encoded length and an optimal depth assignment $\LeafDepth{1}, \ldots, \LeafDepth{m}$, where $\LeafDepth{i}$ is the depth of the leaf for $\Frequency{i}$ in the optimal Huffman tree, equivalently the length of its binary codeword.

 \item \textbf{Iterative-State Reconstruction.}
The Iterative-State Reconstruction subtask takes a solution which is represented by heights $h_1, \ldots, h_m$ and translates such a solution into an arbitrary state of the \textcolor{black}{classical} Huffman coding algorithm without ever running that algorithm. More precisely, it takes as input a bounded-ratio instance, together with the solution produced for it by the previous subtask, and a threshold $\tau$. Its purpose is to determine the state of the \textcolor{black}{classical} iterative Huffman coding algorithm at the moment just before a frequency exceeding $\tau$ would first become involved in the merging process, and it does so crucially without ever running the \textcolor{black}{classical} algorithm itself. Concretely, the input consists of the sorted frequencies $\Frequency{1} \leq \cdots \leq \Frequency{m}$, the optimal depth assignments $\LeafDepth{1}, \ldots, \LeafDepth{m}$ returned by the Bounded-Ratio Huffman Coding algorithm, and a threshold $\tau > \Frequency{m}$.

This subtask connects the two views of the problem: from \textcolor{black}{the frequencies and depth assignments}, and with no simulation of the iterative Huffman coding algorithm, it reconstructs the exact state that the \textcolor{black}{classical} iterative Huffman algorithm would have reached had it actually been run on this instance and terminated right before a merge involves a frequency more than $\tau$ (see Figure~\ref{fig:prefix-reuse-example} for an illustration).
\end{itemize}

\input{Sections/translation}

Suppose for now that Bounded-Ratio Huffman Coding and
Iterative-State Reconstruction can both be solved efficiently. We
show how to combine them to solve the Huffman coding
problem. Fix a parameter $\lambda$, whose value will be chosen later. \textcolor{black}{The algorithm maintains a sorted list $B$ of frequencies and an accumulated cost $\AccumulatedMergeCost$, initially equal to the input list and zero, respectively. Throughout the algorithm, $\AccumulatedMergeCost$ plus the optimum encoded length for $B$ equals the optimum encoded length for the original input. Each iteration performs valid Huffman merges, adds their cost to $\AccumulatedMergeCost$, and replaces the merged frequencies by their parents.}

 In one iteration, let
$\Frequency{1}\leq\cdots\leq\Frequency{m}$ be the current list of frequencies, and let
\(j\) be the largest index such that
$
{\Frequency{j}}/{\Frequency{2}} \le \lambda.
$
Thus, the prefix
$
\Frequency{1},\ldots,\Frequency{j}
$
satisfies the ratio condition required by Bounded-Ratio Huffman Coding. Therefore, we construct a smaller instance with only these frequencies and solve them with our algorithm for Bounded-Ratio Huffman Coding. If \(j=m\) (all frequencies are processed), we add its optimal encoded length to the
running total, delete all frequencies and finish the algorithm. Otherwise, we invoke Iterative-State Reconstruction with the
prefix frequencies $\Frequency{1},\ldots,\Frequency{j}$, their corresponding
leaf depths $\LeafDepth{1},\ldots,\LeafDepth{j}$ returned by Bounded-Ratio
Huffman Coding algorithm, and $\tau=\Frequency{j+1}$ as its threshold. Iterative-State Reconstruction processes the input and without running the iterative algorithm, directly jumps to the state of the iterative Huffman coding
algorithm  just before any frequency larger than \(\tau\) would participate in a merge. Having that state available, we add the
cost of all completed merges to the running total and  then replace
$\Frequency{1},\ldots,\Frequency{j}$ by the frequencies present in the state of the iterative Huffman coding algorithm \textcolor{black}{at the reconstruction stopping point}. After this change, we have a new set of frequencies along with an updated running total and we repeat the same procedure on them consecutively until all frequencies are processed at which point our running total cost is the final answer.

The correctness of this algorithm follows from the iterative nature of the \textcolor{black}{classical} Huffman coding algorithm. It remains to bound the number of iterations. In each iteration, by the maximality of $j$ we have
$\tau = f_{j+1}>\lambda\Frequency{2}$. \textcolor{black}{Iterative-State Reconstruction leaves at most one
frequency smaller than $\tau$.}\footnote{The smallest returned frequency may still be smaller than $\tau$. Thus, we can guarantee progress only for the second-smallest frequency, which is why the ratio condition uses $\Frequency{2}$ rather than $\Frequency{1}$; see Section~\ref{sec:putting-pieces-together} for further details.}
\textcolor{black}{Since every frequency in the untouched suffix is
at least $\tau$, the second-smallest frequency in the updated list is
at least $\tau$} and hence is more than $\lambda$ times the second-smallest
frequency in the previous list. Since every  frequency is bounded by
$\TextLength$, our algorithm terminates after $O(\log_\lambda\TextLength)$ iterations. We later choose
$\lambda=\TextLength^{\Theta(\epsilon)}$ for a fixed constant
$\epsilon>0$, which only adds a multiplicative factor of $O_\epsilon(1)$ to the round complexity of the algorithm. We next describe how these two subtasks can be efficiently solved.

\input{Sections/bounded-ratio}
\input{Sections/kraft-budget-distribution}
\input{Sections/threshold-reconstruction-rewrite}

\subsection{{\color{black}Application to the PRAM Setting}}
\label{sec:pram-overview}
{\color{black}Although our framework is tailored to the \MPC model, we also show in Section~\ref{sec:pram} that it leads to \PRAM algorithms for Huffman coding that go beyond the currently known algorithms. The table below compares our results with previous \PRAM algorithms; see Table~\ref{tab:pram-comparison} for details.}

{\color{black}Unlike most previous work, which aims for NC algorithms, our goal is to minimize the parallel running time while keeping the total work close to linear. Although our algorithm is not an NC algorithm in general, its running time is asymptotically better than that of all previous algorithms as long as $n\leq m^{O(\log m/(\log\log m)^2)}$; in this range, it runs in $O_\epsilon(\log m/\log\log m)$ time. This is optimal in this setting, due to our matching $\Omega(\log m/\log\log m)$ lower bound (Lemma~\ref{lem:pram-lower-bound}).}

%% file: Sections/translation.tex
\begin{figure}[t]
	\centering
	\begin{tikzpicture}[
		>=Stealth,
		leaf/.style={draw, rectangle, rounded corners=1pt, minimum size=6mm, inner sep=1pt, font=\small},
		inode/.style={draw, circle, minimum size=7mm, inner sep=0pt, font=\small},
		kept/.style={thick},
		gone/.style={draw=black, text=black, dashed},
		root/.style={fill=black!12, draw=black, very thick},
		panel/.style={draw=black, rounded corners=3pt, inner sep=7pt},
		ptitle/.style={font=\small\bfseries, anchor=south west, inner sep=2pt},
		]
		\begin{scope}[local bounding box=IN]
			\node[font=\small, anchor=east] at (0,0)    {$i$};
			\node[font=\small, anchor=east] at (0,-0.7) {$\Frequency{i}$};
			\node[font=\small, anchor=east] at (0,-1.4) {$\LeafDepth{i}$};
			\foreach \f/\h [count=\i] in {2/4,3/4,4/3,6/3,7/3,9/2,10/2} {
				\node[font=\scriptsize, text=black] at (0.75*\i-0.25,0) {\i};
				\node[leaf] at (0.75*\i-0.25,-0.7) {$\f$};
				\node[font=\small] at (0.75*\i-0.25,-1.4) {$\h$};
			}
			\node[font=\small, anchor=west] at (-0.6,-2.2) {threshold $\tau = 11$};
		\end{scope}
		\node[panel, fit=(IN)] (INP) {};
		\node[ptitle] at (INP.north west) {Input};
		
		\begin{scope}[shift={(9.3,0)}, local bounding box=OUT]
			\node[font=\small, anchor=west, align=left] at (0,-0.1)
			{State right before the first merge\\involving a frequency $> \tau$:};
			\node[leaf,root] at (0.5,-1.15) {$10$};
			\node[leaf,root] at (1.4,-1.15) {$13$};
			\node[leaf,root] at (2.3,-1.15) {$18$};
			\node[font=\small, anchor=west, align=left] at (0,-2.35)
			{Cost of merges performed:\\[1pt]
				$\underbrace{5+9+18}_{\text{part } 18} + \underbrace{0}_{\text{part } 10} + \underbrace{13}_{\text{part } 13} = 45$};
		\end{scope}
		\node[panel, fit=(OUT)] (OUTP) {};
		\node[ptitle] at (OUTP.north west) {Output};
		
		\draw[->, very thick] (INP.east) -- node[above, font=\scriptsize, align=center]
		{Iterative-State\\Reconstruction} (OUTP.west |- INP.east);

	\end{tikzpicture}
	\caption{Iterative-State Reconstruction on a bounded-ratio instance. Given the sorted frequencies $\Frequency{1} \leq \cdots \leq \Frequency{m}$, their optimal depths $\LeafDepth{1}, \ldots, \LeafDepth{m}$, and a threshold $\tau > \Frequency{m}$, the subtask outputs the state of the \textcolor{black}{classical} Huffman algorithm right before a merge first involves a frequency greater than $\tau$. Here this state is the multiset $\{10, 13, 18\}$, reached after the merges $2{+}3$, $4{+}5$, $6{+}7$, and $9{+}9$, at a total cost of $45$. }
	\label{fig:prefix-reuse-example}
\end{figure}
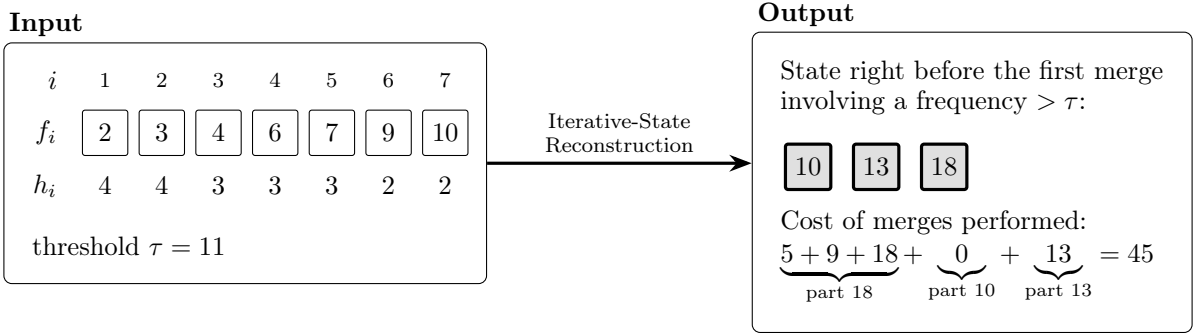

%% file: Sections/bounded-ratio.tex
\subsection{Bounded-Ratio Huffman Coding}
\label{sec:bounded-ratio}
We begin with the first subtask, which captures the bounded-ratio case of the Huffman coding problem. Given sorted frequencies
$\Frequency{1}\leq\cdots\leq\Frequency{m}$ satisfying
$\Frequency{m}/\Frequency{2}\leq\RatioBound$, our goal is to compute the optimal encoded length and a corresponding optimal depth assignment
$\LeafDepth{1}\geq\LeafDepth{2}\geq\cdots\geq\LeafDepth{m}$.

Let us explain how the bounded-ratio assumption leads to an efficient \MPC algorithm. The crux of the argument is that this bounded ratio bounds the range of leaf depths in an optimal Huffman tree. Intuitively, follow a path from a leaf toward the root. The \textcolor{black}{frequencies} along this path grow at least as fast as the Fibonacci sequence. Thus, a large difference between two leaf depths forces a large ratio between their frequencies. Accordingly, let \(\GoldenRatio\) be the golden ratio and let
\(\DepthSpread{\RatioBound} = \lceil \log_{\GoldenRatio}\RatioBound \rceil + 1\).
We show in
Section~\ref{sec:bounded-ratio-formal}
(Corollary~\ref{cor:bounded-ratio-depth-range}) that if
$\Frequency{m}/\Frequency{2}\leq\RatioBound$, there exists an optimal
depth assignment whose depths all lie in the range \([\ell,u]\), where
\[
\ell = \textcolor{black}{\max\{1,\lceil\log_2 m\rceil - \DepthSpread{\RatioBound}\}},
\qquad
u = \lfloor\log_2 m\rfloor + \DepthSpread{\RatioBound}.
\]
We use these two bounds throughout our discussion. Notice
that $u-\ell\leq 2\DepthSpread{\RatioBound}$.

Recall the Kraft--McMillan constraint in
Equation~\eqref{eq:huffman-length-formulation}. Since $u$ is an upper bound on the maximum depth, multiplying the Kraft inequality by $2^u$ gives $\sum_{i=1}^{m}2^{u-\LeafDepth{i}}\leq 2^u$.
This form has a simple interpretation. The right-hand side is a Kraft
budget of $2^u$, and a frequency $\Frequency{i}$ with depth
$\LeafDepth{i}$ uses $2^{u-\LeafDepth{i}}$ units of that budget. We seek
depths whose total usage does not exceed the budget. Since every depth is
at most $u$, each budget contribution is an integer.

Based on this interpretation, we define a more general version of the
Huffman coding problem, which we call the budget-constrained Huffman coding
problem. In this version, we allow an arbitrary total budget rather than
$2^u$. More precisely, for any submultiset $A$ of the input frequencies
$f_1,\ldots,f_m$ and any integer budget
$b\in\{0,1,\ldots,2^u\}$, define $\mu(A,b)$ as the optimum of
\begin{equation}
	\label{eq:budget-constrained-huffman}
	\mu(A,b)=\min\left\{\sum_{f_i\in A} f_i h_i\;:\;\sum_{f_i\in A}2^{u-h_i}\leq b,\ \ \ell\leq h_i\leq u\text{ for every }f_i\in A\right\},
\end{equation}
where the depths $h_i$ are positive integers.
If no feasible assignment exists, we set
$\mu(A,b)=+\infty$.
Note that when $A$ contains all input frequencies, $\mu(A,2^u)$ is
the optimal encoded length of the original instance.

The advantage of this formulation is that it enables a divide-and-conquer
approach: we split the frequencies into smaller subproblems, compute their
optimum values independently, and then combine these values.
For disjoint frequency multisets $A_1$ and $A_2$, let
$A=A_1\uplus A_2$. For any nonnegative integer budgets $b_1,b_2$ with
$b_1+b_2=b$, we have
$
\mu(A,b)\leq \mu(A_1,b_1)+\mu(A_2,b_2).
$
Conversely, an optimal solution for $A$ distributes its budget between
$A_1$ and $A_2$. \footnote{If the solution does not use the entire budget $b$,
we may distribute the unused units arbitrarily between $b_1$ and $b_2$;
the two assignments remain feasible because their budget constraints are
upper bounds.} Therefore
\begin{equation}
	\label{eq:kraft-budget-recurrence}
\mu(A,b)=
\min_{\substack{b_1+b_2=b\\b_1,b_2\geq0}}
\bigl\{\mu(A_1,b_1)+\mu(A_2,b_2)\bigr\}.
\end{equation}
Equation~\eqref{eq:kraft-budget-recurrence} suggests the following
divide-and-conquer strategy. We partition the sorted frequency sequence
into consecutive blocks such that the frequency
ratio within each block is bounded by a constant. This property allows
us to compute the budget-constrained table values for each block efficiently
and in parallel. We then place the blocks at the leaves of a balanced
hierarchy and use Equation~\eqref{eq:kraft-budget-recurrence} to combine
neighboring subproblems level by level. All merges at the same level are
processed in parallel, and the root gives the optimal solution for the
entire frequency sequence.

Figure~\ref{fig:representative-blocks} illustrates our method for partitioning
the frequencies. We choose a set of representative frequencies as follows. Set
$\Representative{1}=1$, so the first representative frequency is $\Frequency{1}$. After choosing
$\Representative{j}$, define $\Representative{j+1}$ to be the smallest
index such that
$\Frequency{\Representative{j+1}}>
\GoldenRatio\Frequency{\Representative{j}}$, provided such an index
exists. Let $\Representative{1},\ldots,\Representative{q}$ be the
resulting representative indices. We append the dummy index
$\Representative{q+1}=m+1$, which serves only as an end marker, and for
every $j\in[q]$ define
$
\FrequencyBlock{j}
=
(f_{\Representative{j}},f_{\Representative{j}+1},
\ldots,f_{\Representative{j+1}-1}).
$
Since
$\Frequency{m}/\Frequency{2}\leq\RatioBound$, the number of blocks is at
most $2+\lceil\log_{\GoldenRatio}\RatioBound\rceil$.\footnote{Except
for the first representative, all representatives lie between
$\Frequency{2}$ and $\Frequency{m}$ and increase by a factor greater
than $\GoldenRatio$. The additional block accounts for the possible
gap between $\Frequency{1}$ and $\Frequency{2}$.}

\input{Sections/representative-blocks-figure}

The key point is that the depths within each block have constant range. By
construction, the frequencies within a block differ by a factor of at most
$\GoldenRatio$. Consequently, in an optimal solution to the subproblem associated with any block, regardless of the Kraft budget, the depths of the frequencies in that block differ by at most two. This
restriction, together with prefix sums of the frequencies,
lets us compute all values $\mu(B_j,b)$ independently in $O(1)$ \MPC{}
rounds. The implementation and memory analysis appear in
Section~\ref{sec:bounded-ratio-formal}.

We combine the blocks using a bottom-up hierarchy whose leaves are
$B_1,\ldots,B_q$. At each level, we pair neighboring nodes and compute
$\mu(A,b)$ for $A=A_1\uplus A_2$ using
\textcolor{black}{Equation~\eqref{eq:kraft-budget-recurrence}}. If the number of nodes is odd,
we carry the last node unchanged. We repeat until only one node remains.
Figure~\ref{fig:representative-blocks} also illustrates this construction.
The hierarchy has
$O(\log(2+\log\RatioBound))$ levels, and all nodes at the same level are
processed in parallel. At the root,
$A=B_1\uplus\cdots\uplus B_q$, and $\mu(A,2^u)$ is the
optimal encoded length.

The procedure above computes the optimum
value of $\mu(T)$. We can also recover a corresponding optimal depth assignment by
storing the budget split attaining each table entry and tracing these
choices from the root to the leaves.  This part is explained in greater detail in
Section~\ref{sec:bounded-ratio-formal}.

It remains to combine two neighboring subproblems efficiently. Consider
a merge \(A=A_1\uplus A_2\). For a frequency multiset \(A_i\), only
\(O(|A_i|2^{u-\ell})\) budget values are relevant, since each frequency
uses between \(1\) and \(2^{u-\ell}\) Kraft units. A direct
computation of the parent table tries every budget split
$b=b_1+b_2$ and evaluates
$\mu(A_1,b_1)+\mu(A_2,b_2)$. Over all parent budgets, this examines every
pair of entries from the two child tables and requires
$O(|A_1||A_2|2^{2(u-\ell)})$
candidate evaluations. Since \(u-\ell\leq
2\DepthSpread{\RatioBound}\), this can be as large as
$
O\!\left(
|A_1||A_2|2^{4\DepthSpread{\RatioBound}} 
\right) = O\!\left(
|A_1||A_2|\lambda^{5.77}
\right).
$
This exhaustive computation is too expensive, even before accounting
for its implementation in the \MPC model. In
Section~\ref{sec:kraft-budget-distribution}, we show how to perform each
merge efficiently without examining every pair of child entries.

%% file: Sections/representative-blocks-figure.tex
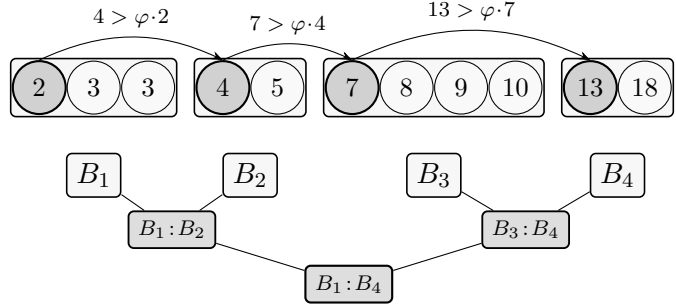
\begin{figure}[t]
  \centering
  \begin{tikzpicture}[
      x=0.9cm,
      y=0.8cm,
      frequency/.style={circle,draw=black,fill=gray!5,
        minimum size=7mm,inner sep=0pt,font=\small},
      representative/.style={frequency,draw=black,
        fill=black!18,line width=0.8pt},
      block/.style={rounded corners=2pt,draw=black,fill=black!3,
        line width=0.6pt},
      merged/.style={rounded corners=2pt,draw=black,
        fill=black!13,line width=0.7pt,font=\scriptsize},
      root/.style={rounded corners=2pt,draw=black,
        fill=black!13,line width=0.8pt,font=\scriptsize},
      edge/.style={draw=black,thin},
      threshold/.style={-{Stealth[length=1.5mm]},draw=black,thin}
    ]
    \draw[block] (-0.42,-0.48) rectangle (2.02,0.48);
    \draw[block] (2.28,-0.48) rectangle (3.92,0.48);
    \draw[block] (4.18,-0.48) rectangle (7.42,0.48);
    \draw[block] (7.68,-0.48) rectangle (9.32,0.48);

    \node[representative] (f2)  at (0,0)   {$2$};
    \node[frequency]      (f3a) at (0.8,0) {$3$};
    \node[frequency]      (f3b) at (1.6,0) {$3$};

    \node[representative] (f4)  at (2.7,0) {$4$};
    \node[frequency]      (f5)  at (3.5,0) {$5$};

    \node[representative] (f7)  at (4.6,0) {$7$};
    \node[frequency]      (f8)  at (5.4,0) {$8$};
    \node[frequency]      (f9)  at (6.2,0) {$9$};
    \node[frequency]      (f10) at (7.0,0) {$10$};

    \node[representative] (f13) at (8.1,0) {$13$};
    \node[frequency]      (f18) at (8.9,0) {$18$};
    \draw[threshold] (f2.north) to[bend left=25]
      node[above,font=\scriptsize] {$4>\GoldenRatio\!\cdot\!2$} (f4.north);
    \draw[threshold] (f4.north) to[bend left=25]
      node[above,font=\scriptsize] {$7>\GoldenRatio\!\cdot\!4$} (f7.north);
    \draw[threshold] (f7.north) to[bend left=25]
      node[above,font=\scriptsize] {$13>\GoldenRatio\!\cdot\!7$} (f13.north);

    \node[block]  (b1)  at (0.8,-1.45) {$B_1$};
    \node[block]  (b2)  at (3.1,-1.45) {$B_2$};
    \node[block]  (b3)  at (5.8,-1.45) {$B_3$};
    \node[block]  (b4)  at (8.5,-1.45) {$B_4$};
    \node[merged] (b12) at (1.95,-2.35) {$B_1\!:\!B_2$};
    \node[merged] (b34) at (7.15,-2.35) {$B_3\!:\!B_4$};
    \node[root]   (b14) at (4.55,-3.25) {$B_1\!:\!B_4$};
    \draw[edge] (b1) -- (b12);
    \draw[edge] (b2) -- (b12);
    \draw[edge] (b3) -- (b34);
    \draw[edge] (b4) -- (b34);
    \draw[edge] (b12) -- (b14);
    \draw[edge] (b34) -- (b14);

  \end{tikzpicture}
  \caption{Frequency blocks and the bottom-up hierarchy used to combine
  them. }
  \label{fig:representative-blocks}
\end{figure}

%% file: Sections/kraft-budget-distribution.tex
\subsection{Fast Merging by Layered Convexity}
\label{sec:kraft-budget-distribution}
\par\smallskip

The recurrence in
Equation~\eqref{eq:kraft-budget-recurrence} is simple, but evaluating it
directly is expensive. For a merge $A=A_1\uplus A_2$ and a fixed parent
budget $b$, direct evaluation loops over every split $b_1+b_2=b$. The number
of such divisions can be
$O(\min\{|A_1|,|A_2|\}2^{u-\ell})$. Our goal is to find the optimal
split without enumerating all feasible divisions. We formalize this task
in Problem~\ref{prob:kraft-budget-merging}.

\begin{problem}
	\label{prob:kraft-budget-merging}
	Let $A_1$ and $A_2$ be two disjoint frequency multisets,
	and suppose that the tables $\mu(A_1,\cdot)$ and $\mu(A_2,\cdot)$ are
	known. For $A=A_1\uplus A_2$, compute $\mu(A,b)$ for every feasible
	budget $b$, together with an optimal division of $b$ between $A_1$
	and $A_2$.
\end{problem}

By the recurrence relation in
Equation~\eqref{eq:kraft-budget-recurrence}, for a fixed  budget
$b$ we must minimize
\[
G_b(x)=\mu(A_1,x)+\mu(A_2,b-x)
\]
over the integers $x\in[0,b]$.
To avoid examining the budget divisions one by one, we seek a
well-behaved structure in $G_b$ that lets us find its minimum quickly.
Convexity would have given such a structure: if $G_b$ were convex, a binary
search comparing neighboring values would find its minimum with a
logarithmic number of evaluations. However, $G_b$ is not convex.
Consider $A_1=A_2=\{1\}$, with $\ell=2$, maximum depth $u=5$, and allowed depths
$2,3,4,5$. For either subproblem, these depths have budget contributions
$8,4,2,1$ and incur costs $2,3,4,5$, respectively. For total budget
$b=10$, Figure~\ref{fig:nonconvex-merge-example} gives the two
subproblem costs and the resulting values of $G_{10}$. In particular,
$G_{10}(3)$ and $G_{10}(7)$ are larger than both of their neighboring
values, so $G_{10}$ is not convex.

\FloatBarrier
\input{Sections/nonconvex-merge-example-figure}

\paragraph{Layered convexity.}
Although $G_b$ is not convex, we show that it satisfies a weaker
property, which we call \emph{layered convexity}.
We say that $G_b$ is layered convex modulo $d$ if, for
every residue $r$ modulo $d$, the sequence
\[
G_b(r),G_b(r+d),G_b(r+2d),\ldots
\]
over its feasible arguments is convex. Thus, a layered-convex function
need not be convex on consecutive elements; convexity is guaranteed
only within each residue class.
Our main structural claim is that, for every instance of
Problem~\ref{prob:kraft-budget-merging} and every feasible parent budget
$b$, the function $G_b$ is layered convex modulo $d=2^{u-\ell-1}$.

Before discussing the general structure, let us see the phenomenon via
an example. Figure~\ref{fig:residue-convex-profile} shows the profile
for two subproblems, each containing three frequencies equal to $2$,
with $\ell=3$, $u=6$, and total budget $b=20$. As
Figure~\ref{fig:residue-convex-profile}(a) shows, the interleaved
profile has several local spikes and is therefore not convex.

For this instance, $d=2^{u-\ell-1}=4$. Grouping the budget divisions
by their residue modulo $4$ separates the profile into four convex
layers. For example, the layer with residue $0$ consists of
$x=4,8,12,16$, where the values of $G_b(x)$ are $56,52,52,56$.
Its successive differences are $-4,0,4$, which are nondecreasing;
hence, this layer is convex. The other three layers behave similarly,
as shown in Figure~\ref{fig:residue-convex-profile}(b).
\input{Sections/residue-convex-profile-figure}

Layered convexity would let us evaluate the recurrence much faster: 
 decompose the feasible divisions $x$ into their residue classes modulo $d$.
Within each class, the values of $G_b(x)$ form a convex sequence, so a binary
search comparing neighboring values finds the minimum in that class. The
different classes can be searched independently and in parallel, after which
we return the smallest of their minima. 

To give the intuition behind layered convexity, we
interpret the depth-assignment problem as a knapsack problem whose items
represent incremental upgrades. Fix the upper depth
bound $u$, and initially place every frequency at depth $u$. Thus, each
frequency uses one budget unit, so the total initial budget usage is
$|A|$. The corresponding baseline cost is
$u\sum_{f_i\in A}f_i$. For every feasible budget $b$, we necessarily
have $b\geq |A|$.

Shortening the codeword of frequency $f_i$ from depth $u$ to $u-s$
increases its Kraft contribution from $1$ to $2^s$. The additional
budget consumed is therefore
$
2^s-1=1+2+4+\cdots+2^{s-1},
$
while the cost decreases by $sf_i$. We view this as taking $s$
successive \emph{upgrades}, of widths
$1,2,4,\ldots,2^{s-1}$, each of which saves $f_i$ in cost.

This interpretation gives a knapsack problem with capacity $b-|A|$.
For each frequency $f_i\in A$, we create $u-\ell$ items of widths
$1,2,4,\ldots,2^{u-\ell-1}$, each having value $f_i$, and seek a
maximum-value set of total width at most $b-|A|$. A valid depth
assignment selects a prefix of the upgrades associated with each
frequency. Conversely, suppose that a knapsack solution selects $t$
upgrades associated with $f_i$. We can replace them by the $t$ smallest
upgrades of that frequency. This replacement preserves the total value,
uses no more capacity, and produces a valid prefix. Applying this
replacement to every frequency transforms any knapsack solution into a
valid depth assignment. Consequently, $\mu(A,b)$ equals the baseline
cost $u\sum_{f_i\in A}f_i$ minus the optimum value of this knapsack
instance.

This knapsack reformulation is the key ingredient in our layered-convexity argument. To understand the intuition, recall that
$d=2^{u-\ell-1}$ is the largest item width. We may add zero-value items
of width $1$, so an optimal solution can use the entire budget without
changing its value. Fix a residue $r$ modulo $d$ and first consider the
width-$1$ items. If $r$ is odd, the solution must contain an odd number
of these items, because every larger item has even width. We may assume
that the solution contains the most valuable width-$1$ item: replacing
any selected width-$1$ item by this item can only improve the solution.
After fixing this item, the solution contains an even number of the
remaining width-$1$ items. We pair the two most valuable remaining
items, then the next two, and so on. Choosing two more width-$1$ items
now corresponds to choosing the first pair; choosing four more
corresponds to choosing the first two pairs. We can therefore regard
each pair as one width-$2$ item whose value is the sum of the two
values. If $r$ is even, no width-$1$ item is forced, and we pair all
width-$1$ items in the same way.

At width $2$, we now have two kinds of items: the items that originally
had width $2$ and the new items obtained by merging pairs of width-$1$
items. We consider both kinds together. After accounting for the
width-$1$ choice, the remaining residue modulo $4$ is either $0$ or
$2$. If it is $2$, the solution must contain an odd number of width-$2$
items, so we may fix the most valuable one. The solution then contains
an even number of the remaining width-$2$ items, which we pair to form
width-$4$ items. If the residue is $0$, we pair the width-$2$ items
without fixing one first. We repeat the same argument for widths
$4,8,\ldots,d/2$.

After these steps, the residue $r$ has fixed all choices involving
widths smaller than $d$. Only width-$d$ items remain. For a budget of
the form $r+td$, the solution takes the fixed smaller-width choices and
then the $t$ most valuable width-$d$ items. Increasing the budget by
$d$ adds the next most valuable item. The value gained at each step can
only decrease, which means the
optimal knapsack values at capacities $r,r+d,r+2d,\ldots$ form a concave
sequence.
Standard properties of convex sequences then imply that
$
G_b(r),G_b(r+d),G_b(r+2d),\ldots
$
is convex. Since this holds for every residue $r$ modulo $d$, the
function $G_b(\cdot)$ is layered convex.
 We prove the structural claim formally in
Section~\ref{sec:layered-merging-formal}. We can therefore
search every layer independently and return the smallest layer minimum.

\paragraph{\MPC implementation.}
For each parent budget $b$, we search all convex layers in parallel.
A standard binary search would examine one position at a time and would
therefore require logarithmically many search steps. Instead, in each step
we probe $\eta\geq 2$ evenly spaced positions in every layer. Convexity
identifies the subinterval containing the layer minimum, reducing the
remaining search range by a factor of $\Theta(\eta)$. We repeat this procedure
until the minimum of every layer is found, and then return the smallest
of these minima. All parent budgets, layers, and probes are evaluated
in parallel.

Each layer contains
$O(\min\{|A_1|,|A_2|\})$ feasible divisions. The search therefore
finishes in
$
O(1+\log_\eta\min\{|A_1|,|A_2|\})
$
\textcolor{black}{search steps}.
The maximum item width, and hence the number of layers, is
$d=2^{u-\ell-1}=O(\RatioBound^{2\log_{\GoldenRatio}2})$. Thus, the probes
for one parent budget require
$
O\!\left(
\eta\RatioBound^{2\log_{\GoldenRatio}2}
\right)
$
local memory. The required child-table entries are distributed across
the machines and can be routed to the corresponding probes using
standard \MPC primitives. \textcolor{black}{Their cost is
\(O(1+\log_\eta(|A_1|+|A_2|))\) rounds per search step, as
accounted for in Lemma~\ref{lem:fast-table-merge}.}
Consequently, for every $\eta\geq2$, a merge of
$A_1$ and $A_2$ takes
\textcolor{black}{\(O((1+\log_\eta(|A_1|+|A_2|))
(1+\log_\eta\min\{|A_1|,|A_2|\}))\)} \MPC{} rounds, uses
$O(\eta\RatioBound^{2\log_{\GoldenRatio}2})$ local memory per machine,
and requires
$\softO((|A_1|+|A_2|)\eta
\RatioBound^{4\log_{\GoldenRatio}2})$ total memory and computation.

Combining this merging procedure with the bottom-up hierarchy in
Section~\ref{sec:bounded-ratio} gives the complexity of Bounded-Ratio
Huffman Coding. At every level of the hierarchy, the
frequency multisets are disjoint and their sizes sum to $m$. Thus,
summing the cost of all merges at one level replaces
$|A_1|+|A_2|$ by $m$.
\begin{lemma}
	\label{lem:bounded-ratio-huffman}
	For every $\eta\geq2$, Bounded-Ratio Huffman Coding on $m$
	frequencies can be solved in
	\textcolor{black}{$O((1+\log_\eta m)^2\log(2+\log\RatioBound))$} \MPC{} rounds using
	$O(\eta\RatioBound^{2\log_{\GoldenRatio}2})$ local memory per machine.
	Its total memory and total computation are
	$\softO(m\eta\RatioBound^{4\log_{\GoldenRatio}2})$.
\end{lemma}

%% file: Sections/nonconvex-merge-example-figure.tex
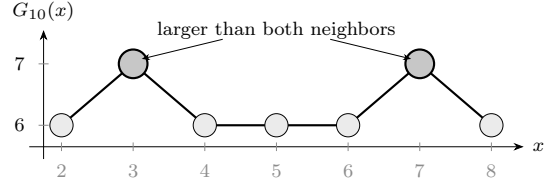
\begin{figure}[t]
	\centering
	\begin{minipage}[c]{0.52\linewidth}
		\centering
		\setlength{\tabcolsep}{2.5pt}
		\renewcommand{\arraystretch}{1.15}
		\scriptsize
		\begin{tabular}{c|ccccccc}
			$x$ & $2$ & $3$ & $4$ & $5$ & $6$ & $7$ & $8$\\ \hline
			$\mu(A_1,x)$ & $4$ & $4$ & $3$ & $3$ & $3$ & $3$ & $2$\\
			$\mu(A_2,10-x)$ & $2$ & $3$ & $3$ & $3$ & $3$ & $4$ & $4$\\ \hline
			$G_{10}(x)$ & $6$ & $7$ & $6$ & $6$ & $6$ & $7$ & $6$
		\end{tabular}
	\end{minipage}\hfill
	\begin{minipage}[c]{0.44\linewidth}
		\centering
		\resizebox{\linewidth}{!}{%
		\begin{tikzpicture}[
			x=1.05cm,
			y=0.9cm,
			axis/.style={-{Stealth[length=1.5mm]},draw=black,thin},
			profile/.style={draw=black,line width=0.9pt},
			value/.style={circle,draw=black,fill=black!8,
				minimum size=3.4mm,inner sep=0pt},
			peak/.style={circle,draw=black,fill=black!22,
				minimum size=4.2mm,inner sep=0pt,line width=0.9pt}
		]
			\draw[axis] (1.55,4.65) -- (8.45,4.65)
				node[right,font=\scriptsize] {$x$};
			\draw[axis] (1.75,4.45) -- (1.75,6.55)
				node[above,font=\scriptsize] {$G_{10}(x)$};
			\draw[profile]
				(2,5)--(3,6)--(4,5)--(5,5)--(6,5)--(7,6)--(8,5);
			\foreach \x/\y in {2/5,4/5,5/5,6/5,8/5}
				\node[value] at (\x,\y) {};
			\foreach \x in {3,7}
				\node[peak] at (\x,6) {};
			\foreach \x in {2,3,4,5,6,7,8}
				\draw[black!45] (\x,4.72)--(\x,4.58)
					node[below=2pt,font=\scriptsize] {$\x$};
			\foreach \y/\label in {5/6,6/7}
				\draw[black!45] (1.82,\y)--(1.68,\y)
					node[left=2pt,font=\scriptsize] {$\textcolor{black}{\label}$};
			\node[font=\scriptsize,text=black] at (5,6.55)
				{larger than both neighbors};
			\draw[-{Stealth[length=1.2mm]},draw=black]
				(4.25,6.40)--(3.12,6.08);
			\draw[-{Stealth[length=1.2mm]},draw=black]
				(5.75,6.40)--(6.88,6.08);
		\end{tikzpicture}%
		}
	\end{minipage}
	\caption{Budget values and the resulting nonconvex merge profile for
		two one-frequency multisets.}
	\label{fig:nonconvex-merge-example}
\end{figure}

%% file: Sections/residue-convex-profile-figure.tex
\begin{figure}[t]
  \centering
  \begin{tikzpicture}[
      yscale=0.9,
      axis/.style={-{Stealth[length=1.5mm]},draw=black,thin},
      profile/.style={draw=black,line width=0.9pt},
      point/.style={circle,draw=black,fill=black!8,
        minimum size=2.5mm,inner sep=0pt},
      spike/.style={circle,draw=black,fill=black!22,
        line width=0.8pt,minimum size=3.5mm,inner sep=0pt},
      strand/.style={draw=black,line width=0.8pt},
      strandpoint/.style={circle,draw=black,fill=black!10,
        minimum size=2.5mm,inner sep=0pt}
    ]
    \begin{scope}
      \node[font=\small] at (2.38,4.05) {(a) Interleaved profile};
      \draw[axis] (-0.25,0.35) -- (5.12,0.35)
        node[right,font=\scriptsize] {$x$};
      \draw[axis] (0,0.20) -- (0,3.55)
        node[above,font=\scriptsize] {$G_{b}(x)$};

      \draw[profile] plot coordinates {
        (0,3.00) (0.34,2.20) (0.68,2.20) (1.02,1.40)
        (1.36,1.40) (1.70,0.60) (2.04,1.40) (2.38,0.60)
        (2.72,1.40) (3.06,0.60) (3.40,1.40) (3.74,1.40)
        (4.08,2.20) (4.42,2.20) (4.76,3.00)
      };
      \foreach \x/\y in {
        0/3.00,0.34/2.20,0.68/2.20,1.02/1.40,1.36/1.40,
        1.70/0.60,2.38/0.60,3.06/0.60,3.74/1.40,
        4.08/2.20,4.42/2.20,4.76/3.00
      }
        \node[point] at (\x,\y) {};
      \foreach \x/\y in {2.04/1.40,2.72/1.40,3.40/1.40}
        \node[spike] at (\x,\y) {};

      \foreach \x/\label in {0/3,1.02/6,2.04/9,3.06/12,4.08/15,4.76/17}
        \draw[black!45!black] (\x,0.41)--(\x,0.29)
          node[below=2pt,font=\tiny] {$\label$};
      \foreach \y/\label in {0.60/52,1.40/54,2.20/56,3.00/58}
        \draw[black!45!black] (0.07,\y)--(-0.07,\y)
          node[left=2pt,font=\tiny] {$\label$};
      \node[font=\scriptsize,text=black,align=center]
        at (2.36,2.30) {\textcolor{black}{convexity}\\\textcolor{black}{violations}};
      \draw[-{Stealth[length=1.2mm]},draw=black]
        (2.30,2.05)--(2.08,1.57);
    \end{scope}

    \begin{scope}[xshift=7.0cm]
      \node[font=\small] at (2.10,6.85)
        {(b) Four convex layers};

      \node[left,font=\scriptsize] at (-0.18,5.55) {$r=0$};
      \draw[strand] plot coordinates
        {(0,5.85) (1.40,5.15) (2.80,5.15) (4.20,5.85)};
      \foreach \x/\y in {0/5.85,1.40/5.15,2.80/5.15,4.20/5.85}
        \node[strandpoint] at (\x,\y) {};
      \node[font=\tiny] at (2.10,4.88) {differences $-4,0,4$};

      \node[left,font=\scriptsize] at (-0.18,4.05) {$r=1$};
      \draw[strand] plot coordinates
        {(0,4.35) (1.40,4.00) (2.80,4.00) (4.20,4.70)};
      \foreach \x/\y in {0/4.35,1.40/4.00,2.80/4.00,4.20/4.70}
        \node[strandpoint] at (\x,\y) {};
      \node[font=\tiny] at (2.10,3.70) {differences $-2,0,4$};

      \node[left,font=\scriptsize] at (-0.18,2.60) {$r=2$};
      \draw[strand] plot coordinates
        {(0,2.85) (2.10,2.50) (4.20,2.85)};
      \foreach \x/\y in {0/2.85,2.10/2.50,4.20/2.85}
        \node[strandpoint] at (\x,\y) {};
      \node[font=\tiny] at (2.10,2.22) {differences $-2,2$};

      \node[left,font=\scriptsize] at (-0.18,1.10) {$r=3$};
      \draw[strand] plot coordinates
        {(0,1.45) (1.40,0.75) (2.80,0.75) (4.20,1.10)};
      \foreach \x/\y in {0/1.45,1.40/0.75,2.80/0.75,4.20/1.10}
        \node[strandpoint] at (\x,\y) {};
      \node[font=\tiny] at (2.10,0.46) {differences $-4,0,2$};

    \end{scope}
  \end{tikzpicture}
  \caption{The split profile for two frequency multisets, each containing three frequencies
  equal to $2$, with $u=6$, $\ell=3$, and total budget $20$.
  \textcolor{black}{The marked points violate convexity in the interleaved profile.} Grouping divisions by their
  residue modulo $d=4$ produces four convex layers.}
  \label{fig:residue-convex-profile}
\end{figure}
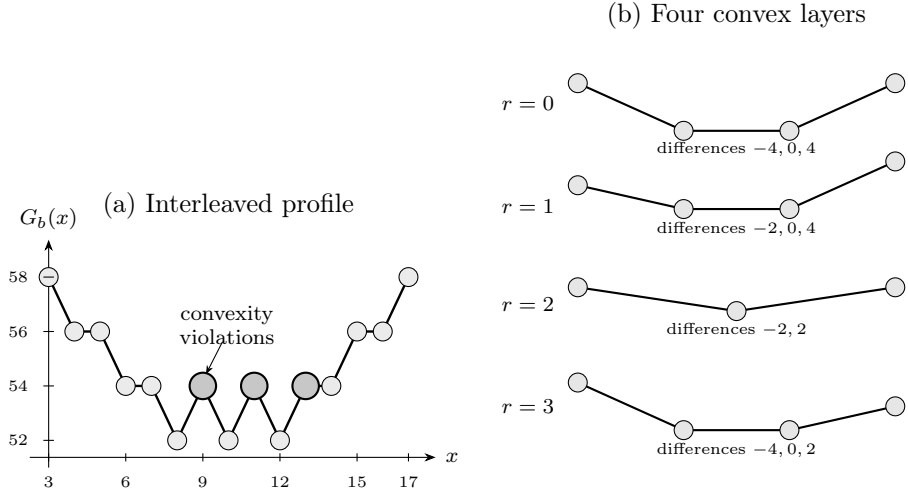

%% file: Sections/threshold-reconstruction-rewrite.tex
\subsection{Iterative-State Reconstruction}
\label{sec:threshold-reconstruction}

Iterative-State Reconstruction receives the pairs
$(\Frequency{1},\LeafDepth{1}),\ldots,
(\Frequency{m},\LeafDepth{m})$, where
$\Frequency{1}\leq\cdots\leq\Frequency{m}$ are the frequencies of the
bounded-ratio input instance, $\Frequency{m}/\Frequency{2}
\leq\RatioBound$, and $\LeafDepth{i}$ is the depth assigned to
$\Frequency{i}$ in the optimal solution returned by Bounded-Ratio Huffman Coding.
The input also contains a threshold $\tau>\Frequency{m}$. Our goal
is to simulate the \textcolor{black}{classical} iterative Huffman coding algorithm without running it and directly jump to its state 
immediately before the first merge involving the leaf of frequency more than $\tau$.

We write \(h_{\min}\) and \(h_{\max}\) for the minimum and maximum leaf depths. By the optimality of the solution, the Kraft--McMillan inequality in
Equation~\eqref{eq:huffman-length-formulation} must be tight, i.e.,
\[
\sum_{i=1}^{m} 2^{-\LeafDepth{i}}
=1,
\qquad\text{or equivalently}\qquad
\sum_{i=1}^{m} 2^{h_{\max}-\LeafDepth{i}}
=2^{h_{\max}}.
\]

Before presenting our algorithm, we give some intuition for what it means to recover the state of the \textcolor{black}{classical} Huffman algorithm immediately before the first merge that involves a frequency greater than $\tau$. Consider an optimal Huffman tree for the subproblem. Its nodes are of two kinds. The \emph{leaves} correspond one-to-one to the input frequencies. The \emph{internal nodes} correspond to the merges performed by the algorithm. We define the frequency of an internal node as the sum of the frequencies of the leaves in its subtree.

For the purpose of exposition, suppose that the optimal Huffman tree is known, and consider running the \textcolor{black}{classical} Huffman algorithm on the input. At each step, the algorithm merges two nodes into a new node whose frequency is the sum of theirs. For each such merge, the optimal Huffman tree contains exactly one node whose two children are the nodes being merged. Removing these two children from the tree makes that node a leaf. After the removal, the multiset of frequencies maintained by the algorithm equals the multiset of leaf frequencies of the modified tree. If we update the tree in the same way after every merge, this correspondence holds throughout the execution. That is, at every step there is a bijection between the frequencies held by the Huffman algorithm and the leaf frequencies of the current tree.

This correspondence characterizes the output we seek. Given the optimal Huffman tree, we remove every node whose parent has no child of frequency greater than $\tau$. The resulting tree represents exactly the target state of the \textcolor{black}{classical} Huffman algorithm. Its leaves, together with the cost of the removed nodes, form the desired output. For ease of exposition, we ignore the cost here and focus on identifying these leaves; the cost can be computed with the same algorithm. In Section~\ref{sec:iterative-state-formal}, we show that this task reduces to recovering the frequencies of all nodes at a given height of the optimal Huffman tree. Therefore, in what follows we assume that a height $h^*$ is given, and our goal is to recover the frequencies of all nodes at height $h^*$ in the optimal Huffman tree.

Keep in mind that although the optimal Huffman tree itself is not available, the input specifies the depths $h_1, h_2, \ldots, h_m$ of the leaves corresponding to the frequencies $f_1, f_2, \ldots, f_m$. These depths characterize the optimal Huffman tree implicitly, and our algorithm uses them to solve the problem.

\textcolor{black}{From} the values of $h_1, \ldots, h_m$ it is easy to find out how many nodes are \textcolor{black}{at depth} $h^*$ of the optimal Huffman tree (see Section~\ref{sec:iterative-state-formal} for more details). Let us denote their count by $q$ and denote their frequencies by $w_1\leq\cdots\leq w_q$. Each such node is either an input
leaf or the root of a subtree whose leaves occur at greater depths. In
either case, the Kraft values of the leaves in this subtree
sum to $2^{-h^*}.$

The Package-Merge
characterization~\cite{larmore1990package} provides a direct formula to determine the frequencies of these nodes:
for every $r\in\{1,\ldots,q\}$, the leaves in the $r$ lightest subtrees
have the minimum total frequency among all sets of leaves with
 Kraft value $r2^{-h^*}$. Denoting this minimum by $P(r)$, we have
\[
P(r)=\sum_{j=1}^{r}w_j
=
\min\left\{
\sum_{f_i\in A}f_i:
A\subseteq(f_1,f_2,\ldots,f_m),\quad
\sum_{f_i\in A}2^{-h_i}=r2^{-h^*}
\right\}.
\]

\begin{figure}[t]
	\centering
	\begin{tikzpicture}[
		x=0.9cm,
		y=0.72cm,
		leaf/.style={
			circle,
			draw=black,
			fill=black!8,
			minimum size=6mm,
			inner sep=1pt,
			font=\small
		},
		internal/.style={
			circle,
			draw=black,
			fill=black!10,
			minimum size=6mm,
			inner sep=1pt,
			font=\small
		},
		selected/.style={
			circle,
			draw=black,
			fill=black!15,
			minimum size=6.5mm,
			inner sep=1pt,
			font=\small\bfseries
		},
		edge/.style={
			draw=black,
			thick
		}
		]

		\node[internal] (root) at (0,7.2) {$21$};
		\node[internal] (v8)  at (-3,6) {$8$};
		\node[internal] (v13) at (3,6) {$13$};

		\node[selected] (v4a) at (-4.5,4.8) {$4$};
		\node[selected] (v4b) at (-1.5,4.8) {$4$};
		\node[selected] (v6)  at (1.5,4.8) {$6$};
		\node[selected] (f7)  at (4.5,4.8) {$7$};

		\node[internal] (v2a) at (-5.15,3.6) {$2$};
		\node[internal] (v2b) at (-3.85,3.6) {$2$};
		\node[leaf] (f1a) at (-5.5,2.4) {$1$};
		\node[leaf] (f1b) at (-4.8,2.4) {$1$};
		\node[leaf] (f1c) at (-4.2,2.4) {$1$};
		\node[leaf] (f1d) at (-3.5,2.4) {$1$};

		\node[leaf] (f2) at (-2.1,3.6) {$2$};
		\node[internal] (v2c) at (-0.9,3.6) {$2$};
		\node[leaf] (f1e) at (-1.25,2.4) {$1$};
		\node[leaf] (f1f) at (-0.55,2.4) {$1$};

		\node[leaf] (f3a) at (0.9,3.6) {$3$};
		\node[leaf] (f3b) at (2.1,3.6) {$3$};

		\draw[edge] (root) -- (v8);
		\draw[edge] (root) -- (v13);
		\draw[edge] (v8) -- (v4a);
		\draw[edge] (v8) -- (v4b);
		\draw[edge] (v13) -- (v6);
		\draw[edge] (v13) -- (f7);
		\draw[edge] (v4a) -- (v2a);
		\draw[edge] (v4a) -- (v2b);
		\draw[edge] (v2a) -- (f1a);
		\draw[edge] (v2a) -- (f1b);
		\draw[edge] (v2b) -- (f1c);
		\draw[edge] (v2b) -- (f1d);
		\draw[edge] (v4b) -- (f2);
		\draw[edge] (v4b) -- (v2c);
		\draw[edge] (v2c) -- (f1e);
		\draw[edge] (v2c) -- (f1f);
		\draw[edge] (v6) -- (f3a);
		\draw[edge] (v6) -- (f3b);

		\foreach \x in {-4.5,-1.5,1.5,4.5} {
			\draw[
				rounded corners=3pt,
				thick,
				draw=black,
				fill=black!5
			] (\x-1.15,0.15) rectangle (\x+1.15,1.35);
		}

		\foreach \xa/\xb/\lab in {-5.5/-5.0/1,-5.0/-4.5/1,-4.5/-4.0/1,-4.0/-3.5/1} {
			\draw[draw=black,fill=black!7]
				(\xa,0.4) rectangle (\xb,1.1);
			\node[font=\small] at ({(\xa+\xb)/2},0.75) {$\lab$};
		}

		\draw[draw=black,fill=black!7]
			(-2.5,0.4) rectangle (-1.5,1.1);
		\draw[draw=black,fill=black!7]
			(-1.5,0.4) rectangle (-1.0,1.1);
		\draw[draw=black,fill=black!7]
			(-1.0,0.4) rectangle (-0.5,1.1);
		\node[font=\small] at (-2.0,0.75) {$2$};
		\node[font=\small] at (-1.25,0.75) {$1$};
		\node[font=\small] at (-0.75,0.75) {$1$};

		\draw[draw=black,fill=black!7]
			(0.5,0.4) rectangle (1.5,1.1);
		\draw[draw=black,fill=black!7]
			(1.5,0.4) rectangle (2.5,1.1);
		\node[font=\small] at (1.0,0.75) {$3$};
		\node[font=\small] at (2.0,0.75) {$3$};

		\draw[draw=black,fill=black!7]
			(3.5,0.4) rectangle (5.5,1.1);
		\node[font=\small] at (4.5,0.75) {$7$};

		\node[font=\scriptsize] at (-4.5,-0.1) {$1+1+1+1$};
		\node[font=\scriptsize] at (-1.5,-0.1) {$2+1+1$};
		\node[font=\scriptsize] at (1.5,-0.1) {$2+2$};
		\node[font=\scriptsize] at (4.5,-0.1) {$4$};

		\node[font=\scriptsize,text=black] at (-4.5,1.7)
			{$P(1)=4$};
		\node[font=\scriptsize,text=black] at (-1.5,1.7)
			{$P(2)=8$};
		\node[font=\scriptsize,text=black] at (1.5,1.7)
			{$P(3)=14$};
		\node[font=\scriptsize,text=black] at (4.5,1.7)
			{$P(4)=21$};
	\end{tikzpicture}
	
	\caption{
		In this example, we have $h^*=2$ and thus there are $q=4$ nodes at depth $2$. Each enclosing box groups leaves of
		total Kraft value $1/4$. The node frequencies are
		$4,4,6,7$, and their prefix sums give the displayed values of
		$P(r)$.
	}
	\label{fig:kraft-budget-reconstruction}
\end{figure}
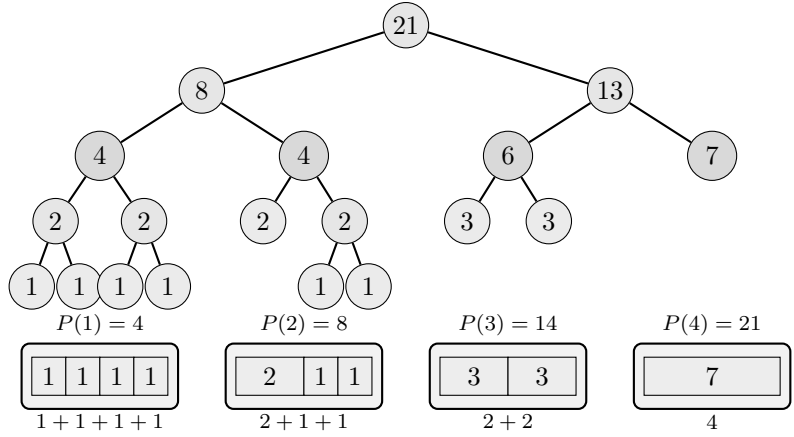

Consequently, it is enough to compute $P(r)$ for every
$r\in\{0,\ldots,q\}$. Setting $P(0)=0$, we can recover the node frequencies
from
$
w_r=P(r)-P(r-1),
$ for every $r=1,\ldots,q.
$
We therefore focus on computing the values of $P(r)$.

Notice that the optimization problem for $P(r)$ is \textcolor{black}{an exact-width knapsack problem}.
Our strategy is to begin with a fractional relaxation of each
optimization problem. The relaxed problem is indeed a fractional knapsack
instance and can be solved easily by a greedy algorithm.
We then show that an optimal integral solution can differ from the
fractional solution only on a small number of leaves near the
fractional cutoff. This allows us to fix all other choices and recover
the exact value of $P(r)$ by solving only the remaining small instance.

Fix $r\in\{1,\ldots,q\}$, and let
$\kappa_i=2^{-\LeafDepth{i}}$ denote the  Kraft value of
leaf $i$. In the fractional relaxation, instead of deciding whether or not a leaf is chosen to contribute to $P(r)$, we may select a fraction
$x_i\in[0,1]$ of each leaf \textcolor{black}{contributing an $x_i$ fraction of both its Kraft value and its frequency}. 
For the relaxed problem, one can sort the leaves by nondecreasing density
$\Frequency{i}/\kappa_i$ and select them in this order until their total
 Kraft value reaches $r2^{-h^*}$. All selected leaves are taken completely, except possibly the leaf containing the cutoff, which is selected fractionally. 
Figure~\ref{fig:fractional-greedy-cutoffs} illustrates the greedy order
and the resulting cutoffs.
\input{Sections/fractional-greedy-cutoffs-figure}

Interestingly, the fractional solution cannot be too far from an
optimal integral solution. In fact, as we show in
Section~\ref{sec:computing-prefix-optima} (Lemma~\ref{lem:fractional-proximity}), once we group the leaves
by their Kraft values \(\kappa_i\), or equivalently by
their depths, there exists an optimal integral solution that agrees
with the fractional solution everywhere except possibly on \(2c\)
leaves within each group where $c=2^{h_{\max} - h_{\min}}$. 

More precisely, consider one depth group. All leaves
in the group have the same value of \(\kappa_i\), so ordering them by
density \(\Frequency{i}/\kappa_i\) is the same as ordering them by
frequency. The fractional solution therefore determines a cutoff in
the group: leaves before the cutoff are selected, whereas those after
it are not. Moreover, an optimal integral solution can be chosen as a
prefix of the group in nondecreasing frequency order which defines a similar cutoff. 
The integral solution can differ from the fractional solution
only by moving the cutoff forward or backward: moving it forward adds
the first leaves after the fractional cutoff, while moving it backward
removes the last leaves before the cutoff. We show in Section~\ref{sec:computing-prefix-optima} that, in every depth group, the two cutoffs for the fractional solution and the integral solution can differ by at most $2c$ leaves. This basically reduces the size of the problem for each $P(r)$ down to $O(c (h_{\max} - h_{\min}))$ elements that can be solved with a small space algorithm.

\input{Sections/fractional-proximity-figure}

\begin{table}[!b]
\centering
\small
\renewcommand{\arraystretch}{1.15}
\begin{tabularx}{\textwidth}{@{}Xll@{}}
\toprule
Algorithm & Time & Total work \\
\midrule
Teng~\cite{teng1987}
 & \(O(\log^2m)\) & \(O(m^6\log^2m)\) \\
Atallah et al.~\cite{atallah1989}, first bound
 & \(O(\log^2m)\) & \(O(m^2\log m)\) \\
Atallah et al.~\cite{atallah1989}, second bound
 & \(O(\log m)\) & \(O(m^3)\) \\
Larmore--Przytycka~\cite{larmore1995}
 & \(O(\sqrt m\log m)\) & \(O(m^{3/2}\log m)\) \\
Berman et al.~\cite{berman2007}
 & \(O(\log n)\) & \(O(m\log n)\) \\
\textcolor{black}{Milidi\'u et al.~\cite{milidiu1999}}
 & \textcolor{black}{\(O(\log n\log\log m)\)} & \textcolor{black}{\(O(m)\)} \\
\midrule
\textbf{Theorem~\ref{thm:pram-prefix-input}} \textcolor{black}{(for sorted prefix sums)}
 & {\color{black}\(O_\epsilon\bigl((1+\tfrac{\log n}{\log m})\log\log m\bigr)\)}
 & \(O_\epsilon(m^{1+\epsilon}\log n\log^3m)\) \\
\addlinespace
\textbf{Corollary~\ref{cor:pram-frequencies}}
 & {\color{black}\(O_\epsilon\bigl(\tfrac{\log m}{\log\log m}+(1+\tfrac{\log n}{\log m})\log\log m\bigr)\)}
 & \(O_\epsilon(m^{1+\epsilon}\log n\log^3m)\) \\
\bottomrule
\end{tabularx}
\end{table}

To construct all these residual instances in the \MPC model,
we sort the leaves once by density and use prefix sums and batched
searches to locate the cutoff for every target Kraft value \textcolor{black}{\(r2^{-h^*}\)}.  Standard \MPC primitives perform these operations in
constant rounds.
For each \(r\), we gather the information of the relevant leaves within its cutoff windows to a
single machine, which solves the corresponding residual problem in
linear time using Package-Merge~\cite{larmore1990package}. All values
of \(P(r)\) are processed in parallel. Once the values $P(0),\ldots,P(q)$ are computed, we can recover the frequencies $w_1, w_2, \ldots, w_q$ and subsequently report the output. 

\begin{lemma}[Restated informally]
\label{lem:iterative-state-reconstruction}
Iterative-State Reconstruction can be solved exactly in
constant \MPC rounds using
$
O\!\left((1+\log\RatioBound)
\RatioBound^{\log_{\GoldenRatio}2}\right)
$
local memory per machine and
$\softO(m\RatioBound^{\log_{\GoldenRatio}2})$ total memory and
computation.
\end{lemma}

%% file: Sections/fractional-greedy-cutoffs-figure.tex
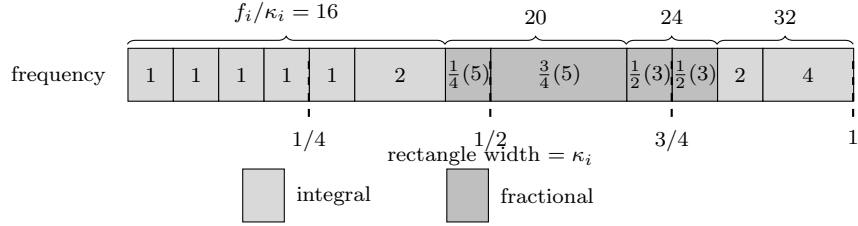
\begin{figure}[t]
  \centering
  \begin{tikzpicture}[
    x=0.6cm,
    y=0.8cm,
    full/.style={
      draw=black,
      fill=black!15,
      minimum height=7mm,
      inner sep=0pt,
      font=\scriptsize
    },
    boundary/.style={
      draw=black,
      fill=black!25,
      minimum height=7mm,
      inner sep=0pt,
      font=\scriptsize
    },
    outside/.style={
      draw=black,
      fill=black!8,
      minimum height=7mm,
      inner sep=0pt,
      font=\scriptsize
    },
    cut/.style={
      dashed,
      thick,
      draw=black
    }
  ]

    \foreach \x in {0,1,2,3,4}
      \node[full,anchor=south west,minimum width=0.6cm]
        at (\x,0) {$1$};
    \node[full,anchor=south west,minimum width=1.2cm]
      at (5,0) {$2$};
    \node[boundary,anchor=south west,minimum width=0.6cm]
      at (7,0) {$\frac14(5)$};
    \node[boundary,anchor=south west,minimum width=1.8cm]
      at (8,0) {$\frac34(5)$};
    \node[boundary,anchor=south west,minimum width=0.6cm]
      at (11,0) {$\frac12(3)$};
    \node[boundary,anchor=south west,minimum width=0.6cm]
      at (12,0) {$\frac12(3)$};
    \node[full,anchor=south west,minimum width=0.6cm]
      at (13,0) {$2$};
    \node[full,anchor=south west,minimum width=1.2cm]
      at (14,0) {$4$};

    \draw[decorate,decoration={brace,amplitude=3pt}]
      (0,0.95) -- (7,0.95)
      node[midway,above=4pt,font=\scriptsize]
      {$\Frequency{i}/\kappa_i=16$};
    \draw[decorate,decoration={brace,amplitude=3pt}]
      (7,0.95) -- (11,0.95)
      node[midway,above=4pt,font=\scriptsize] {$20$};
    \draw[decorate,decoration={brace,amplitude=3pt}]
      (11,0.95) -- (13,0.95)
      node[midway,above=4pt,font=\scriptsize] {$24$};
    \draw[decorate,decoration={brace,amplitude=3pt}]
      (13,0.95) -- (16,0.95)
      node[midway,above=4pt,font=\scriptsize] {$32$};

    \foreach \x/\lab in {4/{1/4},8/{1/2},12/{3/4},16/1} {
      \draw[cut] (\x,-0.25) -- (\x,0.75);
      \node[font=\scriptsize,anchor=north] at (\x,-0.3) {$\lab$};
    }

    \node[font=\scriptsize,anchor=east] at (-0.25,0.43)
      {frequency};
    \node[font=\scriptsize] at (8,-0.9)
      {rectangle width $=\kappa_i$};

    \node[full,minimum width=0.55cm] at (3,-1.55) {};
    \node[anchor=west,font=\scriptsize] at (3.5,-1.55)
      {integral};
    \node[boundary,minimum width=0.55cm] at (7.5,-1.55) {};
    \node[anchor=west,font=\scriptsize] at (8.0,-1.55)
      {fractional};
  \end{tikzpicture}
  \caption{Greedy fractional solutions for frequencies
    $(1,1,1,1,1,2,2,3,4,5)$. Items are ordered by density,
    and their widths are proportional to their Kraft values. \textcolor{black}{The cutoffs for target Kraft values
    $1/2$ and $3/4$} split the items of frequencies $5$ and $3$,
    respectively.}
  \label{fig:fractional-greedy-cutoffs}
\end{figure}

%% file: Sections/fractional-proximity-figure.tex
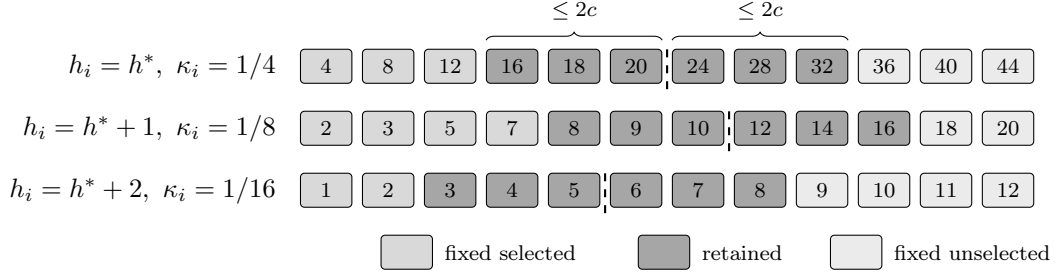
\begin{figure}[t]
	\centering
	\begin{tikzpicture}[
		x=0.82cm,
		y=0.82cm,
		leaf/.style={
			rounded corners=1.5pt,
			draw=black,
			minimum width=6.8mm,
			minimum height=4.5mm,
			inner sep=0pt,
			font=\scriptsize
		},
		fixedin/.style={
			leaf,
			draw=black,
			fill=black!15
		},
		window/.style={
			leaf,
			draw=black,
			fill=black!35
		},
		fixedout/.style={
			leaf,
			draw=black,
			fill=black!8
		},
		cut/.style={
			dashed,
			thick,
			draw=black
		}
		]

		\node[anchor=east,font=\small] at (-0.65,2.25)
			{$\LeafDepth{i}=h^*,\ \kappa_i=1/4$};
		\foreach \j/\v in {0/4,1/8,2/12}
			\node[fixedin] at (\j,2.25) {$\v$};
		\foreach \j/\v in {3/16,4/18,5/20,6/24,7/28,8/32}
			\node[window] at (\j,2.25) {$\v$};
		\foreach \j/\v in {9/36,10/40,11/44}
			\node[fixedout] at (\j,2.25) {$\v$};
		\draw[cut] (5.5,1.88) -- (5.5,2.62);

		\node[anchor=east,font=\small] at (-0.65,1.25)
			{$\LeafDepth{i}=h^*+1,\ \kappa_i=1/8$};
		\foreach \j/\v in {0/2,1/3,2/5,3/7}
			\node[fixedin] at (\j,1.25) {$\v$};
		\foreach \j/\v in {4/8,5/9,6/10,7/12,8/14,9/16}
			\node[window] at (\j,1.25) {$\v$};
		\foreach \j/\v in {10/18,11/20}
			\node[fixedout] at (\j,1.25) {$\v$};
		\draw[cut] (6.5,0.88) -- (6.5,1.62);

		\node[anchor=east,font=\small] at (-0.65,0.25)
			{$\LeafDepth{i}=h^*+2,\ \kappa_i=1/16$};
		\foreach \j/\v in {0/1,1/2}
			\node[fixedin] at (\j,0.25) {$\v$};
		\foreach \j/\v in {2/3,3/4,4/5,5/6,6/7,7/8}
			\node[window] at (\j,0.25) {$\v$};
		\foreach \j/\v in {8/9,9/10,10/11,11/12}
			\node[fixedout] at (\j,0.25) {$\v$};
		\draw[cut] (4.5,-0.12) -- (4.5,0.62);

		\draw[
			decorate,
			decoration={brace,amplitude=3pt},
			black!80!black
		] (2.58,2.67) -- (5.42,2.67)
			node[midway,above=4pt,font=\scriptsize] {$\leq 2c$};
		\draw[
			decorate,
			decoration={brace,amplitude=3pt},
			black!80!black
		] (5.58,2.67) -- (8.42,2.67)
			node[midway,above=4pt,font=\scriptsize] {$\leq 2c$};

		\node[fixedin] at (1.3,-0.75) {};
		\node[anchor=west,font=\scriptsize] at (1.75,-0.75)
			{fixed selected};
		\node[window] at (5.45,-0.75) {};
		\node[anchor=west,font=\scriptsize] at (5.9,-0.75)
			{retained};
		\node[fixedout] at (8.55,-0.75) {};
		\node[anchor=west,font=\scriptsize] at (9.0,-0.75)
			{fixed unselected};
	\end{tikzpicture}
	\caption{
		\textcolor{black}{Windows around the fractional density cutoff $\rho=88$.} Each box shows a frequency;
		lightly shaded choices are fixed, while darkly shaded frequencies are
		retained for the residual instance. \textcolor{black}{The braces indicate at most $2c$ retained leaves on each side of the cutoff within a depth group, where $c=2^{h_{\max}-h_{\min}}$.}
	}
	\label{fig:fractional-proximity-window}
\end{figure}

%% file: Sections/bounded-ratio-huffman.tex
\clearpage
\section{Bounded-Ratio Huffman Coding}
\label{sec:bounded-ratio-formal}

This section gives the formal algorithm for Bounded-Ratio Huffman
Coding.  The input is a sorted multiset
\(A=(f_1,\ldots,f_m)\), where \(m\geq2\) and
\(f_m/f_2\leq\lambda\).  The algorithm computes the optimum encoded
length and one optimal depth vector. 

We solve the bounded-ratio problem using a divide-and-conquer approach. We first show that the leaf depths in an optimal Huffman tree lie within a small range. This allows us to express the Kraft constraint as an integer budget. We then divide the sorted frequencies into consecutive blocks and construct a table for each block. Each table records the minimum cost of assigning depths to the frequencies in that block, for every possible Kraft budget. To solve the entire instance, we arrange these tables at the leaves of a balanced hierarchy and combine them from the bottom up. Each combination determines how to divide the available budget between two neighboring subproblems.

In this section, we establish the depth bound, explain how to construct the tables, and describe the hierarchy. Section~\ref{sec:layered-merging-formal} explains how to combine two tables efficiently.

Let \(\varphi=(1+\sqrt5)/2\), and define
\[
 \Delta_\lambda=\left\lceil\log_\varphi\lambda\right\rceil+1.
\]
The key observation is that node { frequencies} grow quickly as we move from a
leaf toward the root.  In fact, they grow at least as fast as the
Fibonacci sequence.  This growth lets us translate the ratio between two
frequencies into a bound on the difference between their depths.
Lemma~\ref{lem:pairwise-depth-bound} formalizes this connection.

\begin{lemma}
\label{lem:pairwise-depth-bound}
{ Let \((\LeafDepth{1},\ldots,\LeafDepth{m})\) be any optimal
depth assignment for the positive frequencies
\(\Frequency{1},\ldots,\Frequency{m}\). For any \(i,j\) with
\(\Frequency{i}\leq\Frequency{j}\),}
\[
 h_i-h_j\leq
 \left\lceil\log_\varphi(f_j/f_i)\right\rceil+1.
\]
\end{lemma}

\begin{proof}
{
{\color{black} By the Kraft--McMillan theorem, there is a binary coding
tree in which leaf \(i\) has depth \(h_i\) for every \(i\).
Consider such a tree, and define the frequency of each node as the sum
of the frequencies of the leaves in its subtree. Since the depth
assignment is optimal, every internal node has two children; otherwise,
removing a node with only one child would decrease the cost.
Let \(d=h_i-h_j\). The claim is immediate when \(d\leq1\), so
assume \(d\geq2\).} We show that frequencies grow at least as fast as the
Fibonacci sequence along the path from leaf \(i\) toward the root.

Let \(w_0=f_i\), and let \(w_t\) be the frequency of its \(t\)-th
ancestor. For \(1\leq t<h_i\), denote the frequency of that ancestor's
sibling by \(w'_t\), as illustrated in
Figure~\ref{fig:optimal-depth-exchange}. The node with frequency \(w'_t\) is at depth \(h_i-t\), whereas the
node with frequency \(w_{t-1}\) is at depth \(h_i-t+1\).
Swapping these disjoint subtrees moves every leaf below \(w_{t-1}\)
up one depth and every leaf below \(w'_t\) down one depth.Therefore, the swap changes the total cost by $w'_t-w_{t-1}$. Since the original tree is optimal, this change cannot be negative, implying \(w'_t\geq w_{t-1}\), and hence
\[
 w_{t+1}=w_t+w'_t\geq w_t+w_{t-1}.
\]
Since \(w_0=f_i\)
and \(w_1>f_i\), the recurrence gives
\[
 w_t>\textcolor{black}{\FibonacciNumber{t+1}}f_i\qquad(1\leq t\leq h_i),
\]
where \(\textcolor{black}{\FibonacciNumber{t}}\) denotes the \textcolor{black}{\(t\)-th Fibonacci number}, with
\(\textcolor{black}{\FibonacciNumber{1}}=\textcolor{black}{\FibonacciNumber{2}}=1\).
A similar exchange argument gives \(w_{d-1}\leq f_j\). Hence
\[
 \textcolor{black}{\FibonacciNumber{d}}f_i<w_{d-1}\leq f_j.
\]
Using the standard Fibonacci bound
\(\textcolor{black}{\FibonacciNumber{d}}\geq\varphi^{d-2}\) for \(d\geq2\)\textcolor{black}{~\cite[Lemma~19.3]{cormen2009}}, we obtain
\(\varphi^{d-2}<f_j/f_i\). Since \(d\) is an integer,
\[
 d\leq\left\lceil\log_\varphi(f_j/f_i)\right\rceil+1.\qedhere
\]
}
\end{proof}
\input{Sections/optimal-depth-exchange-figure}

As a consequence of Lemma~\ref{lem:pairwise-depth-bound}, applying its
bound to \(f_2\) and \(f_m\) and using Kraft equality gives the depth
range stated in Corollary~\ref{cor:bounded-ratio-depth-range}.

\begin{corollary}
\label{cor:bounded-ratio-depth-range}
There is an optimal tree in which every leaf depth lies in
\([\ell,u]\), where
\[
 \ell=\textcolor{black}{\max\{1,\lceil\log_2m\rceil-\Delta_\lambda\}},
 \qquad
 u=\left\lfloor\log_2m\right\rfloor+\Delta_\lambda.
\]
\end{corollary}

\begin{proof}
{\color{black}
Let \(h_{\min}\) and \(h_{\max}\) be the minimum and maximum leaf
depths of an optimal tree. Kraft equality gives
\[
 h_{\min}\leq\lfloor\log_2m\rfloor,
 \qquad
 h_{\max}\geq\lceil\log_2m\rceil.
\]
Two sibling leaves occur at depth \(h_{\max}\), and at least one has
frequency at least \(f_2\). Comparing this leaf with a minimum-depth
leaf and applying Lemma~\ref{lem:pairwise-depth-bound} gives
\(h_{\max}-h_{\min}\leq\Delta_\lambda\). Therefore,
\[
 h_{\min}\geq h_{\max}-\Delta_\lambda
 \geq\lceil\log_2m\rceil-\Delta_\lambda,
\]
\[
 h_{\max}\leq h_{\min}+\Delta_\lambda
 \leq\lfloor\log_2m\rfloor+\Delta_\lambda=u.
\]
\textcolor{black}{Since leaf depths are at least one,} every leaf depth lies in \([\ell,u]\).
}
\end{proof}

{\color{black} Note that \(u-\ell\leq2\Delta_\lambda\).}
Corollary~\ref{cor:bounded-ratio-depth-range} allows us to restrict every
depth to \([\ell,u]\).  We next use this bound on leaf depths to express the Kraft constraint as an integer budget, which allows us to solve the problem using divide and conquer.  Recall from
Equation~\eqref{eq:huffman-length-formulation} that a feasible depth
assignment satisfies
\[
 \sum_i2^{-h_i}\leq1.
\]
Multiplying this inequality by \(2^u\) gives
\[
 \sum_i2^{u-h_i}\leq2^u.
\]
We can therefore interpret \(2^u\) as an integer Kraft budget, where assigning frequency \(f_i\) to depth \(h_i\) consumes \(2^{u-h_i}\) units. This interpretation suggests a divide-and-conquer approach: we divide the frequencies into smaller subproblems, compute the minimum cost for each possible budget, and then combine their solutions by distributing the available budget among them. To formalize this idea, we introduce the budget-constrained Huffman coding problem in Problem~\ref{prob:budget-constrained-huffman-formal}.

\begin{problem}[Budget-Constrained Huffman Coding]
\label{prob:budget-constrained-huffman-formal}
Given a submultiset \(A'\subseteq A\) and an integer budget
\(0\leq b\leq2^u\), we seek to assign a depth
\(\textcolor{black}{h_i\in\{\ell,\ell+1,\ldots,u\}}\) to every frequency \(f_i\in A'\) such
that the encoded length \(\sum_{f_i\in A'}f_ih_i\) is minimized and the
total budget contribution \(\sum_{f_i\in A'}2^{u-h_i}\) does not exceed
\(b\).  In other words, we want to compute
\begin{equation}
\label{eq:formal-budget-table}
 \mu(A',b)=
 \min\left\{
   \sum_{f_i\in A'}f_ih_i:
   \sum_{f_i\in A'}2^{u-h_i}\leq b,
   \ \textcolor{black}{h_i\in\{\ell,\ell+1,\ldots,u\}}
 \right\}.
\end{equation}
The output consists of this minimum value and a depth assignment that
attains it.  We set \(\mu(A',b)=+\infty\) when the constraints are
infeasible.
\end{problem}

The budget \(b\) in Equation~\eqref{eq:formal-budget-table} is an upper
bound on the total budget usage, so a feasible assignment may use fewer
than \(b\) units.  Corollary~\ref{cor:bounded-ratio-depth-range}
guarantees that restricting the depths to \([\ell,u]\) retains an
optimal solution to the original bounded-ratio instance.  Consequently,
\(\mu(A,2^u)\) equals its optimum encoded length.

Problem~\ref{prob:budget-constrained-huffman-formal} admits a natural
divide-and-conquer approach.  For disjoint multisets \(A_1,A_2\),
additivity of the cost and the budget contribution gives
\begin{equation}
\label{eq:formal-kraft-recurrence}
 \mu(A_1\uplus A_2,b)=
 \min_{b_1+b_2=b}
 \bigl\{\mu(A_1,b_1)+\mu(A_2,b_2)\bigr\}.
\end{equation}

Note that, since the budgets are upper bounds, any unused units can be assigned
to either part, and it is enough to consider divisions satisfying
\(b_1+b_2=b\).  Equation~\eqref{eq:formal-kraft-recurrence} therefore
gives a bottom-up construction.  We begin by partitioning the sorted
frequencies into consecutive blocks, as illustrated in
Figure~\ref{fig:representative-blocks}.
Set \(r_1=1\).  After choosing \(r_j\), if some later frequency is
larger than \(\varphi f_{r_j}\), let \(r_{j+1}\) be the smallest index
such that
$
 f_{r_{j+1}}>\varphi f_{r_j}.
$
Continue until no such index exists.  Let \(r_1,\ldots,r_q\) be the
resulting representative indices.  We append the end marker
\(r_{q+1}=m+1\) and define
\[
 B_j=(f_{r_j},f_{r_j+1},\ldots,f_{r_{j+1}-1}),
 \qquad j=1,\ldots,q.
\]
For every block except the last one, \(r_{j+1}\) is the first index
whose frequency exceeds \(\varphi f_{r_j}\).  Thus, every frequency
included in \(B_j\) is at most \(\varphi f_{r_j}\).  The same conclusion
holds for the last block because the construction stops only when no
later frequency exceeds this threshold.  Hence, the frequency ratio
within every block is at most \(\varphi\).

We next bound the number of blocks.  All representatives after \(r_1\)
have frequencies between \(f_2\) and \(f_m\), and consecutive
representatives differ by a factor greater than \(\varphi\).  Since
\(f_m/f_2\leq\lambda\), they define at most
\(1+\lceil\log_\varphi\lambda\rceil\) blocks.  If \(f_1\) is much
smaller than \(f_2\), it may form one additional initial block.
Therefore, the total number \(q\) of blocks satisfies
\[
 q\leq2+\left\lceil\log_\varphi\lambda\right\rceil.
\]

{\color{black}
To initialize the tables, we need to compute \(\mu(B_j,b)\) for each
block \(B_j\) and every relevant budget \(b\). The useful property of
these blocks is that their frequencies differ by a factor of at most
\(\varphi\). Using an exchange argument similar to the one in Lemma~\ref{lem:pairwise-depth-bound}, we can show that there exists an optimal assignment in which the depths of all frequencies in a block lie within three consecutive levels. This property allows us to consider only a small number of depth patterns when constructing the table for each block. We formalize this restricted
problem in Problem~\ref{prob:block-budget-table} and then describe how
to solve it.

\begin{problem}
\label{prob:block-budget-table}
Given a sorted block \(B=(w_1,\ldots,w_s)\) of \(A\) with
\(0<w_1\leq\cdots\leq w_s\leq\varphi w_1\), compute \(\mu(B,b)\)
for every \(b\in\{s,s+1,\ldots,\BlockBudgetMax\}\), where
\(\BlockBudgetMax=\min\{2^u,s2^{u-\ell}\}\).
\end{problem}

To solve Problem~\ref{prob:block-budget-table}, we first show that an
optimal assignment uses at most three consecutive depths. We then use
this property to evaluate a constant number of candidates for each
exact budget usage and compute the table by taking prefix minima.
Lemma~\ref{lem:constant-time-base-case} gives the resulting \MPC bounds.
In these bounds, logarithms are understood to be at least one.
}

\begin{lemma}
\label{lem:constant-time-base-case}
{\color{black}
For any \(\sigma\geq2\), Problem~\ref{prob:block-budget-table} can be
solved using \(O(\sigma)\) memory per machine in
\(O(\log_\sigma\BlockBudgetMax)\) \MPC rounds, with total memory
\[
 O\!\left(\BlockBudgetMax
       \log_\sigma\BlockBudgetMax\right).
\]
The total computation is
\[
 O\!\left(\BlockBudgetMax\left[
 \log_\sigma\BlockBudgetMax\log\BlockBudgetMax
 +\log u\right]\right).
\]

}
\end{lemma}

\begin{proof}
	We first show that every optimal assignment uses at most three
	consecutive depths. Consider an optimal assignment with budget
	usage at most $b$. Suppose two frequencies $w_i,w_j$ are assigned
	to depths $a<c$, where $c\geq a+3$. By optimality, $w_i\geq w_j$,
	while the block's frequency ratio implies
	$w_i\leq\varphi w_j<2w_j$.
	Now move $w_i$ from depth $a$ to $a+1$, and $w_j$ from depth
	$c$ to $c-2$. Both new depths remain in $[\ell,u]$.
	This change decreases the cost by $2w_j-w_i>0$ and the
	Kraft sum by $2^{-a-1}-3\cdot 2^{-c}>0$.
	Hence, the new assignment has both a smaller cost and a smaller
	budget usage, contradicting optimality.
	
	Let $\ExactBlockCost(B,b')$ denote the minimum cost of an assignment
	that uses exactly $b'$ budget units and at most three consecutive
	depths in $[\ell,u]$. Set $\ExactBlockCost(B,b')=+\infty$ if no
	such assignment exists. Since every optimal assignment for budget
	$b$ uses at most three consecutive depths, we have
	\[
	\mu(B,b)=
	\min_{b'\in\{s,\ldots,b\}}\ExactBlockCost(B,b').
	\]
	Thus, we can compute all values of $\mu(B,b)$ by first evaluating
	$\ExactBlockCost(B,b')$ for every
	$b'\in\{s,\ldots,\BlockBudgetMax\}$ and then taking prefix minima.
	We next explain how to compute $\ExactBlockCost(B,b')$ for a fixed $b'$.
	
	If $\ell=u$, every frequency must be assigned to depth $\ell$.
	The total budget usage is $s$, so
	$\ExactBlockCost(B,b')=\ell\sum_{i=1}^s w_i$ when $b'=s$,
	and $+\infty$ otherwise.
	
	If $u=\ell+1$, each frequency must be assigned to either
	depth $\ell$ or depth $u$. Suppose $x$ frequencies are assigned
	to depth $\ell$ and the remaining $s-x$ to depth $u$.
	The total budget usage is then
	$2x+(s-x)=s+x$, which gives $x=b'-s$.
	If $0\leq x\leq s$, the minimum cost is obtained by assigning
	the $x$ largest frequencies to the smaller depth $\ell$.
	Otherwise, no assignment with budget usage $b'$ exists.
	
	It remains to consider $u\geq\ell+2$.
	By the argument above, we only need to consider assignments
	using three consecutive depths $\BaseWindowDepth,\BaseWindowDepth+1,\BaseWindowDepth+2$.
	For a fixed budget usage $b'$, we first identify the possible
	values of $\BaseWindowDepth$.
	
	Since each of the $s$ frequencies is assigned to one of these
	three depths, the total budget usage is at least
	$s2^{u-\BaseWindowDepth-2}$ and at most $s2^{u-\BaseWindowDepth}$.
	Therefore,
	\[
	s2^{u-\BaseWindowDepth-2}\leq b'\leq s2^{u-\BaseWindowDepth},
	\]
	which implies
	\[
	u-\log_2(b'/s)-2
	\leq \BaseWindowDepth\leq
	u-\log_2(b'/s).
	\]
	This interval contains at most three integers.
	We consider those satisfying $\ell\leq \BaseWindowDepth\leq u-2$.
	These candidates can be identified in $O(\log u)$
	local operations by binary search over powers of two.
	
	Fix one of the candidate values of $\BaseWindowDepth$.
	We now determine how many frequencies should be assigned to
	each of the three depths.
	Let $x$, $y$, and $z$ denote the numbers of frequencies assigned
	to depths $\BaseWindowDepth$, $\BaseWindowDepth+1$, and $\BaseWindowDepth+2$, respectively.
	Since there are $s$ frequencies, we must have $x+y+z=s$.
	Moreover, a frequency at these depths consumes
	$4$, $2$, or $1$ units of $2^{u-\BaseWindowDepth-2}$, respectively.
	Thus, the total budget usage must satisfy
	\[
	x+y+z=s,
	\qquad
	4x+2y+z=\frac{b'}{2^{u-\BaseWindowDepth-2}}.
	\]
	
	Let $v=b'/2^{u-\BaseWindowDepth-2}$.
	If $v$ is not an integer, no assignment is possible for
	this choice of $\BaseWindowDepth$.
	Otherwise, the two equations determine $y$ and $z$
	once $x$ is fixed:
	\[
	y=v-s-3x,
	\qquad
	z=2s-v+2x.
	\]
	Thus, instead of choosing three counts, we only need to choose $x$.
	For a valid assignment, all three counts must be nonnegative.
	Substituting the expressions for $y$ and $z$ gives
	\[
	\max\left\{
	0,\left\lceil\frac{v-2s}{2}\right\rceil
	\right\}
	\leq x\leq
	\left\lfloor\frac{v-s}{3}\right\rfloor.
	\]
	If this interval is empty, the chosen triple is infeasible.
	
	We next show that the smallest feasible value of $x$
	minimizes the cost.
	For fixed counts $x,y,z$, larger frequencies must receive
	smaller depths, as illustrated in
	Figure~\ref{fig:three-depth-base-case}.
	Thus, the largest $x$ frequencies receive depth $\BaseWindowDepth$,
	the next $y$ receive depth $\BaseWindowDepth+1$, and the smallest $z$
	receive depth $\BaseWindowDepth+2$.
	
	Recall that $y=v-s-3x$ and $z=2s-v+2x$.
	Increasing $x$ by one therefore changes the three counts
	by $+1,-3,+2$, respectively.
	In the optimal sorted assignment, this moves one frequency
	from depth $\BaseWindowDepth+1$ to depth $\BaseWindowDepth$, and the two \textcolor{black}{smallest} frequencies
	previously at depth $\BaseWindowDepth+1$ are reassigned to depth $\BaseWindowDepth+2$
	after the group boundaries are updated.
	More precisely, the frequency moving to depth $\BaseWindowDepth$
	is $w_{s-x}$, while the two frequencies moving to depth
	$\BaseWindowDepth+2$ are $w_{z+1}$ and $w_{z+2}$.
	The resulting change in cost is
	\[
	w_{z+1}+w_{z+2}-w_{s-x}
	\geq 2w_1-w_s>0,
	\]
	where the last inequality follows from the bounded frequency
	ratio of the block.
	Therefore, increasing $x$ always increases the cost,
	and the smallest feasible value of $x$ is optimal.
	
	\input{Sections/three-depth-base-case-figure}
	
	Once $x,y,z$ are fixed, we can compute the cost using
	prefix sums of the block frequencies.
	The largest $x$ frequencies have sum
	$\PrefixSum{B}{s}-\PrefixSum{B}{s-x}$,
	the next $y$ have sum
	$\PrefixSum{B}{s-x}-\PrefixSum{B}{z}$,
	and the smallest $z$ have sum $\PrefixSum{B}{z}$.
	Multiplying these sums by their assigned depths gives
	\[
	\begin{aligned}
		\BaseWindowDepth\bigl(\PrefixSum{B}{s}-\PrefixSum{B}{s-x}\bigr)
		 +(\BaseWindowDepth+1)\bigl(
		\PrefixSum{B}{s-x}-\PrefixSum{B}{z}
		\bigr)
		+(\BaseWindowDepth+2)\PrefixSum{B}{z}.
	\end{aligned}
	\]
	Thus, once the prefix sums are available, each candidate
	cost can be evaluated using $O(1)$ arithmetic operations.
	
	For each exact usage $b'$, we evaluate the feasible candidates
	and take their minimum to obtain $\ExactBlockCost(B,b')$.
	If none is feasible, we set this value to $+\infty$.
	We then compute prefix minima over these values.
	Since budget $b$ is an upper bound on the usage, its table entry
	is the minimum cost among all exact usages $b'\leq b$,
	which is precisely $\mu(B,b)$.
	For each entry, we also store the usage, depths, and counts
	attaining its cost, using $O(1)$ additional words.
	
	We finally analyze the MPC implementation.
	There are $s$ input frequencies and
	$\BlockBudgetMax-s+1$ table entries.
	Since each entry requires only a constant number of candidates
	and prefix-sum accesses, the total number of input,
	candidate, query, and output records is $O(\BlockBudgetMax)$.
	
We compute the prefix sums of the block frequencies once.
As each candidate requires only a constant number of prefix-sum
values, the total number of data accesses is
$O(\BlockBudgetMax)$. These accesses, including repeated
requests, as well as the final prefix minima, can be handled
using standard MPC sorting and scan primitives
~\cite[Section~3.1]{goodrich2011sorting}.

These primitives process $O(\BlockBudgetMax)$ records in
$O(\log_\sigma\BlockBudgetMax)$ \MPC rounds using $O(\sigma)$
memory per machine. Candidate-depth searches are local and add no
communication rounds. Thus the overall round bound is
$
 O(\log_\sigma\BlockBudgetMax).
$
Allowing $O(\BlockBudgetMax)$ records per round gives the total-memory bound
$
 O\!\left(\BlockBudgetMax
       \log_\sigma\BlockBudgetMax\right).
$
Sorting and scans use
$O(\BlockBudgetMax\log_\sigma\BlockBudgetMax\log\BlockBudgetMax)$
total computation, and identifying the candidate depths adds
$O(\BlockBudgetMax\log u)$.
The total computation is therefore
\[
 O\!\left(\BlockBudgetMax\left[
 \log_\sigma\BlockBudgetMax\log\BlockBudgetMax
 +\log u\right]\right),
\]
matching the bounds in the lemma.

\end{proof}

\textcolor{black}{Algorithm~\ref{alg:block-budget-table} summarizes the
computation in Lemma~\ref{lem:constant-time-base-case}. The loop over
\(\BaseWindowDepth\) considers at most three values, found by the power-of-two search
in the proof.}
\input{Sections/block-budget-algorithm}

For a submultiset \(I\), define its relevant budget range by
\[
 \mathcal B(I)=
 \bigl\{|I|,|I|+1,\ldots,
 \min\{2^u,|I|2^{u-\ell}\}\bigr\}.
\]
Lemma~\ref{lem:constant-time-base-case} computes
\(\mu(B_j,b)\) for every block \(B_j\) and every
\(b\in\mathcal B(B_j)\).  We then combine consecutive blocks in
stages.  The first stage considers
\(B_1\uplus B_2,B_3\uplus B_4,\ldots\).  For every such union \(I\)
and every \(b\in\mathcal B(I)\), we compute \(\mu(I,b)\) using
Equation~\eqref{eq:formal-kraft-recurrence}.  If the number of blocks is
odd, the last block proceeds unchanged to the next stage.  The second
stage considers
\(B_1\uplus\cdots\uplus B_4,B_5\uplus\cdots\uplus B_8,\ldots\), and
the same process continues until the union is \(A\).  The unions within
each stage are disjoint and are processed in parallel.
Algorithm~\ref{alg:bottom-up-budget-tables} formalizes this construction.

\begin{algorithm}[t]
\caption{Bounded-Ratio Huffman Coding}
\label{alg:bottom-up-budget-tables}
{\color{black}
\begin{algorithmic}[1]
\Require Blocks \(B_1,\ldots,B_q\) and depth bounds \(\ell,u\)
\Ensure \(\mu(A,2^u)\) and the stored choices needed to recover the depths
\ForAll{\(j\in\{1,\ldots,q\}\) \textbf{in parallel}}
  \State Initialize \(T_{B_j}\) and its depth assignments using Algorithm~\ref{alg:block-budget-table}
\EndFor
\State \(\mathcal I\gets(B_1,\ldots,B_q)\)
\While{\(|\mathcal I|>1\)}
  \State \(k\gets|\mathcal I|\); create the ordered list \(\mathcal J\) for the next stage
  \ForAll{consecutive pairs \((I_{2j-1},I_{2j})\) \textbf{in parallel}}
    \State \(I\gets I_{2j-1}\uplus I_{2j},\quad J_j\gets I\)
    \ForAll{\(b\in\mathcal B(I)\) \textbf{in parallel}}
      \State \((b_1,b_2)\gets
        \arg\min_{\substack{b_1+b_2=b\\
          b_1\in\mathcal B(I_{2j-1}),\ b_2\in\mathcal B(I_{2j})}}
        \{T_{I_{2j-1}}[b_1]+T_{I_{2j}}[b_2]\}\)
      \State \(T_I[b]\gets T_{I_{2j-1}}[b_1]+T_{I_{2j}}[b_2]\)
      \State Store the split \((b_1,b_2)\) with \(T_I[b]\)
    \EndFor
  \EndFor
  \If{\(k\) is odd}
    \State Carry the unpaired block into the last position of \(\mathcal J\)
  \EndIf
  \State \(\mathcal I\gets\mathcal J\)
\EndWhile
\State \Return \(T_A[2^u]\), the stored splits, and the block depth assignments
\end{algorithmic}
}
\end{algorithm}

The construction has \(\lceil\log_2q\rceil=
O(\log(2+\log\lambda))\) merging stages.  Section~\ref{sec:layered-merging-formal}
shows how to compute the values for every resulting union without
enumerating all pairs of budget divisions.  At the final stage,
\(\mu(A,2^u)\) is the optimum encoded length.

The upward construction stores, for every computed value \(\mu(I,b)\),
one budget split that attains the minimum in
Equation~\eqref{eq:formal-kraft-recurrence}.  To recover a depth vector,
we start from \((A,2^u)\) and follow these choices down the hierarchy.
If \(I=I_1\uplus I_2\) and the stored choice splits \(b\) into
\(b_1+b_2=b\), we assign budget \(b_1\) to \(I_1\) and budget \(b_2\)
to \(I_2\).  Repeating this step determines one budget for every leaf
block. \textcolor{black}{Once the budget of each block is determined,
its stored assignment gives the leaf depths. In the \MPC implementation,
all nodes at the same hierarchy level pass the stored budgets to their
children in parallel, using \(O(1)\) rounds per level. Thus this top-down
pass takes \(O(1+\log q)=O(\log(2+\log\lambda))\) rounds.}

%% file: Sections/optimal-depth-exchange-figure.tex
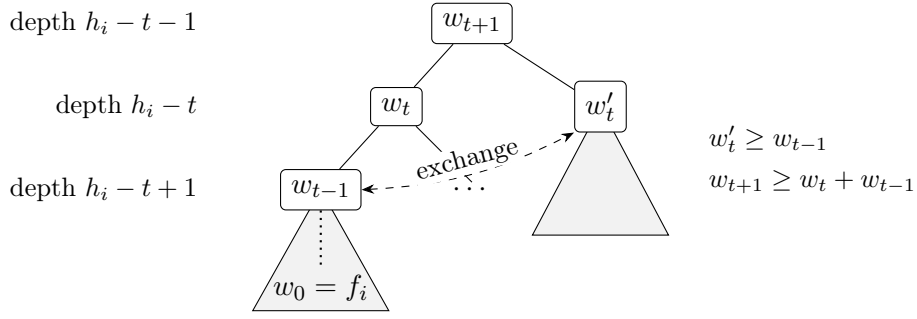
\begin{figure}[t]
\centering
{
\begin{tikzpicture}[x=1cm,y=1cm,
  freq/.style={draw,rounded corners=2pt,fill=white,inner sep=4pt},
  depth/.style={anchor=east,font=\small}]
  \node[depth] at (-3.5,0) {depth $h_i-t-1$};
  \node[depth] at (-3.5,-1.1) {depth $h_i-t$};
  \node[depth] at (-3.5,-2.2) {depth $h_i-t+1$};
  \draw[fill=black!5] (1.7,-1.1)--(0.8,-2.8)--(2.6,-2.8)--cycle;
  \draw[fill=black!5] (-2,-2.2)--(-2.9,-3.8)--(-1.1,-3.8)--cycle;
  \node[freq] (parent) at (0,0) {$w_{t+1}$};
  \node[freq] (ancestor) at (-1,-1.1) {$w_t$};
  \node[freq] (sibling) at (1.7,-1.1) {$w'_t$};
  \node[freq] (child) at (-2,-2.2) {$w_{t-1}$};
  \node at (0,-2.2) {$\cdots$};
  \node (leaf) at (-2,-3.5) {$w_0=f_i$};
  \draw (parent)--(ancestor) (parent)--(sibling)
        (ancestor)--(child) (ancestor)--(0,-2.1);
  \draw[dotted,thick] (child)--(leaf);
  \draw[<->,dashed,>=Stealth] (child.east) to[bend right=12]
      node[above,sloped,font=\small,fill=white,inner sep=1pt] {exchange} (sibling.south west);
  \node[align=left,font=\small,anchor=west] at (3,-1.8)
      {$w'_t\geq w_{t-1}$\\[4pt]$w_{t+1}\geq w_t+w_{t-1}$};
\end{tikzpicture}
\caption{The exchange used in Lemma~\ref{lem:pairwise-depth-bound}.
Labels denote node frequencies, and triangles denote disjoint subtrees.
Moving the subtree of frequency $w_{t-1}$ up one depth and the subtree
of frequency $w'_t$ down one depth changes the cost by $w'_t-w_{t-1}$,
which must be nonnegative. The dotted path may have length zero when
$t=1$.}
\label{fig:optimal-depth-exchange}
}
\end{figure}

%% file: Sections/three-depth-base-case-figure.tex
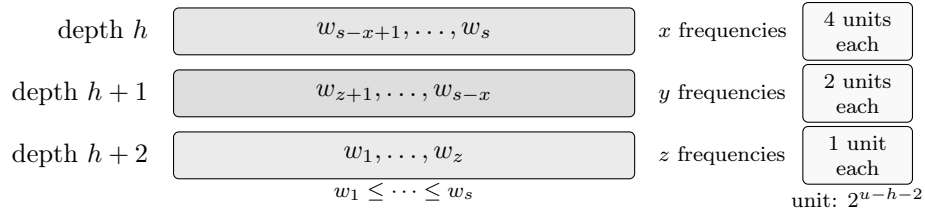
\begin{figure}[t]
  \centering
  \begin{tikzpicture}[
      x=1cm,
      y=0.82cm,
      band/.style={rounded corners=2pt,draw=black,
        minimum width=6.1cm,minimum height=0.62cm,
        font=\small,align=center},
      depthlabel/.style={font=\small,anchor=east},
      count/.style={font=\scriptsize,anchor=west},
      budget/.style={rounded corners=2pt,draw=black,fill=gray!5,
        font=\scriptsize,align=center,minimum width=1.45cm,
        minimum height=0.55cm}
    ]
    \node[depthlabel] at (0,2.0) {depth $\BaseWindowDepth$};
    \node[band,draw=black,fill=black!10] (shallow) at (3.25,2.0)
      {$w_{s-x+1},\ldots,w_s$};
    \node[count] at (6.48,2.0) {$x$ frequencies};
    \node[budget] at (9.25,2.0) {$4$ units\\each};

    \node[depthlabel] at (0,1.0) {depth $\BaseWindowDepth+1$};
    \node[band,draw=black,fill=black!13] (middle) at (3.25,1.0)
      {$w_{z+1},\ldots,w_{s-x}$};
    \node[count] at (6.48,1.0) {$y$ frequencies};
    \node[budget] at (9.25,1.0) {$2$ units\\each};

    \node[depthlabel] at (0,0) {depth $\BaseWindowDepth+2$};
    \node[band,draw=black,fill=black!8] (deep) at (3.25,0)
      {$w_1,\ldots,w_z$};
    \node[count] at (6.48,0) {$z$ frequencies};
    \node[budget] at (9.25,0) {$1$ unit\\each};

    \node[font=\scriptsize,text=black] at (3.25,-0.58)
      {$w_1\leq\cdots\leq w_s$};
    \node[font=\scriptsize,text=black,align=center] at (9.25,-0.68)
      {unit: $2^{u-\BaseWindowDepth-2}$};
  \end{tikzpicture}
  \caption{An assignment within a three-depth window.  The sorted block
  is divided into three consecutive frequency ranges with per-frequency
  budget contributions $4,2,1$, in units of $2^{u-\BaseWindowDepth-2}$.  The total
  frequency of each range is obtained from two prefix sums.}
  \label{fig:three-depth-base-case}
\end{figure}

%% file: Sections/block-budget-algorithm.tex
\begin{algorithm}[t]
\caption{Computing the table for one block}
\label{alg:block-budget-table}
{\color{black}
\begin{algorithmic}[1]
\Require Sorted \(B=(w_1,\ldots,w_s)\), \(w_s\leq\varphi w_1\), depth bounds \(\ell,u\)
\Ensure \(T_B[b]=\mu(B,b)\) and the corresponding depths and counts
\State \(\BlockBudgetMax\gets\min\{2^u,s2^{u-\ell}\}\)
\State Compute the prefix sums \(\PrefixSums{B}\) of the block frequencies
\ForAll{\(b'\in\{s,\ldots,\BlockBudgetMax\}\) \textbf{in parallel}}
  \State \(\ExactBlockCost(B,b')\gets+\infty\)
  \If{\(u-\ell<2\)}
    \State Determine the counts at the allowed depths from \(s\) and \(b'\)
    \State If feasible, set \(\ExactBlockCost(B,b')\) to their cost and store the counts
  \Else
    \For{each \(\BaseWindowDepth\in\{\ell,\ldots,u-2\}\) satisfying \(s2^{u-\BaseWindowDepth-2}\leq b'\leq s2^{u-\BaseWindowDepth}\)}
      \State Find the smallest feasible \(x\) and its counts \(y,z\), as in Lemma~\ref{lem:constant-time-base-case}
      \If{such counts exist}
        \State \(c\gets \BaseWindowDepth\PrefixSum{B}{s}+\PrefixSum{B}{s-x}+\PrefixSum{B}{z}\)
        \If{\(c<\ExactBlockCost(B,b')\)}
          \State \(\ExactBlockCost(B,b')\gets c\)
          \State Store depths \(\BaseWindowDepth,\BaseWindowDepth+1,\BaseWindowDepth+2\) and counts \(x,y,z\)
        \EndIf
      \EndIf
    \EndFor
  \EndIf
\EndFor
\State \(T_B[b]\gets\min_{s\leq b'\leq b}\ExactBlockCost(B,b')\) for all \(b\), using a prefix minimum
\State Keep the depths and counts of a minimizing entry with each finite \(T_B[b]\)
\State \Return \(T_B\) and the stored depths and counts
\end{algorithmic}
}
\end{algorithm}

%% file: Sections/fast-layered-merging.tex
\newpage
\section{Fast Merging by Layered Convexity}
\label{sec:layered-merging-formal}

Section~\ref{sec:bounded-ratio-formal} leaves one task unresolved.
Given the optimum values for two disjoint frequency multisets, we must
compute the corresponding values for their union, for every budget.
By Equation~\eqref{eq:formal-kraft-recurrence}, this requires choosing
how to divide each budget between the two multisets.  This is precisely
Problem~\ref{prob:kraft-budget-merging}, which we restate here.

\begin{mergingproblemrestatement}
Let \(A_1\) and \(A_2\) be two disjoint frequency multisets, and suppose
that the values \(\mu(A_1,\cdot)\) and \(\mu(A_2,\cdot)\) are known.
For \(A=A_1\uplus A_2\), compute \(\mu(A,b)\) for every
\(b\in\mathcal B(A)\), together with a budget division attaining each
minimum. Recall that
\[
 \mathcal B(A)=\{|A|,|A|+1,\ldots,
 \min\{2^u,|A|2^{u-\ell}\}\}.
\]
\end{mergingproblemrestatement}

Fix a budget \(b\) for \(A=A_1\uplus A_2\). If we allocate \(x\)
units to \(A_1\) and \(b-x\) units to \(A_2\), the minimum total cost is
\[
 G_b(x)=\mu(A_1,x)+\mu(A_2,b-x).
\]
Thus, by Equation~\eqref{eq:formal-kraft-recurrence}, computing
\(\mu(A,b)\) amounts to minimizing \(G_b(x)\) over the integers \(x\)
for which \(x\in\mathcal B(A_1)\) and \(b-x\in\mathcal B(A_2)\).
To find this minimum efficiently, we use a structural property of
\(G_b\) called \emph{layered convexity}, which is a weaker form of
convexity. We define it formally in Definition~\ref{def:layered-convexity}.

\begin{definition}
\label{def:layered-convexity}
{\itshape
Let \(g\) be a real-valued function on an integer interval.  We call
\(g\) \emph{convex} if its successive differences are nondecreasing;
that is, whenever \(x,x+1,x+2\) belong to its domain,
\[
 g(x+2)-g(x+1)\geq g(x+1)-g(x).
\]
For an integer \(d\geq1\), we call \(g\) \emph{layered convex modulo
\(d\)} if its restriction to every residue class modulo \(d\) is
convex.  Equivalently, whenever \(x,x+d,x+2d\) belong to its domain,
\[
 g(x+2d)-g(x+d)\geq g(x+d)-g(x).
\]
}
\end{definition}

A function can be convex on each residue class without
being convex on its entire domain. Figure~\ref{fig:nonconvex-merge-example}
gives \(G_b(2)=6\), \(G_b(3)=7\), and \(G_b(4)=6\), which violate
ordinary convexity. Figure~\ref{fig:residue-convex-profile} illustrates
the residue classes with modulus \(d=4\): the values in each of the
four classes form a convex sequence.

We prove that \(G_b\) is layered convex modulo \(d=2^{u-\ell-1}\)
by expressing the depth-assignment problem as an equivalent knapsack
problem. We show that the knapsack optimum is concave on each residue
class modulo \(d\), which in turn implies the layered convexity of
\(G_b\).

Our goal is to find an optimal depth assignment for \(A\) with budget
at most \(b\). \textcolor{black}{We begin with an assignment in which every frequency is placed at depth \(u\).} The cost is then
\(u\sum_{f_i\in A}f_i\), and the assignment uses \(|A|\) budget units.
We can reduce this cost by moving frequencies to smaller depths, spending
some of the remaining \(b-|A|\) budget units.

Consider moving \(f_i\) from depth \(u\) to depth \(u-1\). Its cost
decreases by \(f_i\), while its budget contribution increases from
\(1\) to \(2\), using one additional unit. Moving it one more depth
up saves another \(f_i\), but now the budget contribution increases
from \(2\) to \(4\), using two additional units. Continuing in this
way, each one-depth move, which we call an upgrade, saves \(f_i\),
while the additional budget required is successively
\(1,2,4,\ldots,2^{u-\ell-1}\).

We now need to choose upgrades that give the largest total saving
without using more than the remaining \(b-|A|\) budget units.
This suggests a knapsack formulation. Each upgrade becomes an item:
its size is the additional budget it uses, and its value is the
reduction in coding cost. The knapsack capacity is \(b-|A|\), so
maximizing the total item value corresponds to maximizing the saving
within the available budget.

We formalize this knapsack problem in
Problem~\ref{prob:equivalent-knapsack}. We then show that an optimal
selection of items can be converted into a valid depth assignment
with the same saving.

\begin{problem}
\label{prob:equivalent-knapsack}
Given a frequency multiset \(A\) and depth bounds \(\ell,u\),
create, for each \(f_i\in A\), one item of each size
\(1,2,4,\ldots,2^{u-\ell-1}\), all with value \(f_i\).
For a capacity \(c\geq0\), compute the maximum total value
of a subset of these items whose total size is at most \(c\).
Denote this optimum by \(\KnapsackValue{A}{c}\).
\end{problem}

Every depth assignment using at most \(b\) budget units gives a
solution to Problem~\ref{prob:equivalent-knapsack} with capacity
\(b-|A|\): assigning \(f_i\) depth \(u-s\) selects its \(s\) smallest
items. The total item value is exactly the saving in coding cost.

Conversely, a knapsack solution may select items for a frequency
without selecting its smallest items first. If it selects \(s\) items
for \(f_i\), replace them by the \(s\) smallest items for that frequency.
Their value remains \(sf_i\), and their total size cannot increase.
Doing this for every frequency gives a depth assignment with the same
saving and within the original budget. Therefore,
\begin{equation}
\label{eq:mu-knapsack-relation}
 \mu(A,b)=u\sum_{f_i\in A}f_i-\KnapsackValue{A}{b-|A|}.
\end{equation}

Figure~\ref{fig:upgrade-knapsack-interpretation} illustrates this
correspondence.

\input{Sections/upgrade-knapsack-figure}
\FloatBarrier

Lemma~\ref{lem:layered-knapsack-concavity} establishes
the concavity of \(\KnapsackValue{A}{\cdot}\) on each residue class modulo \(d\).

\begin{lemma}
\label{lem:layered-knapsack-concavity}
Let \(A\) be a frequency multiset, and let \(\KnapsackValue{A}{c}\) be the optimum
value of the equivalent knapsack problem with capacity \(c\).  Set
\(d=2^{u-\ell-1}\).  For every \(r\in\{0,\ldots,d-1\}\), the sequence
\[
 \KnapsackValue{A}{r},\KnapsackValue{A}{r+d},\KnapsackValue{A}{r+2d},\ldots
\]
is concave.
\end{lemma}

\begin{proof}
Consider the upgrade items for \(A\), whose sizes are powers of two
up to \(d\).  Fix a finite prefix of the capacities in the sequence and
add enough zero-value items of size one to fill every capacity exactly.
These dummy items do not change any optimum, so it is enough to consider
exact capacities.

If \(d=1\), then \(u=\ell+1\), so each frequency has only one
possible upgrade, from depth \(u\) to depth \(\ell\). This upgrade
uses one additional budget unit, so all knapsack items have size one.
For each capacity, we select the most valuable items, so successive increases in
the optimum are nonincreasing. This proves the claim in this case.
Henceforth assume \(d\geq2\).

Fix a residue \(r\) modulo \(d\), and process the item sizes from small
to large.  Start with the size-one items.  Every larger item has even
size, so the parity of \(r\) determines the parity of the number of
selected size-one items.  If \(r\) is odd, an optimal solution must select
an odd number of size-one items.  The most valuable size-one item is
therefore included in an optimal solution: replacing any selected
size-one item by it preserves the size and does not decrease the value.
We fix this item as selected.  If \(r\) is even, no size-one item is
forced.

After removing the fixed item, if any, the number of selected
size-one items must be even. We sort the remaining items by
nonincreasing value. For any fixed number of selected items,
the maximum value is obtained by choosing the most valuable ones.
We therefore pair consecutive items in this order and replace
each pair by a size-two package whose value is the sum of the
two item values. Selecting the first $t$ packages is equivalent
to selecting the first $2t$ items, which preserves the optimum.
Any unpaired last item can be discarded, since only an even
number of these items can be selected.

If \(2<d\), we next process size two. Combine the original size-two
items with the packages just created and sort them together by
nonincreasing value. The remaining capacity modulo four is determined
by \(r\) and is either zero or two. In the latter case, we fix the most
valuable size-two item as selected, by the same exchange argument; in
the former case, no item is fixed. Pair the remaining items in sorted
order, discard any unpaired last item, and pass the resulting packages
to size four. Repeat this procedure at each size smaller than \(d\).

The items fixed as selected depend only on \(r\), and their total size
is exactly \(r\): at each size, the procedure accounts for the
corresponding binary digit of \(r\). The other smaller items are
represented by packages; their selection is not yet fixed and may
vary with the capacity.

After these steps, all remaining items have size \(d\).  Let their values
in nonincreasing order be \(p_1\geq p_2\geq\cdots\).  For capacity
\(r+td\), the fixed items use \(r\) units,
leaving \(td\) units for these size-\(d\) items. The optimal solution
therefore adds the first \(t\) of them.  Thus,
increasing the capacity from \(r+td\) to \(r+(t+1)d\) increases the
optimum by \(p_{t+1}\).  These increments are nonincreasing, so
\(\KnapsackValue{A}{r},\KnapsackValue{A}{r+d},\KnapsackValue{A}{r+2d},\ldots\) is concave.
\end{proof}

Fix a frequency multiset \(A\) and a residue \(r\) modulo \(d\), and
write \(x_t=r+td\).  Lemma~\ref{lem:layered-knapsack-concavity} gives
\[
 \KnapsackValue{A}{x_{t+2}}-\KnapsackValue{A}{x_{t+1}}
 \;\leq\;
 \KnapsackValue{A}{x_{t+1}}-\KnapsackValue{A}{x_t}.
\]
By Equation~\eqref{eq:mu-knapsack-relation}, each value of
$\mu$ is equal to the same baseline cost minus the corresponding
knapsack optimum. When taking differences, the baseline costs
cancel, leaving
\begin{align*}
 \mu(A,x_{t+2}+|A|)-\mu(A,x_{t+1}+|A|)
 &\quad=-\bigl(\KnapsackValue{A}{x_{t+2}}-\KnapsackValue{A}{x_{t+1}}\bigr)\\
 &\quad\geq-\bigl(\KnapsackValue{A}{x_{t+1}}-\KnapsackValue{A}{x_t}\bigr)\\
 &\quad=\mu(A,x_{t+1}+|A|)-\mu(A,x_t+|A|).
\end{align*}
Thus, \(\mu(A,\cdot)\) is convex on every residue class modulo \(d\).
Now fix a budget \(b\) for \(A_1\uplus A_2\) and a residue \(r\) of feasible divisions.
The first term of \(G_b(x_t)\) is convex in \(t\),
and the second term is convex in reverse order because its argument
decreases by \(d\).  Combining these two convexity inequalities gives
\begin{align*}
 G_b(x_{t+2})-G_b(x_{t+1})
 &=\bigl[\mu(A_1,x_{t+2})-\mu(A_1,x_{t+1})\bigr]
   +\bigl[\mu(A_2,b-x_{t+2})-\mu(A_2,b-x_{t+1})\bigr]\\
 &\geq\bigl[\mu(A_1,x_{t+1})-\mu(A_1,x_t)\bigr]
   +\bigl[\mu(A_2,b-x_{t+1})-\mu(A_2,b-x_t)\bigr]\\
 &=G_b(x_{t+1})-G_b(x_t).
\end{align*}
Therefore, \(G_b\) is convex on every residue class modulo \(d\), which
is precisely layered convexity.

\begin{lemma}
\label{lem:layered-convexity}
For every merge \(A=A_1\uplus A_2\) and every budget \(b\) for \(A\), the
function \(G_b\) is layered convex modulo
\(d=2^{u-\ell-1}\).
\end{lemma}

\textcolor{black}{Layered convexity lets us minimize \(G_b\) by searching
each residue class separately and taking the smallest of their minima.
Lemma~\ref{lem:fast-table-merge} describes the search and bounds its
\MPC resources, including access to the distributed child tables.}

\begin{lemma}
\label{lem:fast-table-merge}

For every \(\eta\geq2\), the table for \(A=A_1\uplus A_2\) can be
computed in
\[
 O\!\left((1+\log_\eta|A|)
          (1+\log_\eta\min\{|A_1|,|A_2|\})\right)
\]
\MPC rounds using \(O(\eta\lambda^{2\log_\varphi2})\) local memory.
Total memory and computation are each bounded by
\[
 O\!\left(|A|\eta\lambda^{4\log_\varphi2}
 (1+\log_\eta|A|)^2\log(2|A|\eta\lambda)\right).
\]

\end{lemma}

\begin{proof}

Write \(\alpha=\log_\varphi2\). We first count
how many search steps are needed, and then explain how to carry out one
step in \MPC. 
Fix a budget $b$. \textcolor{black}{We seek a feasible budget division between}
$A_1$ and $A_2$ that minimizes
\[
G_b(x)=\mu(A_1,x)+\mu(A_2,b-x).
\]
By Lemma~\ref{lem:layered-convexity}, $G_b$ is convex within
each residue class. We can therefore find the minimum in each
class separately and return the smallest value among them.

We first bound the number of feasible budget divisions in
each residue class. The feasible values of $x$ span an interval
of length $O(\min\{|A_1|,|A_2|\}2^{u-\ell})$.
Within a residue class, consecutive values of $x$ differ by
$d=2^{u-\ell-1}$. Thus, each residue class contains
$O(\min\{|A_1|,|A_2|\})$ feasible divisions.
At each step, we select $O(\eta)$ evenly spaced positions
in the current search interval, including its endpoints,
and evaluate $G_b$ at these positions.
We then find the sampled position with the smallest value.
By convexity, an optimum must lie between the two samples
immediately before and after this position.
If the minimum occurs at an endpoint, we keep the interval
between that endpoint and its only neighboring sample.
In case of ties, we choose the first sampled minimum.

Since the samples are evenly spaced, the remaining interval
contains only an $O(1/\eta)$ fraction of the previous
positions. \textcolor{black}{We repeat this process until at most $O(\eta)$
positions remain, at which point we evaluate all of them
and return the minimum.}
Since each step reduces the search interval by a factor
of order $\eta$, and each residue class initially contains
$O(\min\{|A_1|,|A_2|\})$ feasible divisions,
the search takes
$O(1+\log_\eta\min\{|A_1|,|A_2|\})$ steps.
We perform these searches in parallel for all budgets
and residue classes.

\textcolor{black}{To evaluate $G_b$ at a position $x$, we need the two table}
entries $\mu(A_1,x)$ and $\mu(A_2,b-x)$.
For each budget $b$, a search step examines $O(\eta)$
\textcolor{black}{positions in each of the $d$ residue classes.}
\textcolor{black}{The bound \(u-\ell\leq2\Delta_\lambda\) following
Corollary~\ref{cor:bounded-ratio-depth-range} gives
\(d=O(\lambda^{2\alpha})\).}
Thus, each search step examines
\textcolor{black}{$O(\eta\lambda^{2\alpha})$ positions per budget.}
We assign one machine to each budget $b$.
This machine maintains the searches across all residue
classes using $O(\eta\lambda^{2\alpha})$ local memory.
The required child-table entries are retrieved
during each search step.

We now bound the amount of data processed in one search step.
The table for $A$ contains $O(|A|2^{u-\ell})$ budget entries.
For each budget, we examine
\textcolor{black}{$O(\eta\lambda^{2\alpha})$ positions, each requiring}
two values from the child tables.
Since $2^{u-\ell}=O(\lambda^{2\alpha})$,
the total number of records needed to evaluate all these
\textcolor{black}{positions is}
\[
\RecordCount{}
=O\left(|A|\eta\lambda^{4\alpha}\right).
\]
The child tables are distributed across the machines.
A machine may need values stored on other machines, and several
machines may need the same table entry. To handle these accesses,
we first sort the requests to group those referring to the same
entry. We then distribute each entry's value to all requests
in its group and send the answers back to the corresponding
machines. \textcolor{black}{Each machine can then evaluate the sampled positions}
and continue searching within the smaller intervals.

With local memory \(\LocalMemorySize=\Theta(\eta\lambda^{2\alpha})\),
we implement sorting by Goodrich's BSP sorting algorithm
\cite{goodrich1999sorting}, using the simulation described in
\cite{goodrich2011sorting}. Sorting \(\RecordCount{}\) records takes
\(O(1+\log_{\LocalMemorySize}\RecordCount{})\) rounds and
\(O(\RecordCount{}\log(2\RecordCount{}))\) total internal computation.
The simulation preserves these local computations and adds only
linear work in the records communicated in each round.

The remaining operations, including distributing table
values to their requests and returning the answers,
can be performed using standard sorting and scan
primitives. These operations fit within the same
round and computation bounds. In particular,
repeated requests are handled without exceeding
the communication capacity of any machine.

\textcolor{black}{Including evaluations at the sampled positions}, one search step takes
\[
O(1+\log_{\LocalMemorySize}\RecordCount{})
=O(1+\log_\eta|A|)
\]
rounds and
\[
O\!\left(
\RecordCount{}\log(2\RecordCount{})
(1+\log_{\LocalMemorySize}\RecordCount{})
\right)
\]
total computation. The simultaneous total memory is
\[
O\!\left(
\RecordCount{}
(1+\log_{\LocalMemorySize}\RecordCount{})
\right).
\]

The searches for all budgets run in parallel. Multiplying the rounds
per step by the number of search steps gives
\[
 O\!\left((1+\log_\eta|A|)
          (1+\log_\eta\min\{|A_1|,|A_2|\})\right)
\]
rounds. For total computation, we sum over these steps and use
\(\min\{|A_1|,|A_2|\}\leq|A|\),
\(\RecordCount{}=O(|A|\eta\lambda^{4\alpha})\), and
\(\log(2\RecordCount{})=O(\log(2|A|\eta\lambda))\).
The resulting bound is
\[
 O\!\left(|A|\eta\lambda^{4\log_\varphi2}
 (1+\log_\eta|A|)^2\log(2|A|\eta\lambda)\right).
\]
This also bounds simultaneous total memory, even if the records from
all search steps are retained. Finally, each budget selects the minimum
among its residue classes locally, within the same bounds.

\end{proof}
We now use Lemma~\ref{lem:fast-table-merge} to analyze the
bottom-up construction in
Algorithm~\ref{alg:bottom-up-budget-tables}.
We start with the tables for $B_1,\ldots,B_q$.
At each stage, we merge neighboring pairs of tables in parallel
and carry any unpaired table to the next stage.
We repeat until only the table for $A$ remains.
Since the number of tables is roughly halved at each stage,
there are $\lceil\log_2 q\rceil
=O(\log(2+\log\lambda))$ stages.

To analyze one stage, note that all merges run in parallel.
Hence, its round complexity is determined by the slowest merge.
Moreover, the frequency multisets involved in these merges
are disjoint, and their total size is at most $m$.
We can therefore add their memory and computation bounds.
Summing these costs over the stages and applying
Lemma~\ref{lem:fast-table-merge}
gives the bounds in
Lemma~\ref{lem:bounded-ratio-huffman-formal}.

\begin{lemma}
\label{lem:bounded-ratio-huffman-formal}
For every \(\eta\geq2\), Bounded-Ratio Huffman Coding on \(m\)
frequencies takes
\[
 O\bigl((1+\log_\eta m)^2\log(2+\log\lambda)\bigr)
\]
\MPC rounds and uses
\(O(\eta\lambda^{2\log_\varphi2})\) local memory per machine.  Its
total memory and computation are
\[
 O\!\left(m\eta\lambda^{4\log_\varphi2}
 (1+\log_\eta m)^2\log(2m\eta\lambda)
 \log(2+\log\lambda)\right).
\]
The algorithm returns both the optimum value and a corresponding
optimal depth vector.
\end{lemma}

The base tables of Lemma~\ref{lem:constant-time-base-case}
fit these bounds by taking local memory
\(\sigma=\Theta(\eta\lambda^{2\log_\varphi2})\) and summing their
budget ranges over the blocks. The final depth reconstruction also
fits within the stated round bound.

%% file: Sections/upgrade-knapsack-figure.tex
\begin{figure}[t]
  \centering
  \begin{tikzpicture}[
      x=0.9cm,
      y=0.78cm,
      depthline/.style={draw=black,line width=0.45pt},
      leaf/.style={circle,draw=black,fill=black!12,
        minimum size=0.58cm,inner sep=0pt,font=\small},
      ghost/.style={circle,draw=black,dashed,fill=white,
        minimum size=0.52cm,inner sep=0pt},
      move/.style={-{Stealth[length=2.2mm]},draw=black,
        line width=0.85pt},
      item/.style={rounded corners=1.5pt,draw=black,
        minimum width=1.45cm,minimum height=0.68cm,
        align=center,font=\scriptsize},
      selected/.style={item,draw=black,fill=black!16},
      unselected/.style={item,fill=gray!5,text=black},
      note/.style={font=\scriptsize,text=black}
    ]

    \node[font=\small\bfseries] at (2.75,4.45) {Depth assignment};
    \foreach \y/\lab in {0/{u},1/{u-1},2/{u-2},3/{u-3}} {
      \draw[depthline] (0.45,\y) -- (5.75,\y);
      \node[font=\scriptsize,anchor=east] at (0.28,\y) {depth $\lab$};
    }
    \draw[-{Stealth[length=2mm]},draw=black] (0.62,3.35) -- (0.62,4.02);
    \node[note,anchor=west] at (0.78,3.68) {toward the root};

    \node[leaf] (f1) at (1.35,0) {$f_1$};
    \node[ghost] (f2start) at (3.05,0) {};
    \node[leaf] (f2) at (3.05,1) {$f_2$};
    \draw[move] (f2start.north) -- node[right,note] {$+1$} (f2.south);

    \node[ghost] (f3start) at (4.9,0) {};
    \coordinate (f3a) at (4.9,1);
    \coordinate (f3b) at (4.9,2);
    \node[leaf] (f3) at (4.9,3) {$f_3$};
    \draw[move] (f3start.north) -- node[right,note] {$+1$} (f3a);
    \draw[move] (f3a) -- node[right,note] {$+2$} (f3b);
    \draw[move] (f3b) -- node[right,note] {$+4$} (f3.south);
    \node[note,align=center] at (2.85,-0.62)
      {labels show additional Kraft units};

    \node[font=\small\bfseries] at (10.05,4.45) {Upgrade items};
    \node[note] at (8.25,3.72) {size $1$};
    \node[note] at (10.05,3.72) {size $2$};
    \node[note] at (11.85,3.72) {size $4$};

    \foreach \row/\freq in {2.85/{f_1},1.75/{f_2},0.65/{f_3}} {
      \node[font=\small,anchor=east] at (7.25,\row) {$\freq$};
    }
    \node[unselected] at (8.25,2.85) {value $f_1$};
    \node[unselected] at (10.05,2.85) {value $f_1$};
    \node[unselected] at (11.85,2.85) {value $f_1$};

    \node[selected] at (8.25,1.75) {value $f_2$};
    \node[unselected] at (10.05,1.75) {value $f_2$};
    \node[unselected] at (11.85,1.75) {value $f_2$};

    \node[selected] at (8.25,0.65) {value $f_3$};
    \node[selected] at (10.05,0.65) {value $f_3$};
    \node[selected] at (11.85,0.65) {value $f_3$};

    \draw[decorate,decoration={brace,mirror,amplitude=4pt},
      draw=black] (7.48,0.15) -- (12.62,0.15);
    \node[note,anchor=north] at (10.05,-0.12)
      {selected items form a prefix in each row};
  \end{tikzpicture}
  \caption{A depth assignment as a collection of upgrade prefixes.
  Moving frequency $f_i$ upward for the $t$-th time consumes
  $2^{t-1}$ additional Kraft units and saves $f_i$ in encoded length.
  Thus, assigning $f_i$ to depth $u-s$ corresponds to selecting its
  first $s$ upgrade items.}
  \label{fig:upgrade-knapsack-interpretation}
\end{figure}
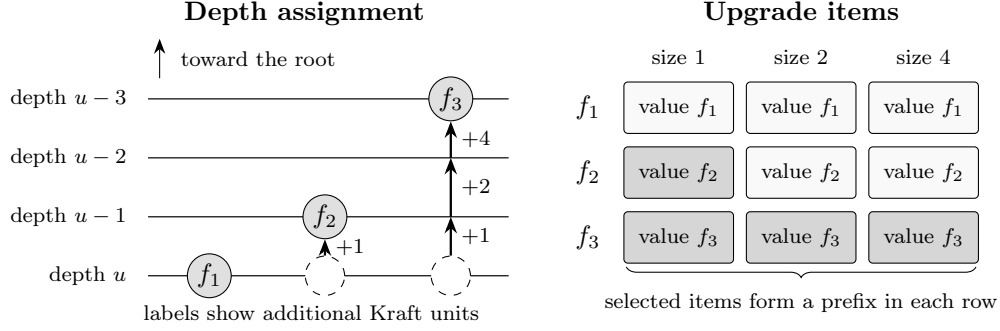

%% file: Sections/iterative-state-reconstruction.tex
\section{Iterative-State Reconstruction}
\label{sec:iterative-state-formal}

We now turn to the second subtask, Iterative-State Reconstruction.
Our goal is to determine the state that Huffman's algorithm
would reach after performing all merges up to a given threshold,
without actually simulating these merges.

The input consists of \textcolor{black}{a frequency multiset $A$ with sorted entries}
$f_1\leq\cdots\leq f_m$ satisfying $f_m/f_2\leq\lambda$,
an optimal depth assignment $h_1,\ldots,h_m$ for these
frequencies, and a threshold $\tau>f_m$.
We consider running Huffman's algorithm on these frequencies
and stop just before the first merge involving a frequency
strictly greater than $\tau$.
We seek to recover the frequencies remaining at this point,
together with the total cost of the merges performed so far.
If no such merge occurs, we return the final state
and the cost of all merges.

For the analysis, let $h_{\min}$ and $h_{\max}$ denote
the minimum and maximum leaf depths, respectively,
and let $\Delta=h_{\max}-h_{\min}$.

\begin{example}
\label{ex:iterative-state-reconstruction}
To illustrate the input and the required output, consider the frequencies
\((2,2,3,3,4,4,5,5)\), the optimal depth vector
\((4,4,3,3,3,3,3,2)\), and the threshold \(\tau=6\).
Figure~\ref{fig:iterative-state-example} shows a corresponding Huffman
tree.  The depths are part of the input; our task is to recover the
\textcolor{black}{frequencies} and the accumulated cost at the stopping point.
Huffman's algorithm first performs the merges
\(2+2=4\), \(3+3=6\), \(4+4=8\), \(4+5=9\), and
\(5+6=11\).  Each merge uses two \textcolor{black}{frequencies} at most \(\tau\), even
though its resulting \textcolor{black}{frequency} may exceed \(\tau\).  The \textcolor{black}{frequencies}
are now \(8,9,11\).  The next merge would use \(8\) and \(9\), so
we stop before it.  The required output is therefore the multiset
\(\{8,9,11\}\) and the cost \(4+6+8+9+11=38\).
\end{example}
\input{Sections/iterative-state-example-figure}

A direct way to obtain this state is to run Huffman's algorithm,
building the tree one merge at a time. We stop when the next merge
would use a \textcolor{black}{frequency} greater than \(\tau\), and return the \textcolor{black}{frequencies} and the accumulated cost. However, dependencies between merges
make this simulation sequential and costly in $\MPC$ rounds. We instead
use the given leaf depths to reconstruct the stopping state directly,
without simulating the intervening merges one by one.

{\color{black}
We next show that the threshold $\tau$ allows us to narrow
the depth we need to recover to at most two possibilities.

Let $\TotalFrequency{A}=\sum_{f_i\in A} f_i$ denote the
total frequency. At every depth $h\leq h_{\min}$,
the Huffman tree has exactly $2^h$ nodes.
Their frequencies sum to $\TotalFrequency{A}$,
and hence their average frequency is
$\TotalFrequency{A}/2^h$.
This observation helps us locate a depth containing a node
of frequency at most $\tau$.
If the average frequency at depth $h$ is at most $\tau$,
at least one node at that depth has frequency at most $\tau$.
We therefore choose
\[
\AveragingDepth=\min\left\{
h_{\min},
\max\left\{
0,\left\lceil
\log_2\left(\frac{\TotalFrequency{A}}{\tau}\right)
\right\rceil
\right\}
\right\}.
\]
By this choice, depth $\AveragingDepth$ contains a node of frequency
at most $\tau$. This follows from the average when
$\AveragingDepth<h_{\min}$, and from the existence of an input leaf
when $\AveragingDepth=h_{\min}$, since $\tau>f_m$.

It remains to rule out depths smaller than $\AveragingDepth-1$.
When $\AveragingDepth\geq 1$, the average frequency at depth $\AveragingDepth-1$
is greater than $\tau$.
By the ordering property of Huffman trees, every node
at depth $\AveragingDepth-2$ must then have frequency greater
than $\tau$, whenever that depth exists.
The same holds for all smaller depths.

Consequently, the smallest depth containing a node
of frequency at most $\tau$ is either $\AveragingDepth-1$ or $\AveragingDepth$
(and is $0$ when $\AveragingDepth=0$).
Recovering the frequencies at these depths is enough to obtain the
stopping state. In fact, recovering depth \(\AveragingDepth\) suffices, since
consecutive pair sums give the frequencies at \(\AveragingDepth-1\). We therefore
focus on the following task: given a depth
\(h^*\in\{0,\ldots,h_{\min}\}\), recover the frequencies of all
nodes at that depth.

}

Since the given depths form an optimal binary code,
Kraft equality holds:
\[
\sum_{i=1}^m 2^{-h_i}=1.
\]
Moreover, no leaf has depth smaller than $h^*$.
Thus, there are exactly $q=2^{h^*}$ nodes at depth $h^*$.
These nodes partition the leaves into $q$ subtrees,
each with total Kraft value $2^{-h^*}$.

We next show how to determine the frequencies of these
subtrees without constructing the tree.
Using the Package-Merge method of Larmore and
Hirschberg~\cite{larmore1990package}, we show that
there is an optimal tree with the given leaf depths
that satisfies the following property.
Let $w_1\leq\cdots\leq w_q$ be the frequencies of
the nodes at depth $h^*$, in nondecreasing order.
For every $r$, the leaves in the $r$ lightest subtrees
have minimum total frequency among all subsets of
input leaves whose total Kraft value is $r2^{-h^*}$.
Example~\ref{ex:kraft-prefixes} illustrates this property.

\begin{example}
\label{ex:kraft-prefixes}
Figure~\ref{fig:iterative-state-example} has four nodes
at depth \(h^*=h_{\min}=2\), with \textcolor{black}{frequencies} \(5,6,8,9\). Their
subtrees contain the following input leaves:
\begin{itemize}
\item The node of \textcolor{black}{frequency} \(5\) is itself a leaf at depth \(2\),
with Kraft value \(1/4\).
\item The node of \textcolor{black}{frequency} \(6\) has two leaves of frequency \(3\)
at depth \(3\). Their total frequency is \(3+3=6\), and their
total Kraft value is \(1/8+1/8=1/4\).
\item The node of \textcolor{black}{frequency} \(8\) has two leaves of frequency \(2\)
at depth \(4\) and one leaf of frequency \(4\) at depth \(3\).
Their total frequency is \(2+2+4=8\), and their total Kraft value
is \(1/16+1/16+1/8=1/4\).
\item The node of \textcolor{black}{frequency} \(9\) has leaves of frequencies \(4\)
and \(5\), both at depth \(3\). Their total frequency is
\(4+5=9\), and their total Kraft value is \(1/8+1/8=1/4\).
\end{itemize}
Taking the first one, two, three, or four subtrees therefore gives
total Kraft values \(1/4,1/2,3/4,1\) and minimum total frequencies
\(5,11,19,28\), respectively. In particular, for target \(1/2\),
we take the leaf of frequency \(5\) at depth \(2\) and the two
leaves of frequency \(3\), obtaining total frequency \(11\).
\end{example}

This property suggests computing the minimum total frequency for each
such target. For \(r\in\{0,\ldots,q\}\), define
\begin{equation}
\label{eq:prefix-package-value}
 P(r)=\min\left\{
   \sum_{f_i\in S}f_i:
   S\text{ is a submultiset of the input and }
   \sum_{f_i\in S}2^{-h_i}=r2^{-h^*}
 \right\},
 \qquad P(0)=0.
\end{equation}

{\color{black}
The Package-Merge algorithm of Larmore and
Hirschberg~\cite[Section~2]{larmore1990package} solves the following
selection problem. Each item has a width that is a power of two and
a cost. Given a target width, the goal is to select a subset whose
widths sum exactly to the target and whose total cost is minimum.
In our setting, each leaf \(i\) is an item of width \(2^{-h_i}\)
and cost \(f_i\), and the target is \(r2^{-h^*}\). The optimum
of this instance is precisely \(P(r)\).

\textcolor{black}{Lemma~\ref{lem:kraft-prefixes} follows from
this construction. We include the argument to show that one tree
realizes the prefix identities for all targets simultaneously.}
}

\begin{lemma}
\label{lem:kraft-prefixes}
\textcolor{black}{Fix an integer depth \(0\leq h^*\leq h_{\min}\).}
\textcolor{black}{There is an optimal tree with the given leaf depths
whose node frequencies \(w_1\leq\cdots\leq w_q\) at depth \(h^*\)
satisfy}
\[
 P(r)=\sum_{j=1}^r w_j
 \quad\text{and hence}\quad
 w_r=P(r)-P(r-1).
\]
\end{lemma}

{\color{black}
\begin{proof}
We show that, for every $r$, the total frequency of the
$r$ lightest subtrees rooted at depth $h^*$ equals $P(r)$.
Our proof follows the Package-Merge method of
Larmore and Hirschberg~\cite{larmore1990package}.
To apply Package-Merge, we first convert the Kraft values
into integers by multiplying them by $2^{h_{\max}}$.
For each leaf $i$, we therefore create an item with
cost $f_i$ and integer width $2^{h_{\max}-h_i}$.
The leaves of any subtree rooted at depth $h^*$
have total Kraft value $2^{-h^*}$.
Thus, their corresponding items have total width
\[
c=2^{h_{\max}-h^*}.
\]
To compute $P(r)$, we therefore seek a minimum-cost
subset of items whose total width is exactly $rc$.

We start with the items of width $1$, assuming $c>1$.
Since $c$ is a power of two, the target width $rc$ is even.
All other items have even widths.
Therefore, the width-$1$ items must also contribute an even
total width, meaning that an even number of them must be selected.

If we select a fixed number of width-$1$ items, we may choose
the cheapest ones: this can only decrease their total cost
while leaving their total width unchanged.
Arrange the width-$1$ items in nondecreasing cost order.
Since an even number of these items must be selected,
an optimal selection takes an even-length prefix of this list.
We can therefore pair consecutive items and replace each pair
with a package of width $2$, whose cost is the sum of the
two item costs. Selecting the first $t$ packages is equivalent
to selecting the first $2t$ original items.
Hence, this replacement preserves the optimum.

We now have items that originally had width \(2\), together with
new packages of width \(2\). If \(c>2\), we apply the same argument
to their combined list. Since \(c\) is a power of two greater than
\(2\), the required total width \(rc\) is a multiple of \(4\).
Every item wider than \(2\) also has width divisible by \(4\).
Thus the selected width-\(2\) items must together contribute a
multiple of \(4\), which requires an even number of them.
We arrange the original items and the new packages together in
nondecreasing cost order. As before, we group consecutive entries
into pairs and replace each pair by a single package: its width is
\(2+2=4\), and its cost is the sum of the two costs.
We continue with widths \(4,8,\ldots,c/2\). At each step, all smaller
widths have already been replaced by packages, so the same divisibility
argument applies.

The pairing is independent of $r$. We use the same ordering
of items, including ties, for every target.
\textcolor{black}{Every item is included in these pairings.
Indeed, Kraft equality gives total width $2^{h_{\max}}$, which is
divisible by twice the current width at every step below $c$.
All larger items also have widths divisible by twice the current
width. Since no smaller items remain, the number of items of the
current width is even. Thus every item can be paired, and no leaf
is discarded.}

After processing all widths below $c$, every remaining item
has width $c$. Each package was formed by pairing two items
of the same width, with cost equal to the sum of their costs.
We can view each such pairing as joining two nodes under a
common parent. Thus, the final items correspond to the nodes
at depth $h^*$ of a tree constructed by pairing consecutive
nodes in nondecreasing frequency order.
Continuing this process to the root produces a full binary
tree with the prescribed leaf depths. Since these depths are
optimal, the resulting tree is also optimal.

Let $w_1\leq\cdots\leq w_q$ be the frequencies of the nodes
at depth $h^*$. These are precisely the costs of the final
items, in nondecreasing order. 

Each final item has width $c$. Therefore, selecting items
of total width $rc$ amounts to selecting exactly $r$ of them.
The minimum cost is obtained by choosing the $r$ cheapest
items. Since the pairings preserve the optimum, we obtain
\[
P(r)=w_1+\cdots+w_r.
\]
Taking the difference between consecutive values gives
\[
w_r=P(r)-P(r-1).
\]\end{proof}
}

\textcolor{black}{The tree constructed in the proof pairs consecutive nodes in frequency order at each depth.} In an optimal tree, a deeper node cannot
have larger frequency than a shallower node: exchanging disjoint
subtrees would otherwise decrease the cost, while a descendant has
smaller frequency than its ancestor. Together with consecutive
pairing, this allows the sibling pairs to be processed in Huffman
order, with ties resolved consistently. This also justifies the
ordering between depths used in the averaging argument.

By Lemma~\ref{lem:kraft-prefixes}, it is enough to compute
$P(0),\ldots,P(q)$ to recover the node frequencies at
depth $h^*$. The remaining challenge is to compute these
values efficiently in the \MPC model. Each optimization
problem involves all input leaves, which may not fit
on a single machine. In
Section~\ref{sec:computing-prefix-optima}, we show how
to reduce these problems to smaller instances that
can be solved on individual machines.

\subsection{Computing the values \texorpdfstring{$P(r)$}{P(r)}}
\label{sec:computing-prefix-optima}

{\color{black}
We compute all values \(P(r)\) without placing the complete input on
one machine. For each \(r\), recall that \(P(r)\) is the minimum
total frequency of a selection of leaves whose Kraft values sum to
\(r2^{-h^*}\). We first relax this selection problem by allowing
fractions of leaves, while keeping the same target Kraft value and
objective. We then show that some
optimal integral solution differs from this fractional solution on only
a small number of leaves. Using this bound, we identify a small set of
leaves containing all necessary changes and fix the choices outside
that set. The remaining problem fits on one machine and can be solved
exactly. We carry out these computations for all values of \(r\) in
parallel.
}

\textcolor{black}{For a fixed \(r\), the fractional relaxation seeks a
selection of total Kraft value exactly \(r2^{-h^*}\) and minimum
total frequency, allowing each leaf to be selected partially.
Selecting a fraction \(x_i\in[0,1]\) of leaf \(i\) contributes
\(x_if_i\) to the total frequency and \(x_i2^{-h_i}\) to the
required Kraft value.} For example, selecting half of a
leaf with frequency \(4\) and depth \(3\) contributes cost \(2\)
and Kraft value \(1/16\).

Computing the optimum of this relaxation is indeed a fractional
knapsack problem: we may select a fraction of each leaf, and we seek
a selection of total Kraft value exactly \(r2^{-h^*}\) and
minimum total cost.
The fractional knapsack problem admits a simple greedy algorithm.
We order the leaves by nondecreasing ratio \(f_i/2^{-h_i}\), which
is the cost per unit of Kraft value. We then select leaves in this
order, taking each leaf completely as long as its Kraft contribution
does not exceed the amount still needed to reach \(r2^{-h^*}\).
If the next leaf would exceed the target, we take only the fraction
needed to reach it exactly and stop. This gives an optimal fractional
solution with at most one partially selected leaf. We sort the leaves
only once and use this order for every value of \(r\).
Figure~\ref{fig:fractional-greedy-cutoffs} illustrates the resulting
cutoffs.

An interesting property of the fractional solution is that it differs
from some optimal integral solution on only a small number of items.
Lemma~\ref{lem:fractional-proximity} makes this bound precise.

\begin{lemma}
\label{lem:fractional-proximity}
Fix \(r\in\{0,\ldots,q\}\), and let \(R\) be the set of leaves
selected completely by the fractional greedy solution for target
\(r2^{-h^*}\). There is a set of leaves \(S\) attaining
\(P(r)\) such that
\[
 |R\mathbin\triangle S|\leq2^{\Delta+1}-2,
 \qquad \Delta=h_{\max}-h_{\min}.
\]
\end{lemma}

\begin{proof}
If the fractional solution is integral, it is also an optimal
integral solution, and the claim follows immediately.
Otherwise, exactly one leaf is selected fractionally.
We focus on this case.
Let \(R\) be the leaves selected completely before this fractional
leaf. Set $c=2^{h_{\max}-h^*}$. The corresponding scaled
target Kraft value is $rc$. We write
\[
 \sum_{i\in R}2^{h_{\max}-h_i}=rc-\delta,
 \qquad 1\leq\delta<2^\Delta.
\]

The scaled Kraft value of each leaf is at most
$2^{h_{\max}-h_{\min}}=2^\Delta$.
Since $\delta$ is smaller than the scaled Kraft value
of the partially selected leaf, we also have
$\delta<2^\Delta$.

Among all sets of leaves attaining \(P(r)\), choose one, denoted by \(S\),
that minimizes \(|R\mathbin\triangle S|\).  Let
\(R^-=R\setminus S\) be the leaves removed from the greedy prefix and let
\(R^+=S\setminus R\) be the leaves added to it.  These two sets satisfy
\[
 \sum_{i\in R^+}2^{h_{\max}-h_i}-\sum_{i\in R^-}2^{h_{\max}-h_i}=\delta.
\]

We claim that there are no nonempty subsets
$R^-_0\subseteq R^-$ and $R^+_0\subseteq R^+$
with equal total Kraft value.
Suppose such subsets exist.
Every leaf in $R$ has frequency-to-Kraft ratio
at most that of any leaf outside $R$.
We can therefore replace $R^+_0$ in $S$ with $R^-_0$
without increasing the cost.
Since the two subsets have equal total Kraft value,
the replacement also preserves feasibility.
The resulting solution is optimal and closer to $R$,
contradicting our choice of $S$.

It remains to recall a simple fact about such a minimal zero-sum
sequence.  Its terms are nonzero integers in \([-2^\Delta,2^\Delta]\).  Start with \(-\delta\), then order the remaining terms
by choosing a negative term when the current partial sum is
positive and a positive term when it is negative.  Every proper partial
sum then lies in \([-2^\Delta+1,2^\Delta-1]\).  No proper partial sum is zero, and no
two proper partial sums are equal, since either event would identify
a proper zero-sum subsequence.  There are only \(2^{\Delta+1}-2\) possible
nonzero values in this interval.  The signed sequence therefore has at
most \(2^{\Delta+1}-1\) terms.  Hence
\(|R^-|+|R^+|+1\leq2^{\Delta+1}-1\), which proves
Lemma~\ref{lem:fractional-proximity}.
\end{proof}

By Lemma~\ref{lem:fractional-proximity}, the fractional solution and
an optimal integral solution for \(P(r)\) can differ on only a bounded
number of leaves. However, the lemma does not identify these leaves. Furthermore,
 checking all leaves would be too expensive. To resolve this, we narrow the search
by grouping the leaves according to their Kraft value
\(2^{-h_i}\), or equivalently according to their depth.
\textcolor{black}{Figure~\ref{fig:fractional-proximity-window} illustrates
these groups and the small windows of leaves we retain around their
fractional cutoffs.} Within each
group, the leaves have equal Kraft value, so the fractional solution
orders them by nondecreasing frequency and selects a prefix. We may also assume that an optimal integral solution selects
a prefix of each group. Indeed, replacing any selected leaf
with a cheaper unselected leaf from the same group preserves
the total Kraft value and does not increase the cost. This replacement
also cannot increase the number of differences from the fractional
prefix.

Fix \(r\) and a group of \(g\) leaves, and let \(t\) be the number
of leaves in this group selected completely by the fractional solution
for \(P(r)\). These leaves occupy positions \(1,\ldots,t\) in
frequency order, and we call the boundary after position \(t\) the
cutoff. By Lemma~\ref{lem:fractional-proximity}, some optimal integral
solution has a prefix length differing from \(t\) by at most
\(2^{\Delta+1}-2\). Thus, its cutoff moves forward or backward
by at most this many positions. We retain the leaves at positions
\[
 \max\{1,t-2^{\Delta+1}+1\},\ldots,
 \min\{g,t+2^{\Delta+1}\}.
\]
Leaves before this window remain
selected, and leaves after it remain unselected. Only the leaves
inside the windows require further decisions. Figure~\ref{fig:fractional-proximity-window}
illustrates these windows.

For each target \(r2^{-h^*}\), we subtract the total Kraft value of
the fixed selected leaves to obtain the residual target. The remaining instance contains
\(O((\Delta+1)2^\Delta)\) leaves and can be gathered on one machine
with this much local memory.
We send each depth group's retained leaves in frequency
order. The residual instance therefore has equal-Kraft groups already
sorted by cost, as required by the linear-time implementation of
Package-Merge~\cite[Section~2.2]{larmore1990package}.
Using Package-Merge~\cite{larmore1990package},
we solve this residual instance in time linear in the number of leaves
it contains. Adding the total frequency of the fixed selected leaves
to this residual optimum gives \(P(r)\). All targets are solved
independently and in parallel.

After solving these residual instances, the values
\(P(0),\ldots,P(q)\) yield the sorted \textcolor{black}{frequencies} at depth \(h^*\)
by \(w_r=P(r)-P(r-1)\), as shown in Lemma~\ref{lem:kraft-prefixes}.

\paragraph{Recovering the stopping state.}

We first determine the depth $h^*$.
Recall that the averaging argument gives two possible depths,
$\AveragingDepth-1$ and $\AveragingDepth$.
If $\AveragingDepth=0$, the stopping state consists only
of the root. We return its frequency and the total
Huffman cost $\sum_i f_i h_i$.

Otherwise, we reconstruct the node frequencies at depth
$\AveragingDepth$. We then sum consecutive pairs to obtain
the node frequencies at depth $\AveragingDepth-1$.
As shown earlier, the minimum frequency at this depth
determines which of the two depths is $h^*$.
We use the reconstructed frequencies at the chosen depth.

Let $q=2^{h^*}$, and write the frequencies at depth $h^*$
in nondecreasing order as
$w_1\leq\cdots\leq w_q$.
Let $s$ be the number of these frequencies that are at most
$\tau$. We merge consecutive pairs among the first $s$
frequencies, leaving the remaining frequencies unchanged.
The stopping state therefore consists of the
$\lfloor s/2\rfloor$ new parent frequencies and
\[
w_{2\lfloor s/2\rfloor+1},\ldots,w_q.
\]
No further merge is possible before a frequency greater
than $\tau$ must participate.

The total cost of the merges performed before this state is
\[
\sum_{i=1}^m f_i(h_i-h^*)
+\sum_{j=1}^{2\lfloor s/2\rfloor}w_j.
\]
The first term accounts for the merges needed to construct
the nodes at depth $h^*$, while the second accounts for
the additional pairs. All pair sums, cutoff computations,
and cost calculations can be performed in parallel.
Lemma~\ref{lem:iterative-state-reconstruction-formal}
summarizes the resulting bounds.

\begin{lemma}
\label{lem:iterative-state-reconstruction-formal}
Iterative-State Reconstruction takes
\(O(1+\log_\lambda m)\) \MPC rounds, uses
\[
 O\bigl((1+\log\lambda)\lambda^{\log_\varphi2}\bigr)
\]
local memory per machine, and has total memory and computation
\textcolor{black}{
\[
 O\!\left(m\lambda^{\log_\varphi2}(1+\log\lambda)
 (1+\log_\lambda m)\log(2m\lambda)\right).
\]}
\end{lemma}

\begin{proof}
We first bound the memory needed to compute one value \(P(r)\).
Lemma~\ref{lem:fractional-proximity} allows us to retain
\(O(2^\Delta)\) leaves from each of the \(\Delta+1\) depth groups.
The resulting residual instance therefore contains
\(O((\Delta+1)2^\Delta)\) leaves.
{\color{black}
The proof of Corollary~\ref{cor:bounded-ratio-depth-range} gives
\(\Delta\leq\lceil\log_\varphi\lambda\rceil+1\). Hence
\[
 2^\Delta
 \leq 2^{\lceil\log_\varphi\lambda\rceil+1}
 \leq 4\lambda^{\log_\varphi2}.
\]
Together with \(\Delta=O(1+\log\lambda)\), this bounds the number
of retained leaves by
\(O((1+\log\lambda)\lambda^{\log_\varphi2})\), which fits on one
machine with the stated local memory.
} Package-Merge
solves this instance in time linear in its number of leaves.

We only need to compute $P(1),\ldots,P(q)$ at depth
$\AveragingDepth$, where $q=2^{\AveragingDepth}$.
The averaging argument determines this depth directly
from $\TotalFrequency{A}$ and $h_{\min}$, without requiring
any additional residual instances.

By Kraft equality,
$q\leq 2^{h_{\min}}\leq m$.
Since each value of $P(r)$ requires
$O((\Delta+1)2^\Delta)$ retained records,
the total number of records is
\[
O\bigl(q(\Delta+1)2^\Delta\bigr)
=O\bigl(m(1+\log\lambda)\lambda^{\log_\varphi 2}\bigr).
\]

Since each local solve takes linear time, this also bounds their total
computation. \textcolor{black}{Let
\(\RecordCount{}=O(m(1+\log\lambda)\lambda^{\log_\varphi2})\) be the number
of input, query, and retained-leaf records. Sorting, scans, and batched
routing use \(O(1+\log_\lambda m)\) rounds, with at most \(O(\RecordCount{})\)
records and \(O(\RecordCount{}\log(2\RecordCount{}))\) computation per round. Since
\(\log(2\RecordCount{})=O(\log(2m\lambda))\), these operations and the local
solves satisfy the explicit total-resource bound in the lemma.}

{\color{black}
Finally, we bound the number of rounds. We compute \(\TotalFrequency{A}\) and
\(h_{\min}\) by reductions and choose \(\AveragingDepth\). After sorting the
leaves, prefix sums and batched searches locate the cutoff windows,
and routing sends each residual instance to its machine. These
operations take \(O(1+\log_\lambda m)\) \MPC rounds with the stated
resources~\cite{goodrich2011sorting}. All residual instances are then
solved locally in parallel. Pair sums, the choice between \(\AveragingDepth-1\)
and \(\AveragingDepth\), and the final frequencies and merge cost fit within the
same round bound. Thus the overall bound remains
\(O(1+\log_\lambda m)\).
}

\end{proof}

%% file: Sections/iterative-state-example-figure.tex
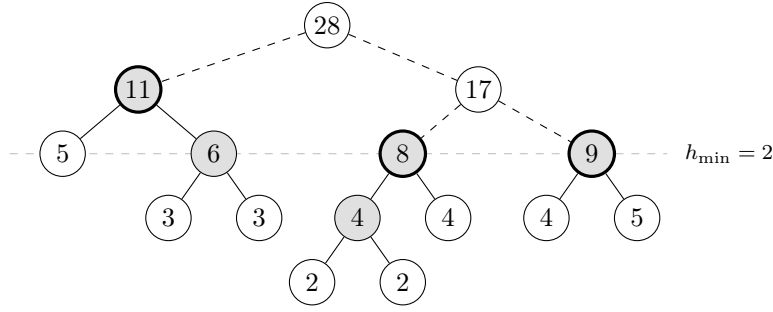
\begin{figure}[t]
\centering
\begin{tikzpicture}[x=1cm,y=0.85cm,
  v/.style={circle,draw=black,fill=white,minimum size=6mm,inner sep=1pt,font=\small},
  done/.style={v,fill=black!12,draw=black},
  frontier/.style={line width=1.2pt,draw=black},
  pending/.style={draw=black,dashed}]
\draw[black!25,dashed] (-0.7,-2) -- (8,-2);
\node[anchor=west,font=\scriptsize,text=black] at (8.1,-2) {$h_{\min}=2$};
\node[v] (root) at (3.5,0) {28};
\node[done,frontier] (n11) at (1,-1) {11};
\node[v] (n17) at (5.5,-1) {17};
\node[v] (f5b) at (0,-2) {5};
\node[done] (n6) at (2,-2) {6};
\node[done,frontier] (n8) at (4.5,-2) {8};
\node[done,frontier] (n9) at (7,-2) {9};
\node[v] (f3a) at (1.4,-3) {3};
\node[v] (f3b) at (2.6,-3) {3};
\node[done] (n4) at (3.9,-3) {4};
\node[v] (f4a) at (5.1,-3) {4};
\node[v] (f4b) at (6.4,-3) {4};
\node[v] (f5a) at (7.6,-3) {5};
\node[v] (f2a) at (3.3,-4) {2};
\node[v] (f2b) at (4.5,-4) {2};
\foreach \p/\c in {root/n11,root/n17,n17/n8,n17/n9}
  \draw[pending] (\p) -- (\c);
\foreach \p/\c in {n11/f5b,n11/n6,n6/f3a,n6/f3b,n8/n4,n8/f4a,n9/f4b,n9/f5a,n4/f2a,n4/f2b}
  \draw[black!65!black] (\p) -- (\c);
\end{tikzpicture}
\caption{Iterative-State Reconstruction with $\tau=6$.
Shaded internal nodes represent completed merges; thick outlines mark
the \textcolor{black}{frequencies} $8,9,11$. Dashed edges belong to merges not yet
performed. The horizontal line marks the minimum leaf depth.}
\label{fig:iterative-state-example}
\end{figure}

%% file: Sections/putting-pieces-together.tex
\section{Putting the Pieces Together}
\label{sec:putting-pieces-together}

In this section, we combine Bounded-Ratio Huffman Coding and
Iterative-State Reconstruction to solve Huffman coding for
arbitrary positive integer frequencies.

Our algorithm processes the frequencies in batches, starting
with the smallest ones. For each batch, Bounded-Ratio Huffman
Coding computes an optimal depth assignment.
Iterative-State Reconstruction then uses these depths to
determine which Huffman merges can be performed before
involving the remaining input frequencies.
We add the cost of these merges to a running total and
replace the processed batch with the frequencies remaining
after the merges.
We repeat this process until the remaining frequencies form
a bounded-ratio instance, which we solve directly.

We first analyze our algorithm for a list of $m$
frequencies with total sum $n$. The complexity of
the two subtasks depends on $m$, while the sum of
the frequencies bounds the number of iterations.
We then introduce a preprocessing step that performs
Huffman merges on small frequencies, reducing the
number of frequencies passed to the main algorithm.
We use it to obtain nearly linear total computation
in the text length $n$, as stated in
Theorem~\ref{thm:main-result}.

Fix \(0<\epsilon<1\), and set
\[
 \lambda=\max\{2,m^{\epsilon/16}\},\qquad
 \eta=\max\{2,m^{\epsilon/4}\}.
\]
These parameters remain fixed throughout the algorithm.
We maintain a multiset $B$ of frequencies and an accumulated
merge cost $\AccumulatedMergeCost$, initially set to $A$
and zero, respectively.

In each iteration, let $k=|B|$ and sort the current
frequencies as $f_1\leq\cdots\leq f_k$.
We choose the largest prefix $f_1,\ldots,f_j$
satisfying $f_j\leq\lambda f_2$ and solve it using
Bounded-Ratio Huffman Coding.

If $j=k$, the entire remaining instance has been solved.
We return its optimal cost plus $\AccumulatedMergeCost$.
Otherwise, we invoke Iterative-State Reconstruction
with threshold $\tau=f_{j+1}$.
It returns the frequencies remaining after the completed
Huffman merges within the prefix, together with their
total cost. We replace the prefix in $B$ by these
frequencies, add the returned cost to
$\AccumulatedMergeCost$, and proceed to the next iteration.

Algorithm~\ref{alg:complete-huffman} summarizes the procedure.
Throughout the analysis, $m$ denotes the original number
of frequencies, while $k$ is the number in the current iteration.

\begin{algorithm}[t]
\caption{Huffman coding from frequencies}
\label{alg:complete-huffman}
\begin{algorithmic}[1]
\Require Positive integer frequencies \(A\), a size bound \(m\geq\max\{2,|A|\}\), and \(0<\epsilon<1\)
\Ensure The optimum encoded length
\State \(\lambda\gets\max\{2,m^{\epsilon/16}\}\);
  \(\eta\gets\max\{2,m^{\epsilon/4}\}\)
\State \(\AccumulatedMergeCost\gets0\); \(B\gets A\)
\While{\(|B|\geq2\)}
  \State Sort \(B\) as \(f_1\leq\cdots\leq f_k\)
  \State \(j\gets\max\{i\in\{2,\ldots,k\}:f_i\leq\lambda f_2\}\)
  \State Apply Bounded-Ratio Huffman Coding with \(\eta\) to \((f_1,\ldots,f_j)\)
  \State Let \(v\) be its optimum and \((h_1,\ldots,h_j)\) its depth assignment
  \If{\(j=k\)}
    \State \Return \(\AccumulatedMergeCost+v\)
  \EndIf
  \State \(\tau\gets f_{j+1}\)
  \State Apply Iterative-State Reconstruction to the prefix, its depths, and \(\tau\)
  \State Let \(D\) be the returned multiset and \(\RecoveredMergeCost\) the returned merge cost
  \State \(\AccumulatedMergeCost\gets\AccumulatedMergeCost+\RecoveredMergeCost\); \(B\gets D\uplus\{f_{j+1},\ldots,f_k\}\)
\EndWhile
\State \Return \(\AccumulatedMergeCost\)
\end{algorithmic}
\end{algorithm}

\textcolor{black}{Note that, for frequency input, local and total memory are measured in words of $O(\log n)$ bits, where $n$ is the sum of the frequencies. Arithmetic and comparisons on these words take constant time.}

The resource bounds are stated in Theorem~\ref{thm:frequency-input-mpc}.
For later use, we allow \(m\) to be an upper bound on the input size;
the same fixed parameter choices and analysis apply.

\begin{theorem}
\label{thm:frequency-input-mpc}
For every \(0<\epsilon<1\), Huffman coding on at most \(m\geq2\)
positive integer frequencies of total sum \(n\) can be solved in
\[
 O\!\left(\epsilon^{-3}
 \left(1+\frac{\log n}{\log m}\right)\log(2+\log m)\right)
\]
\MPC rounds using \(O(m^\epsilon\log(2m))\) local memory per machine.
The total memory is 
\[
 O\!\left(m^{1+\epsilon}\epsilon^{-2}
 \log^2(2m)\log(2+\log m)\right),
\]
and the total computation is 
\[
 O\!\left(m^{1+\epsilon}\epsilon^{-3}
 \left(1+\frac{\log n}{\log m}\right)
 \log^2(2m)\log(2+\log m)\right).
\]
\end{theorem}

\begin{proof}
We first prove that the algorithm returns the optimal Huffman cost.
Our argument maintains the following invariant: the accumulated
merge cost plus the optimal cost of the remaining frequencies
equals the optimal cost of the original input.
This holds initially, since no merges have been performed.

Consider an iteration in which the chosen prefix does not
contain all frequencies. Let $\tau$ be the smallest frequency
outside the prefix. Iterative-State Reconstruction performs
valid Huffman merges on the prefix, stopping before a merge
would involve a frequency greater than $\tau$.
Every merge performed before this point combines the two
smallest frequencies in the prefix, both at most $\tau$.
Since all frequencies outside the prefix are at least $\tau$,
these are also the two smallest frequencies in the full
instance, with ties broken in favor of the prefix.
Thus, the same merges are valid for the full instance.

Each such merge preserves the optimal Huffman cost after
adding its cost to the accumulated total. The algorithm
therefore maintains the invariant when it replaces the
prefix with the remaining frequencies.
When the entire list satisfies the bounded-ratio condition,
we compute its optimal cost and add it to the accumulated
total. If only one frequency remains, the accumulated
cost is already optimal.

{\color{black}
\paragraph{Number of rounds.}
We first bound the number of iterations.
By the maximality of the chosen prefix, we have
$\tau>\lambda f_2$. After reconstruction, at most one
frequency is smaller than $\tau$, while all frequencies
outside the chosen prefix are at least $\tau$.
Hence, if another iteration is needed, the second-smallest
frequency increases by a factor greater than $\lambda$.
Since the frequencies remain positive and their sum is $n$,
the second-smallest frequency cannot exceed $n$.
The number of iterations is therefore
$O(1+\log_\lambda n)$.
Recall that $\lambda=\max\{2,m^{\epsilon/16}\}$.
Using $\log\lambda\geq(\epsilon/16)\log m$, we obtain
\[
1+\log_\lambda n
\leq 1+\frac{16\log n}{\epsilon\log m}
=O\!\left(
\epsilon^{-1}\left(1+\frac{\log n}{\log m}\right)
\right).
\]

We next bound the rounds needed for one iteration.
There are at most $m$ current frequencies.
Our choices of $\eta$ and $\lambda$ give
\[
1+\log_\eta m=O(\epsilon^{-1}),
\qquad
1+\log_\lambda m=O(\epsilon^{-1}).
\]
Moreover, since $\lambda=O(m^{\epsilon/16})$,
\[
\log(2+\log\lambda)=O(\log(2+\log m)).
\]
By Lemma~\ref{lem:bounded-ratio-huffman-formal},
Bounded-Ratio Huffman Coding takes
\[
O\!\left(
(1+\log_\eta m)^2\log(2+\log\lambda)
\right)
=
O\!\left(\epsilon^{-2}\log(2+\log m)\right)
\]
rounds.
By Lemma~\ref{lem:iterative-state-reconstruction-formal},
Iterative-State Reconstruction takes
$O(1+\log_\lambda m)=O(\epsilon^{-1})$ rounds.
Sorting the current frequencies and updating the
accumulated cost also take $O(\epsilon^{-1})$ rounds.
Thus, each iteration takes
$O(\epsilon^{-2}\log(2+\log m))$ rounds.

Combining the number of iterations with the rounds
per iteration gives the total round complexity
\[
O\!\left(
\epsilon^{-3}
\left(1+\frac{\log n}{\log m}\right)
\log(2+\log m)
\right),
\]
as claimed.

\paragraph{Local memory.}
Let $\alpha=\log_\varphi 2<2$.
We bound the local memory required by each subtask.
Since $m\geq2$, our parameter choices give
$\eta=O(m^{\epsilon/4})$ and
$\lambda=O(m^{\epsilon/16})$.

By Lemma~\ref{lem:bounded-ratio-huffman-formal},
Bounded-Ratio Huffman Coding uses
\begin{align*}
	O(\eta\lambda^{2\alpha})
	&=O\!\left(m^{(1/4+\alpha/8)\epsilon}\right)\\
	&=O(m^{\epsilon/2}),
\end{align*}
where the last inequality follows from $\alpha<2$.

Similarly, by Lemma~\ref{lem:iterative-state-reconstruction-formal},
Iterative-State Reconstruction uses
\begin{align*}
	O((1+\log\lambda)\lambda^\alpha)
	&=O\!\left(m^{\alpha\epsilon/16}\log(2m)\right)\\
	&=O\!\left(m^{\epsilon/8}\log(2m)\right).
\end{align*}

Both subtasks therefore fit within the claimed
$O(m^\epsilon\log(2m))$ local memory per machine.
Sorting the at most $m$ frequencies can also be performed
using standard MPC sorting algorithms with this local memory.
Updating the accumulated cost requires only a parallel sum
and a constant number of additional values per machine,
which fit within the same bound.
\paragraph{Total memory and computation.}
We first bound the memory and computation required in one
iteration. By Lemma~\ref{lem:bounded-ratio-huffman-formal},
the total memory and computation of a bounded-ratio call
are each bounded by
\begin{align*}
	&O\!\left(
	m\eta\lambda^{4\alpha}\epsilon^{-2}
	\log(2m)\log(2+\log m)
	\right)\\
	&\qquad=
	O\!\left(
	m^{1+(1+\alpha)\epsilon/4}\epsilon^{-2}
	\log(2m)\log(2+\log m)
	\right)\\
	&\qquad=
	O\!\left(
	m^{1+\epsilon}\epsilon^{-2}
	\log(2m)\log(2+\log m)
	\right).
\end{align*}
Here we substitute the chosen values of $\eta$ and
$\lambda$ and use $\alpha<2$.
We also use
$\log(2m\eta\lambda)=O(\log(2m))$,
which follows from $\epsilon<1$.

Similarly, Lemma~\ref{lem:iterative-state-reconstruction-formal}
bounds the total memory and computation of reconstruction by
\begin{align*}
	O\!\left(
	m\lambda^\alpha\epsilon^{-1}\log^2(2m)
	\right)
	&=O\!\left(
	m^{1+\alpha\epsilon/16}\epsilon^{-1}
	\log^2(2m)
	\right)\\
	&=O\!\left(
	m^{1+\epsilon}\epsilon^{-1}
	\log^2(2m)
	\right),
\end{align*}
where the last bound again uses $\alpha<2$.

Including sorting and updates to the current frequency list,
the total memory and computation of one iteration are each
bounded by
\[
O\!\left(
m^{1+\epsilon}\epsilon^{-2}
\log^2(2m)\log(2+\log m)
\right).
\]
Since memory can be reused between iterations, this also
bounds the total memory of the algorithm.

For total computation, we multiply the per-iteration bound
by the number of iterations established above. This gives
\[
O\!\left(
m^{1+\epsilon}\epsilon^{-3}
\left(1+\frac{\log n}{\log m}\right)
\log^2(2m)\log(2+\log m)
\right),
\]
as claimed.

}

\end{proof}

{\color{black}
For a text of length $n$, the character frequencies sum to $n$,
and there are at most $n$ distinct characters.
Theorem~\ref{thm:frequency-input-mpc} therefore gives
Corollary~\ref{cor:text-input-direct}.

\begin{corollary}
\label{cor:text-input-direct}
For every \(0<\epsilon<1\), Huffman coding for a text of length
\(n\geq2\) can be solved in
\(O(\epsilon^{-3}\log(2+\log n))\) \MPC rounds using
\(O(n^\epsilon\log(2n))\) local memory per machine.
The total memory is
\[
 O\!\left(n^{1+\epsilon}\epsilon^{-2}
 \log^2(2n)\log(2+\log n)\right),
\]
and the total computation is
\[
 O\!\left(n^{1+\epsilon}\epsilon^{-3}
 \log^2(2n)\log(2+\log n)\right).
\]
\end{corollary}

\begin{proof}
We first compute the character frequencies by sorting and aggregation.
These operations take $O(\epsilon^{-1})$ rounds using
$O(n^\epsilon)$ local memory and
$O(n\epsilon^{-1}\log(2n))$ total memory and computation.
We then apply Theorem~\ref{thm:frequency-input-mpc}, using $n$
as an upper bound on the number of frequencies.
The factor involving the total frequency becomes
$1+\log n/\log n=2$, giving the stated bounds.
The initial aggregation also fits within these bounds.
\end{proof}

Corollary~\ref{cor:text-input-direct} gives the desired round bound,
but its total memory and computation still contain the factor
$n^{1+\epsilon}$. We next remove the extra polynomial factor
by preprocessing the small frequencies, obtaining nearly linear
total memory and computation in $n$.
}
We first perform Huffman
merges on small frequencies to reduce the size of the
instance. We then apply Theorem~\ref{thm:frequency-input-mpc}
to the remaining frequencies and add the cost of the
merges performed during preprocessing.

\paragraph{Compressing small frequencies.}
\textcolor{black}{The idea is to reduce the number of frequencies before}
running the main algorithm. We perform Huffman merges
until at most one frequency smaller than
$\lceil n^\epsilon\rceil$ remains.
Rather than processing these merges individually,
we group equal frequencies and process them together.

For each integer $x<\lceil n^\epsilon\rceil$, let
$\FrequencyMultiplicity{x}$ denote the number of input
frequencies equal to $x$.
We gather these $\lceil n^\epsilon\rceil-1$ counters
on one machine and process them in increasing order,
starting from $x=1$.
The machine also maintains at most one \emph{carry},
representing an unpaired frequency smaller than
the current value of $x$.

When processing frequency $x$, we first merge the carry
with one copy of $x$, if both are present.
We then pair the remaining copies of $x$.
If there are $a$ such copies, these pairs produce
$\lfloor a/2\rfloor$ new frequencies equal to $2x$
and contribute $2x\lfloor a/2\rfloor$ to the
accumulated merge cost.
Any remaining unpaired copy of $x$ becomes the new carry.

Each newly created frequency smaller than
$\lceil n^\epsilon\rceil$ is added to its corresponding
counter and will be processed later.
Frequencies at least $\lceil n^\epsilon\rceil$
are kept for the main algorithm.
We also retain the final carry, if one exists.

These merges follow Huffman's algorithm.
Before processing $x$, at most one frequency is smaller
than $x$. We merge it with the smallest available
frequency and then merge the remaining copies of $x$
in pairs. All newly created frequencies are greater
than $x$, allowing us to process them at later steps.

After processing all counters, at most one remaining
frequency is smaller than $\lceil n^\epsilon\rceil$.
Every other frequency is at least $\lceil n^\epsilon\rceil$,
and their total sum is $n$.
Thus, the remaining instance contains at most
\[
\frac{n}{\lceil n^\epsilon\rceil}+1
\leq n^{1-\epsilon}+1
\]
frequencies.
We pass this smaller instance to the main algorithm
and add the preprocessing merge cost to its answer. Lemma~\ref{lem:small-weight-preprocessing} establishes the
correctness of this preprocessing step and bounds its
resource requirements.
\begin{lemma}
\label{lem:small-weight-preprocessing}
For a text of length \(n\), the preprocessing with threshold
\(\lceil n^\epsilon\rceil\) returns a multiset \(A_0\) of \textcolor{black}{frequencies} and the total cost \(c_0\) of the merges performed, with the
following properties:
\begin{itemize}
\item The optimum encoded length for the input frequencies equals
\(c_0\) plus the optimum encoded length for \(A_0\).
\item Either \(|A_0|\leq1\), or the second-smallest \textcolor{black}{frequency} in
\(A_0\) is at least \(\lceil n^\epsilon\rceil\).
\item The number of remaining \textcolor{black}{frequencies} satisfies
\(|A_0|\leq n^{1-\epsilon}+1\).
\end{itemize}
The procedure takes \(\textcolor{black}{O(\epsilon^{-1})}\) \MPC rounds,
uses \(O(n^\epsilon)\) local memory, and has \(\textcolor{black}{O(n\epsilon^{-1}\log(2n))}\) total memory and
computation.
\end{lemma}

\begin{proof}
When we begin processing a frequency $x$, all smaller
frequencies have already been merged, except possibly
one carry. If both the carry and a copy of $x$ exist,
they are the two smallest remaining frequencies, and
merging them is a valid Huffman step.
We then merge the remaining copies of $x$ in pairs.
Each pair consists of the two smallest remaining
frequencies and is therefore a valid Huffman merge.
Every new frequency is greater than $x$ and will be
processed at a later step.
Thus, processing each value of $x$ performs only valid
Huffman merges. By induction on $x$, the entire
preprocessing procedure is an initial sequence of
a valid Huffman execution.

Since the preprocessing performs a prefix of a valid Huffman
execution, we can complete that execution by running Huffman's
algorithm on \(A_0\). The merges already performed contribute
\(c_0\), and the remaining merges contribute the optimum encoded
length for \(A_0\).

After processing all \textcolor{black}{frequencies} below \(\lceil n^\epsilon\rceil\), at most the carry has \textcolor{black}{frequency}
smaller than \(\lceil n^\epsilon\rceil\).  Hence, if at least two \textcolor{black}{frequencies} remain, the
second-smallest one is at least \(\lceil n^\epsilon\rceil\).  The sum of the \textcolor{black}{frequencies} is
\(n\). Hence \textcolor{black}{
\begin{align*}
 |A_0|-1&\leq\frac{n}{\lceil n^\epsilon\rceil}\\
 &\leq\frac{n}{n^\epsilon}\\
 &=n^{1-\epsilon}.
\end{align*}}
Thus \(|A_0|\leq n^{1-\epsilon}+1\).

For the \MPC implementation, we first compute the character frequencies
and aggregate all frequencies below \(\lceil n^\epsilon\rceil\) into
\(\FrequencyMultiplicity{1},\ldots,\FrequencyMultiplicity{\lceil n^\epsilon\rceil-1}\).  Standard sorting, prefix-sum, and routing
primitives perform this aggregation~\cite{goodrich2011sorting}.
\textcolor{black}{With \(O(n^\epsilon)\) local memory, their round bound is
\begin{align*}
 O(1+\log_{n^\epsilon}n)
 &=O\!\left(1+\frac{\log n}{\epsilon\log n}\right)\\
 &=O(\epsilon^{-1}).
\end{align*}
These MPC primitives take $O(\epsilon^{-1})$ rounds,
with $O(n\log(2n))$ total computation per round.
Their total computation is therefore
$O(n\epsilon^{-1}\log(2n))$.
The $O(n^\epsilon)$ counters and the local simulation
fit within these resource bounds.
}

The compressed simulation processes each frequency below
$\lceil n^\epsilon\rceil$ in constant time, using
$O(n^\epsilon)$ local time and memory.

Every new frequency that reaches or exceeds the threshold
is smaller than $2\lceil n^\epsilon\rceil$.
We store the number of copies of each such frequency
using another $O(n^\epsilon)$ counters.

To construct the remaining instance, we use prefix sums
of these counters to determine where each copy should
appear in the output list.
We then add the original frequencies at least
$\lceil n^\epsilon\rceil$ and the final carry, if any,
and sort the combined list.
As shown above, the remaining instance contains
$O(n^{1-\epsilon})$ frequencies.
The entire procedure takes $O(\epsilon^{-1})$ \MPC rounds
and uses $O(n\epsilon^{-1}\log(2n))$ total
memory and computation.
\end{proof}

{\color{black}
We can now obtain Theorem~\ref{thm:main-result} as a consequence of
Theorem~\ref{thm:frequency-input-mpc} and
Lemma~\ref{lem:small-weight-preprocessing}. 
}

\begin{mainresultrestatement}
{\color{black}
For every \(0<\epsilon<1\), there is an \MPC algorithm that receives
a text \(\Text\) of length \(n\) and computes
\(\MinimumEncodedLength(\Text)\) in
\(O(\epsilon^{-3}\log(2+\log n))\) rounds using
\(O(n^\epsilon\log(2n))\) memory per machine. Its total memory and
total computation are each at most
\[
 O\!\left(n\epsilon^{-3}\log^2(2n)\log(2+\log n)\right).
\]
}
\end{mainresultrestatement}

\begin{proof}
{\color{black}
Apply Lemma~\ref{lem:small-weight-preprocessing} with exponent
\(\epsilon/2\). This computes the cost \(c_0\) of the initial
merges and a residual multiset \(A_0\) with
\[
 |A_0|\leq n^{1-\epsilon/2}+1.
\]
Use the size bound \(m_0=\lceil n^{1-\epsilon/2}\rceil+1\) and apply
Theorem~\ref{thm:frequency-input-mpc} to \(A_0\), again with exponent
\(\epsilon/2\). The sum of its frequencies is still \(n\). Adding
\(c_0\) to its optimum gives the optimum for the original text.

Since \(m_0=\Theta(n^{1-\epsilon/2})\), its logarithm is
\(O(\log(2n))\). Moreover,
\begin{align*}
 \log m_0&\geq(1-\epsilon/2)\log n\\
 &\geq\tfrac12\log n,
\end{align*}
where the last inequality uses \(\epsilon<1\). Therefore,
\begin{align*}
 1+\frac{\log n}{\log m_0}
 &\leq1+2\\
 &=3.
\end{align*}
Substituting this bound into the round guarantee of
Theorem~\ref{thm:frequency-input-mpc}, with exponent \(\epsilon/2\), gives
\begin{align*}
 &O\!\left((\epsilon/2)^{-3}
       \left(1+\frac{\log n}{\log m_0}\right)
       \log(2+\log m_0)\right)\\
 &\qquad{}=O\!\left((\epsilon/2)^{-3}\log(2+\log m_0)\right)\\
 &\qquad{}=O\!\left(\epsilon^{-3}\log(2+\log m_0)\right)\\
 &\qquad{}=O\!\left(\epsilon^{-3}\log(2+\log n)\right).
\end{align*}
Its local-memory bound becomes
\begin{align*}
 O(m_0^{\epsilon/2}\log(2m_0))
 &=O\!\left((n^{1-\epsilon/2})^{\epsilon/2}\log(2m_0)\right)\\
 &=O\!\left(n^{\epsilon/2-\epsilon^2/4}\log(2m_0)\right)\\
 &=O\!\left(n^{\epsilon/2}\log(2m_0)\right)\\
 &=O\!\left(n^{\epsilon/2}\log(2n)\right)\\
 &=O\!\left(n^\epsilon\log(2n)\right).
\end{align*}
Preprocessing adds \(O((\epsilon/2)^{-1})=O(\epsilon^{-1})\)
rounds and uses \(O(n^{\epsilon/2})\) local memory, both within
these bounds.

For total resources, we first bound the polynomial factor:
\begin{align*}
 m_0^{1+\epsilon/2}
 &=O\!\left((n^{1-\epsilon/2})^{1+\epsilon/2}\right)\\
 &=O\!\left(n^{(1-\epsilon/2)(1+\epsilon/2)}\right)\\
 &=O\!\left(n^{1-\epsilon^2/4}\right)\\
 &=O(n).
\end{align*}
The total-computation bound of Theorem~\ref{thm:frequency-input-mpc}
therefore simplifies as follows:
\begin{align*}
 &O\!\left(m_0^{1+\epsilon/2}(\epsilon/2)^{-3}
   \left(1+\frac{\log n}{\log m_0}\right)
   \log^2(2m_0)\log(2+\log m_0)\right)\\
 &\qquad{}=O\!\left(n(\epsilon/2)^{-3}
   \left(1+\frac{\log n}{\log m_0}\right)
   \log^2(2m_0)\log(2+\log m_0)\right)\\
 &\qquad{}=O\!\left(n\epsilon^{-3}
   \left(1+\frac{\log n}{\log m_0}\right)
   \log^2(2m_0)\log(2+\log m_0)\right)\\
 &\qquad{}=O\!\left(n\epsilon^{-3}
   \log^2(2m_0)\log(2+\log m_0)\right)\\
 &\qquad{}=O\!\left(n\epsilon^{-3}
   \log^2(2n)\log(2+\log n)\right).
\end{align*}
The total-memory guarantee is no larger. Preprocessing uses
\(O(n(\epsilon/2)^{-1}\log(2n))=O(n\epsilon^{-1}\log(2n))\)
total resources. Adding this contribution yields
\begin{align*}
 &O\!\left(n\epsilon^{-3}\log^2(2n)\log(2+\log n)
          +n\epsilon^{-1}\log(2n)\right)\\
 &\qquad{}=O\!\left(n\epsilon^{-3}\log^2(2n)\log(2+\log n)\right),
\end{align*}
which completes the proof.

}
\end{proof}

%% file: Sections/recovering-codewords.tex
\section{Recovering the Codewords}
\label{sec:recovering-codewords}

{\color{black}
{\color{black}
So far, our algorithm computes only the optimal encoded length.
We now explain how to recover the codewords themselves.

The main idea is to build and maintain a weighted tree that records
how frequencies are combined throughout the algorithm.
Initially, each input character is represented by a leaf.
In each iteration, Bounded-Ratio Huffman Coding determines
the depths of the items in the selected prefix. During
Iterative-State Reconstruction, we also retain the information
needed to identify which items belong to each new frequency.

We represent each new frequency by a node and connect it
to the items that form it. The weight of each edge is the
depth of the corresponding item relative to the new node.
Thus, an edge of weight $k$ records $k$ steps toward the root
in the binary coding tree.

As the algorithm proceeds, newly created frequencies
may be combined again, and we extend the weighted tree
accordingly. Eventually, every original character is
connected to the final root. The sum of the edge weights
along this path gives its final codeword length.
We then construct a prefix-free code from these lengths.
We show below that the lengths satisfy Kraft equality and that
their total encoded length equals the cost computed by the algorithm.
}

For example, take characters $a,b,c,d,e$ with frequencies
$(1,1,2,2,3)$. Bounded-Ratio Huffman Coding first assigns local
depth $2$ to each item in the prefix $a,b,c,d$. With threshold
$\tau=3$, Iterative-State Reconstruction groups $a,b$ into a node
$u$ and $c,d$ into a node $v$, both at local depth $1$. We therefore
record depth $2-1=1$ for each item within its group; $e$ remains
unchanged. The final Bounded-Ratio solve assigns depths $2,2,1$
to the remaining frequencies $u:2,e:3,v:4$, respectively. Character $a$ therefore has depth $1+2=3$, whereas $c$
has depth $1+1=2$. Figure~\ref{fig:codeword-recovery} shows the
connections and a code with the recovered lengths $(3,3,2,2,2)$.

\input{Sections/codeword-recovery-figure}

\paragraph{Constructing a prefix code from its lengths.}
Suppose we have recovered the final depth of every character.
Our goal is to assign each character a binary codeword of
the prescribed length.

We begin by sorting the characters by nondecreasing depth,
$\LeafDepth{1}\leq\cdots\leq\LeafDepth{m}$.
For each character $i$, we reserve an interval of length
$2^{-\LeafDepth{i}}$ inside $[0,1)$.
We place these intervals next to one another, starting
from zero. Since the depths satisfy Kraft equality,
the intervals together fill $[0,1)$.

The idea is to use each interval to determine its codeword.
A binary codeword of length $h$ identifies one of the
$2^h$ equal intervals of length $2^{-h}$ in $[0,1)$.
For example, the two-bit codewords
$\texttt{00},\texttt{01},\texttt{10},\texttt{11}$
correspond to
\[
[0,1/4),\quad [1/4,1/2),\quad
[1/2,3/4),\quad [3/4,1).
\]
We assign each character the codeword corresponding
to its reserved interval.

To make this assignment possible, each reserved interval
must start at a multiple of its own length.
This follows from the ordering of the depths.
For character $i$, every preceding interval has length
$2^{-\LeafDepth{j}}$, where
$\LeafDepth{j}\leq\LeafDepth{i}$.
Each such length is an integer multiple of
$2^{-\LeafDepth{i}}$.
The starting position of character $i$ is therefore
also a multiple of its interval length.

We can compute the codewords using prefix sums.
Let $\MaximumLeafDepth=\max_i\LeafDepth{i}$.
We measure all interval lengths in units of
$2^{-\MaximumLeafDepth}$, making them integers.
Character $i$ then has interval width
$2^{\MaximumLeafDepth-\LeafDepth{i}}$.
Its starting position and interval index are
\[
a_i=\sum_{j<i}2^{\MaximumLeafDepth-\LeafDepth{j}},
\qquad
k_i=\frac{a_i}{2^{\MaximumLeafDepth-\LeafDepth{i}}}.
\]
We obtain the codeword for character $i$ by writing
$k_i$ in binary using exactly $\LeafDepth{i}$ bits,
adding leading zeros when needed.
Figure~\ref{fig:codeword-intervals} illustrates
this construction.

Finally, the resulting code is prefix-free.
Each bit of a codeword chooses one half of the current
interval, starting from $[0,1)$.
Therefore, extending a codeword selects a smaller
interval contained in the original one.
If one codeword were a prefix of another, their
corresponding intervals would overlap.
Our reserved intervals are disjoint, ruling this out.
Every codeword has its prescribed length, and the
resulting code therefore has optimal encoded length.

\input{Sections/codeword-intervals-figure}

\paragraph{Bounded-Ratio Huffman Coding.}
We recover the local depths by following the minimizing choices
stored by Bounded-Ratio Huffman Coding. Whenever two table entries
are combined, we store the minimizing split of the parent's budget.
Starting from the optimum root entry, we use this split to select the
two child entries, then repeat until reaching the blocks. The minimizing counts stored in
the chosen block entries then determine a depth for every current
frequency, as described in Section~\ref{sec:bounded-ratio-formal}.
In a nonterminal
iteration, we pass these depths to Iterative-State Reconstruction.
In the final solve, they specify the connections from the remaining
items to the final root. We thus recover the depth assignment from the stored table entries,
without solving another optimization problem.

\paragraph{Iterative-State Reconstruction.}
{\color{black}
For each $P(r)$, we record how many items are selected
at each depth. Within each depth, items are ordered by
nondecreasing frequency, and an optimal selection takes
a prefix of this order.

We choose the optimal solutions consistently so that
the selection for $P(r-1)$ is contained in that for $P(r)$.
Therefore, the items belonging to node $r$ are exactly
those added when moving from $P(r-1)$ to $P(r)$.

More precisely, if $c_{r,h}$ denotes the number of items
selected at depth $h$ for $P(r)$, then node $r$ contains
items
\[
c_{r-1,h}+1,\ldots,c_{r,h}
\]
from that depth group.
Thus, storing these counts is enough to recover the
members of every reconstructed node.

To ensure nested optimal solutions, we number the items once.
Between equal-cost selections, we prefer the one containing the
smallest-numbered item on which they differ.
The Package-Merge construction in Lemma~\ref{lem:kraft-prefixes}
pairs items independently of $r$ and produces a common ordered
list of packages representing the nodes at the reconstruction depth.
An optimal solution for $P(r)$ selects the first $r$
packages in this list. Hence, increasing $r$ only
adds one package, and all previously selected items
remain selected.

We use the same tie-breaking rule in the fractional solution and
in the residual optimization problems. The exchange argument of
Lemma~\ref{lem:fractional-proximity} still applies, since an exchange
toward the fractional prefix either lowers the cost or improves
this preference. Thus the retained windows contain the chosen optimum.

}

\input{Sections/codeword-membership-figure}

\paragraph{Combining the iterations.}
We maintain a weighted tree that records how items are
combined across iterations. Whenever several items are
replaced by a reconstructed node, we connect them to
that node while keeping all earlier connections.
If an item has depth $h_i$ in the Bounded-Ratio solution
and its reconstructed node has depth $h$, we assign
weight $h_i-h$ to their connection.
Items that remain unchanged need no new connections.

If two reconstructed nodes are merged at the stopping
state, we record their connections to the resulting
parent in the same way.

We also record the merges performed during
small-frequency preprocessing. Each merge connects
two children to their parent with edges of weight one.
Since these merges are performed in batches, we store
the connections as ranges and recover the individual
connections using prefix sums and routing.

In the final iteration, we connect the remaining items
to the root using their depths in the final Bounded-Ratio
solution. If only one item remains, it is already the root.
Every original character now has a path to the root,
and the sum of the edge weights along this path gives
its final codeword length.

We can maintain this tree with $O(m)$ connections.
To compute the codeword lengths, we use pointer doubling,
which combines consecutive connections along each path.
Since the final depths are $O(\log n)$, this takes
$O(\log\log n)$ steps and
$O(\epsilon^{-1}\log\log n)$ \MPC rounds.
Finally, we construct the codewords using the sorting
and prefix-sum procedure described above.
The additional storage and computation fit within
the bounds of Theorem~\ref{thm:main-result}.
}
\FloatBarrier

%% file: Sections/codeword-recovery-figure.tex
\begin{figure}[!t]
\centering
{\color{black}
\begin{subfigure}[t]{0.48\textwidth}
\centering
\begin{tikzpicture}[x=1cm,y=0.9cm,font=\small,
 v/.style={circle,draw=black,fill=white,inner sep=2pt,minimum size=6mm},
 e/.style={draw=black},lab/.style={fill=white,inner sep=1pt,font=\footnotesize}]
\node[v] (r) at (0,0) {$r$};
\node[v] (u) at (-1.6,-1.2) {$u$};
\node[v] (v) at (0.6,-1.2) {$v$};
\node[v] (e) at (2.2,-1.2) {$e$};
\node[v,text=black] (a) at (-2.2,-2.4) {$a$};
\node[v] (b) at (-1,-2.4) {$b$};
\node[v] (c) at (0,-2.4) {$c$};
\node[v] (d) at (1.2,-2.4) {$d$};
\draw[e,black!65!black,thick] (r)--node[lab,left] {2}(u);
\draw[e] (r)--node[lab,right] {1}(v);
\draw[e] (r)--node[lab,above right] {2}(e);
\draw[e,black!65!black,thick] (u)--node[lab,left] {1}(a);
\draw[e] (u)--node[lab,right] {1}(b);
\draw[e] (v)--node[lab,left] {1}(c);
\draw[e] (v)--node[lab,right] {1}(d);
\node[font=\footnotesize] at (0,-3.3) {$\textcolor{black}{\LeafDepth{a}}=1+2=3,\quad\textcolor{black}{\LeafDepth{c}}=1+1=2$};
\end{tikzpicture}
\caption{Stored connections and path lengths.}
\end{subfigure}\hfill
\begin{subfigure}[t]{0.48\textwidth}
\centering
\begin{tikzpicture}[x=1cm,y=0.8cm,font=\small,
 v/.style={circle,draw=black,fill=white,inner sep=2pt,minimum size=5mm},
 e/.style={draw=black},lab/.style={fill=white,inner sep=1pt,font=\footnotesize}]
\node[v] (r) at (0,0) {};
\node[v] (l) at (-1.4,-1) {};
\node[v] (rr) at (1.4,-1) {};
\node[v] (c) at (-2,-2) {$c$};
\node[v] (d) at (-0.8,-2) {$d$};
\node[v] (e) at (0.7,-2) {$e$};
\node[v] (ab) at (2.1,-2) {};
\node[v,text=black] (a) at (1.65,-3) {$a$};
\node[v] (b) at (2.7,-3) {$b$};
\foreach \p/\q/\bit in {r/l/0,l/c/0,l/d/1,rr/e/0,ab/b/1}
 \draw[e] (\p)--node[lab] {\bit}(\q);
\foreach \p/\q/\bit in {r/rr/1,rr/ab/1,ab/a/0}
 \draw[e,black!65!black,thick] (\p)--node[lab,text=black] {\bit}(\q);
\node[font=\footnotesize] at (0.3,-3.7) {$c:00,\ d:01,\ e:10,\ a:110,\ b:111$};
\end{tikzpicture}
\caption{A prefix code with these lengths.}
\end{subfigure}
\caption{From stored connections to codewords for frequencies
$(1,1,2,2,3)$. Left: $u$ groups $a,b$ and $v$ groups $c,d$;
edge labels are depth increments, not code bits. This contracted
representation need not be binary. Right: the recovered lengths
$(3,3,2,2,2)$ give a full binary prefix tree with optimum weighted
length $20$. The thick paths illustrate character $a$: its length is
$1+2=3$ and its codeword is \texttt{110}.}
\label{fig:codeword-recovery}
}
\end{figure}
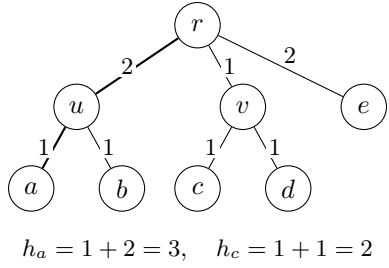
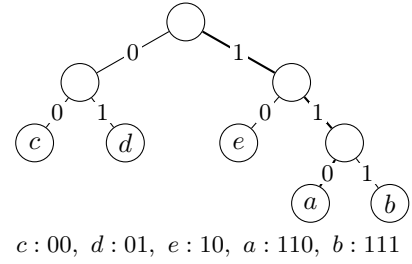

%% file: Sections/codeword-intervals-figure.tex
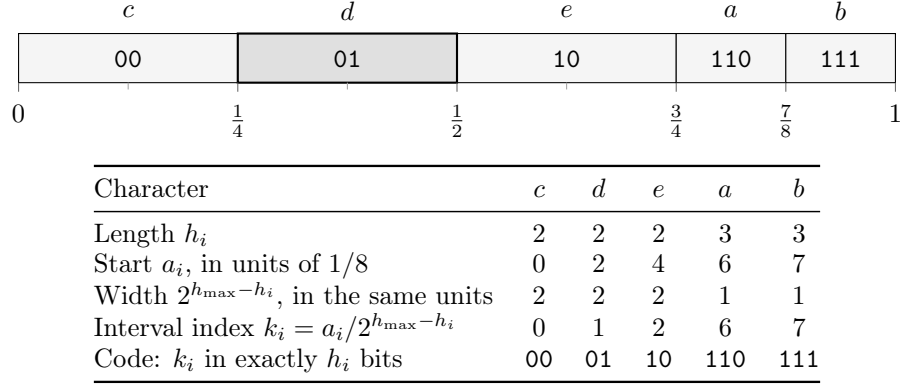
\begin{figure}[!t]
\centering
{\color{black}
\begin{tikzpicture}[x=1.45cm,y=1cm,font=\small]
\foreach \l/\r/\ch/\code in {0/2/c/00,4/6/e/10,6/7/a/110,7/8/b/111}{
 \draw[fill=black!4,draw=black] (\l,0) rectangle (\r,0.65);
 \node at ({(\l+\r)/2},0.95) {$\ch$};
 \node at ({(\l+\r)/2},0.32) {\texttt{\code}};
}
\draw[fill=black!12,draw=black,thick] (2,0) rectangle (4,0.65);
\node[text=black] at (3,0.95) {$d$};
\node[text=black] at (3,0.32) {\texttt{01}};
\foreach \x in {1,3,5}
 \draw[black!35] (\x,0)--(\x,-0.09);
\foreach \x/\t in {0/0,2/{\frac14},4/{\frac12},6/{\frac34},7/{\frac78},8/1}{
 \draw[black!60] (\x,0)--(\x,-0.12);
 \node[anchor=north] at (\x,-0.16) {$\t$};
}
\end{tikzpicture}
\par\medskip
\small
\begin{tabular}{@{}lccccc@{}}
\toprule
Character & $c$ & \textcolor{black}{$d$} & $e$ & $a$ & $b$\\
\midrule
Length $\textcolor{black}{\LeafDepth{i}}$ & 2 & 2 & 2 & 3 & 3\\
Start $a_i$, in units of $1/8$ & 0 & 2 & 4 & 6 & 7\\
Width $2^{\textcolor{black}{\MaximumLeafDepth}-\textcolor{black}{\LeafDepth{i}}}$, in the same units & 2 & 2 & 2 & 1 & 1\\
Interval index $k_i=a_i/2^{\textcolor{black}{\MaximumLeafDepth}-\textcolor{black}{\LeafDepth{i}}}$ & 0 & \textcolor{black}{1} & 2 & 6 & 7\\
Code: $k_i$ in exactly $\textcolor{black}{\LeafDepth{i}}$ bits & \texttt{00} & \textcolor{black}{\texttt{01}} & \texttt{10} & \texttt{110} & \texttt{111}\\
\bottomrule
\end{tabular}
\caption{Constructing codewords from lengths $(2,2,2,3,3)$, with
characters ordered as $c,d,e,a,b$ and $\textcolor{black}{\MaximumLeafDepth}=3$.
The intervals are placed consecutively in $[0,1)$; each has length
$2^{-\textcolor{black}{\LeafDepth{i}}}$. For $d$, the interval starts at $2/8$ and has width
$2/8$, so its index among the four quarter-intervals is $2/2=1$,
counting from zero. Writing this index in two bits gives \texttt{01}:
first choose the left half of $[0,1)$, then its right half.}
\label{fig:codeword-intervals}
}
\end{figure}

%% file: Sections/codeword-membership-figure.tex
\begin{figure}[!t]
\centering
{\color{black}
\begin{tikzpicture}[x=1.5cm,y=0.85cm,font=\small,
 selected/.style={draw=black,fill=black!9,minimum width=1.2cm,minimum height=6mm},
 unselected/.style={draw=black,fill=white,text=black,minimum width=1.2cm,minimum height=6mm}]
\node[anchor=east] at (-0.6,0) {First selection: $P(1)=2$};
\node[anchor=east] at (-0.6,-1) {Second selection: $P(2)=6$};
\node[selected] at (0,0) {$a:1$};
\node[selected] at (1,0) {$b:1$};
\node[unselected] at (2,0) {$c:2$};
\node[unselected] at (3,0) {$d:2$};
\foreach \x/\t in {0/{a:1},1/{b:1},2/{c:2},3/{d:2}}
 \node[selected] at (\x,-1) {$\t$};
\draw[decorate,decoration={brace,mirror,amplitude=4pt},black!65] (-0.4,-1.55)--(1.4,-1.55);
\draw[decorate,decoration={brace,mirror,amplitude=4pt},black!65] (1.6,-1.55)--(3.4,-1.55);
\node[align=center,font=\footnotesize] at (0.5,-2.3) {First selected: $a,b$\\node $u$, frequency $2$};
\node[align=center,font=\footnotesize] at (2.5,-2.3) {Newly selected: $c,d$\\node $v$, frequency $4$};
\end{tikzpicture}
\caption{Recovering membership at depth $h^*=1$ for the local depths
$h_a=h_b=h_c=h_d=2$. Each input leaf has Kraft value $1/4$,
so the first and second selections have Kraft values $1/2$ and $1$.
Blue cells are selected. Differences between the nested selections
identify the two groups. Each character receives a connection of
length $2-1=1$ to its group; $e$ is outside this call and remains unchanged.}
\label{fig:codeword-membership}
}
\end{figure}
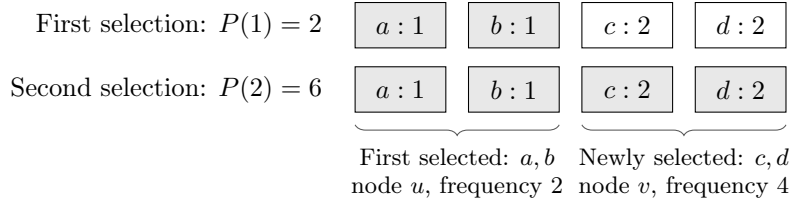

%% file: Sections/pram.tex
\section{Huffman Coding on PRAMs}
\label{sec:pram}

We next consider the parallel running time of our algorithm in the
\PRAM model. {\color{black}In the \MPC setting, the input is a string of length \(n\). In the \PRAM model, that setting is not particularly interesting: a straightforward parallel implementation of Huffman's algorithm, similar to Trivial Approach~2 in the introduction, already computes the optimum encoded length in roughly logarithmic time with \(\softO(n)\) total work. Moreover, as we show in Lemma~\ref{lem:pram-lower-bound}, the parallel running time of the problem is subject to a nearly logarithmic lower bound, which leaves little room for improvement when the input is the raw string. We therefore consider the setting in which the input consists of the \(m\) frequencies themselves, so that the input size is \(m\) rather than \(n\), and ask whether the parallel running time can be made substantially smaller than \(\log n\) while keeping the total work close to linear in \(m\). We assume throughout that \(\softO_\epsilon(m^{1+\epsilon})\) processors are available, for a fixed \(0<\epsilon<1\), where \(\softO_\epsilon\) hides factors depending on \(\epsilon\) and factors polylogarithmic in \(n\).}

The input in this section consists of \(m\) positive
integer frequencies, given in arbitrary order, and we write
\(n=\sum_{i=1}^m f_i\).
Every frequency and the optimum encoded length fit in
\(O(\log n)\) bits, since \(m\leq n\) and the optimum is at most
\(n\lceil\log_2 m\rceil\).  The task is to compute the optimum encoded length.

A \PRAM consists of processors that perform synchronous operations on
shared memory. We count both the number of parallel steps and the total
number of operations, which we call the \emph{work}. {\color{black}The number of processors required is always bounded by the work, so we do not bound it explicitly in our results.} Our algorithm uses
the common-CRCW variant. Concurrent reads are allowed, and processors
may write to the same cell simultaneously only when they write the same
value. Arithmetic and comparisons on \(O(\log n)\)-bit words take
constant time. {\color{black}Unlike in the \MPC model, local computation is not free: a step that takes a single \MPC round may take more than constant \PRAM time, since, for example, concurrent writes cannot add up the values written to a cell.} We take logarithmic factors to be at least one.

\begin{table}[t]
\centering
\small
\renewcommand{\arraystretch}{1.15}
\begin{tabularx}{\textwidth}{@{}Xll@{}}
\toprule
Algorithm & Time & Total work \\
\midrule
\rowcolor{gray!20}
Teng~\cite{teng1987}
 & \(O(\log^2m)\) & \(O(m^6\log^2m)\) \\
\rowcolor{gray!20}
Atallah et al.~\cite{atallah1989}, first bound
 & \(O(\log^2m)\) & \(O(m^2\log m)\) \\
\rowcolor{gray!20}
Atallah et al.~\cite{atallah1989}, second bound
 & \(O(\log m)\) & \(O(m^3)\) \\
Larmore--Przytycka~\cite{larmore1995}
 & \(O(\sqrt m\log m)\) & \(O(m^{3/2}\log m)\) \\
Berman et al.~\cite{berman2007}
 & \(O(\log n)\) & \(O(m\log n)\) \\
\textcolor{black}{Milidi\'u et al.~\cite{milidiu1999}}
 & \textcolor{black}{\(O(\log n\log\log m)\)} & \textcolor{black}{\(O(m)\)} \\
\midrule
\textbf{Theorem~\ref{thm:pram-prefix-input}} \textcolor{black}{(for sorted prefix sums)}
 & {\color{black}\(O_\epsilon\bigl((1+\tfrac{\log n}{\log m})\log\log m\bigr)\)}
 & \(O_\epsilon(m^{1+\epsilon}\log n\log^3m)\) \\
\addlinespace
\textbf{Corollary~\ref{cor:pram-frequencies}}
 & {\color{black}\(O_\epsilon\bigl(\tfrac{\log m}{\log\log m}+(1+\tfrac{\log n}{\log m})\log\log m\bigr)\)}
 & \(O_\epsilon(m^{1+\epsilon}\log n\log^3m)\) \\
\bottomrule
\end{tabularx}
\caption{{\color{black}PRAM comparison of parallel time and total work. The rows in light gray are NC algorithms.}
{\color{black}Atallah et al., Larmore--Przytycka,} {\color{black}Berman et al., and Milidi\'u et al.}\ {\color{black}use CREW; our results use common-CRCW.}}
\label{tab:pram-comparison}
\end{table}

{\color{black}
{\color{black}We begin by establishing a lower bound on the parallel running time. In Lemma~\ref{lem:pram-lower-bound}, we show that computing the optimum encoded length requires \(\Omega(\log m/\log\log m)\) time on a deterministic common-CRCW \PRAM with polynomially many processors. Our proof reduces counting to Huffman coding, showing that this lower bound holds even when the frequencies are small and within a factor of two of each other.}

\begin{lemma}[Lower bound]
\label{lem:pram-lower-bound}
{\color{black}Computing the optimum encoded length from \(m\) positive integer
frequencies requires \(\Omega(\log m/\log\log m)\) time in the worst case on a
deterministic common-CRCW \PRAM with \(m^{O(1)}\) processors.}
\end{lemma}

\begin{proof}
{\color{black}We reduce integer summation to computing the optimum encoded length.
Let \(a_1,\ldots,a_m\) be nonnegative integers. For simplicity, assume that \(m\) is a power of two, say \(m=2^k\); otherwise, we pad the input with dummy values \(a_i=0\), which does not change the sum and at most doubles the input size.
Let \(a_{\textsf{max}}=\max_i(a_i+1)\), which can be computed in \(O_\epsilon(1)\) time with \(O(m^{1+\epsilon})\) processors on a common-CRCW \PRAM for any fixed \(\epsilon>0\), and set \(f_i=a_i+a_{\textsf{max}}\), in parallel for all \(i\).
Every frequency then lies in \([a_{\textsf{max}},2a_{\textsf{max}})\). Therefore, all leaves have the same depth \(k\) in an optimal Huffman tree, and the optimum encoded length \(\MinimumEncodedLength\) satisfies
\[
 \MinimumEncodedLength=k\sum_{i=1}^m f_i
   =k\left(ma_{\textsf{max}}+\sum_{i=1}^m a_i\right).
\]
The sum of the original integers is therefore \(\MinimumEncodedLength/k-ma_{\textsf{max}}\).
Both the reduction and this recovery take \(O_\epsilon(1)\) parallel time.
In particular, taking \(a_i\in\{0,1\}\), an algorithm for Huffman coding computes the number of ones among \(a_1,\ldots,a_m\), and hence their parity.
Beame and H{\aa}stad~\cite{beame1989optimal} show that computing the parity of \(m\) bits requires \(\Omega(\log m/\log\log m)\) time on a common-CRCW \PRAM with polynomially many processors; in fact, their bound holds even in the stronger priority-CRCW model. The same bound follows for our problem.}
\end{proof}

{\color{black}Note that the lower bound holds even when the ratio between
the largest and smallest frequencies is less than two, and every frequency is at most three.}
{\color{black} Our framework shows that this lower bound is not inherent to Huffman coding, but rather stems from the representation of the input. Indeed, the reduction derives its hardness from integer summation, which becomes trivial once the prefix sums of the frequencies are available. As we show in Theorem~\ref{thm:pram-prefix-input}, if the input is provided as the prefix sums of the sorted frequencies, our algorithm runs in \(O_\epsilon\bigl((1+\frac{\log n}{\log m})\log\log m\bigr)\) time, which goes beyond the \(\Omega(\log m/\log\log m)\) bound as long as \(n\leq m^{o(\log m/(\log\log m)^2)}\); for instance, it gives \(O_\epsilon(\log\log m)\) time when the frequencies are bounded by a polynomial in \(m\), as in the hard instances above.}

We now turn to our upper bound. {\color{black}In Theorem~\ref{thm:pram-prefix-input}, we establish the time and work bounds when the input is given as the prefix sums of the sorted frequencies. In Corollary~\ref{cor:pram-frequencies}, we extend this to arbitrary, unsorted frequencies, at the cost of an additive \(O(\log m/\log\log m)\) term for sorting the frequencies and computing their prefix sums. By Lemma~\ref{lem:pram-lower-bound}, this additive term is necessary in the worst case.}

\begin{theorem}
\label{thm:pram-prefix-input}
{\color{black}For every fixed \(0<\epsilon<1\), suppose that the prefix sums of
\(m\geq2\) sorted positive integer frequencies with total \(n\)
are supplied. The optimum encoded length can be computed
deterministically on a common-CRCW \PRAM with \(O(\log n)\)-bit
words in}
 $$\textcolor{black}{O\Biggl(\epsilon^{-2}\biggl(1+\frac{\log n}{\log m}\biggr)\log\log m\Biggr)}$$ {\color{black}time and}
$$\textcolor{black}{O\Biggl(\epsilon^{-2}\biggl(1+\frac{\log n}{\log m}\biggr) m^{1+\epsilon}\log^4m\Biggr)}$$ {\color{black}work.}
\end{theorem}

{\color{black}If the frequencies are instead given in arbitrary order, we first sort them and compute their prefix sums, which adds \(O(\log m/\log\log m)\) to the parallel time. This gives Corollary~\ref{cor:pram-frequencies}.}

\begin{corollary}
\label{cor:pram-frequencies}
{\color{black}For every fixed \(0<\epsilon<1\), the optimum encoded length of
\(m\geq2\) positive integer frequencies, given in arbitrary order, with total \(n\) can be
computed deterministically on a common-CRCW \PRAM with
\(O(\log n)\)-bit words in}
\[
 {\color{black}O\Biggl(\frac{\log m}{\log\log m}+\epsilon^{-2}\biggl(1+\frac{\log n}{\log m}\biggr)\log\log m\Biggr)}
 \]
 {\color{black}time and}
 \[ \textcolor{black}{O\Biggl(\epsilon^{-2}\biggl(1+\frac{\log n}{\log m}\biggr) m^{1+\epsilon}\log^4m\Biggr)}
\] {\color{black}work.}
\end{corollary}

\paragraph{Comparison with earlier algorithms.}
{\color{black}Most of the previous work on parallel Huffman coding aims at algorithms in NC, that is, algorithms with polylogarithmic parallel time in \(m\), and as a result their total work deviates considerably from the optimal \(\softO(m)\). Our goal is different: rather than aiming for an NC algorithm, we aim to minimize the parallel time while keeping the total work close to linear.} {\color{black}For instance, the NC algorithms of Teng~\cite{teng1987} and Atallah et al.~\cite{atallah1989} achieve polylogarithmic parallel time, but require at least quadratic work. Larmore--Przytycka~\cite{larmore1995} trade some of this parallel time for less work, achieving $O(\sqrt m\log m)$ time, but their work is still $O(m^{3/2}\log m)$.} {\color{black}Karpinski, Larmore, and Nekrich~\cite{karpinski2004} parameterize the number of processors as $p=m^x$ for $0<x<1$ and obtain $O(m^{1-x/2}\log m)$ time with $O(m^{1+x/2}\log m)$ total work; thus, their running time is polynomial in $m$. Milidi\'u et al.~\cite{milidiu1999} and Berman et al.~\cite{berman2007} obtain running times proportional to the height of the Huffman tree, which can be $\Theta(\log n)$.} In contrast, our algorithm achieves $O_\epsilon(m^{1+\epsilon}\log n\log^3m)$ work. {\color{black}Moreover, our parallel running time is asymptotically better than that of all previous algorithms whenever \(n\leq 2^{O\left(\left(\frac{\log m}{\log\log m}\right)^2\right)}\). Indeed, in this range we have \(\frac{\log n}{\log m}\log\log m=O\bigl(\frac{\log m}{\log\log m}\bigr)\), so Corollary~\ref{cor:pram-frequencies} gives \(O_\epsilon(\log m/\log\log m)\) time, whereas the time bounds of all previous algorithms in Table~\ref{tab:pram-comparison} are at least logarithmic in \(m\). For instance, when the frequencies are bounded by a polynomial in \(m\), we have \(n\leq m^{O(1)}\), and our algorithm runs in \(O_\epsilon(\log m/\log\log m)\) time with \(O_\epsilon(m^{1+\epsilon}\log^4 m)\) work. By Lemma~\ref{lem:pram-lower-bound}, this running time is optimal, even though the work is close to linear.}

\subsection{Proof of the PRAM Bounds}

We now prove \textcolor{black}{Theorem~\ref{thm:pram-prefix-input} and Corollary~\ref{cor:pram-frequencies}}. Our \PRAM algorithm follows the same iterations as the \MPC algorithm, but requires careful parallel implementations of the computations within each iteration. {\color{black}We first describe the standard parallel primitives that we use throughout. We then show that the prefix sums can be maintained throughout the algorithm, eliminating the need to recompute them in each iteration. Next, we develop parallel implementations of Bounded-Ratio Huffman Coding and Iterative-State Reconstruction using these primitives.} Finally, we combine these ingredients to establish the claimed time and work bounds.

\paragraph{Standard Parallel Primitives.}
We use four standard parallel operations throughout our
algorithm: searching in sorted arrays, finding minima,
computing prefix minima, and summing small collections
of values. We briefly describe these operations and
their time and work bounds on a common-CRCW \PRAM.

Fix a constant $0<\delta<1$, and set
\[
\PRAMFanout=\lceil m^\delta\rceil,
\]
where $m$ is the original number of frequencies.
This parameter controls the trade-off between
parallel time and work.

\begin{itemize}
	
	\item \textbf{Parallel search.}
	Given a sorted array and a query value, the goal is
	to find the position of the query in the array.
	We use a multiway search: at each step, we compare
	the query with $\PRAMFanout$ approximately equally
	spaced entries in parallel and continue in the
	appropriate interval. Ties are resolved by a
	fixed ordering. For arrays of length $m^{O(1)}$,
	this gives~\cite{karp1990parallel}
	\begin{equation}
		\label{eq:pram-search-bound}
		\begin{aligned}
			\text{Time} &= O(\delta^{-1}),\\
			\text{Work} &= O(\delta^{-1}\PRAMFanout).
		\end{aligned}
	\end{equation}
	We use parallel searches to find frequency-block
	boundaries, locate cutoff positions, and merge
	sorted lists.
	
	\item \textbf{Parallel minimum.}
	Given $b$ values, the goal is to find their minimum
	and an index attaining it.
	We compare every pair in parallel and mark each
	index that loses a comparison. Ties are broken
	by index, leaving exactly one unmarked index.
	All concurrent writes to a marked position write
	the same value, as required by common-CRCW.
	This standard procedure gives~\cite{karp1990parallel}
	\begin{equation}
		\label{eq:pram-minimum-bound}
		\begin{aligned}
			\text{Time} &= O(1),\\
			\text{Work} &= O(b^2).
		\end{aligned}
	\end{equation}
	We use it to find minima among candidates in
	dynamic programming table computations.
	
	\item \textbf{Parallel prefix minima.}
	Given an array $a_1,\ldots,a_\nu$, the goal is
	to compute
	$\min_{1\leq i\leq t}a_i$ for every $t\in[\nu]$.
	This is an instance of parallel prefix
	computation~\cite{blelloch1990prefix}.
	We use a $\PRAMFanout$-ary tree, computing the
	prefix minima of the children's subtree minima
	at each node by parallel comparisons, and then
	propagating the results downward.
	The tree has $O(\delta^{-1})$ levels, and each
	level requires $O(\nu\PRAMFanout^2)$ work.
	For $\nu=m^{O(1)}$, we obtain
	\begin{equation}
		\label{eq:pram-prefix-minimum-bound}
		\begin{aligned}
			\text{Time} &= O(\delta^{-1}),\\
			\text{Work} &=
			O(\delta^{-1}\nu\PRAMFanout^2).
		\end{aligned}
	\end{equation}
	We use prefix minima to construct the initial
	tables for Bounded-Ratio Huffman Coding.

\item \textbf{Parallel reduction and scan.}
Given $b$ values, a reduction computes their sum,
while a scan computes all prefix sums.
Both operations admit standard balanced-tree
implementations~\cite{blelloch1990prefix,jaja1992introduction},
with
\begin{equation}
	\label{eq:pram-reduction-bound}
	\begin{aligned}
		\text{Time} &= O(\log b),\\
		\text{Work} &= O(b).
	\end{aligned}
\end{equation}
We use these operations on $b=O(\log m)$ values
to combine information across depth groups and
recover depth intervals, giving
$O(\log\log m)$ parallel time.
	
\end{itemize}

Table~\ref{tab:pram-primitives} summarizes the
bounds of these four operations.

\begin{table}[t]
	\centering
	\small
	\renewcommand{\arraystretch}{1.15}
	\begin{tabular}{@{}lll@{}}
		\toprule
		Operation & Parallel time & Work \\
		\midrule
		Search
		& $O(\delta^{-1})$
		& $O(\delta^{-1}\PRAMFanout)$ \\
		Minimum ($b$ values)
		& $O(1)$
		& $O(b^2)$ \\
		Prefix minima ($\nu$ values)
		& $O(\delta^{-1})$
		& $O(\delta^{-1}\nu\PRAMFanout^2)$ \\
		Reduction / scan
		& $O(\log\log m)$
		& $O(\log m)$ \\
		\bottomrule
	\end{tabular}
	\caption{Parallel primitives used in our common-CRCW
		\PRAM algorithm. Search and prefix minima operate
		on arrays of polynomial length in $m$; reduction
		and scan operate on $O(\log m)$ values.}
	\label{tab:pram-primitives}
\end{table}

\paragraph{Maintaining Prefix Sums Across Iterations.}
{\color{black}We first show that two sorted lists whose prefix sums are known can be merged without recomputing any prefix sums.}

\begin{lemma}
\label{lem:pram-prefix-preservation}
{\color{black}Given two sorted lists of total length \(s\leq m\), together with their prefix sums, their sorted merge and its prefix sums can be computed in \(O(\delta^{-1})\) parallel time and \(O(\delta^{-1}s\PRAMFanout)\) work.}
\end{lemma}

\begin{proof}
{\color{black}Let the two lists be $L_1$ and $L_2$, and break ties by a fixed ordering.
For each element, we find its rank in the other list by a parallel search~\eqref{eq:pram-search-bound}, in $O(\delta^{-1})$ time and $O(\delta^{-1}\PRAMFanout)$ work.
An element's position in the merged list is its position in its
own list plus the number of elements preceding it in the other list.
If a prefix of the merged list contains $i$ elements of $L_1$ and
$j$ elements of $L_2$, its sum is
\[
 \PrefixSum{L_1}{i}+\PrefixSum{L_2}{j}.
\]
For each output position, the ranks determine $i$ and $j$.
Thus each element and its prefix sum can be written independently
in constant time.}
\end{proof}

{\color{black}Lemma~\ref{lem:pram-prefix-preservation} is what allows us to compute prefix sums only once. Every iteration of our algorithm needs the current frequencies in sorted order together with their prefix sums: the bounded-ratio solver and Iterative-State Reconstruction repeatedly read the total frequency of consecutive runs, each as a difference of two prefix sums. Computing prefix sums from scratch is at least as hard as computing a sum, so by the argument of Lemma~\ref{lem:pram-lower-bound} it would cost \(\Omega(\log m/\log\log m)\) time in every iteration. Instead, an iteration ends with two sorted lists whose prefix sums are already known: the stopping state returned by Iterative-State Reconstruction, and the untouched suffix of the input, whose prefix sums are those of the input minus a single constant, namely the total frequency of the processed prefix. Lemma~\ref{lem:pram-prefix-preservation} merges these two lists into the input of the next iteration together with its prefix sums, in \(O(\delta^{-1})\) time.}

\begin{example}
\label{ex:pram-prefix-maintenance}
{\color{black}Suppose the current frequencies are
\[
(2,2,3,3,4,4,5,5,7,14),
\]
whose sorted prefix sums are $(2,4,7,10,14,18,23,28,35,49)$.
For illustration, let $\lambda=3$. The bounded-ratio prefix consists of all frequencies at most $\lambda f_2=6$, namely $(2,2,3,3,4,4,5,5)$, and the untouched suffix is $(7,14)$.
We solve the prefix with Bounded-Ratio Huffman Coding and pass it to Iterative-State Reconstruction with threshold $\tau=7$, the smallest frequency outside the prefix. Iterative-State Reconstruction returns the sorted frequencies $(8,9,11)$ that remain after the merges within the prefix, together with their prefix sums $(8,17,28)$ and the merge cost $38$.
The untouched suffix $(7,14)$ has prefix sums $(35-28,49-28)=(7,21)$, obtained from the input by subtracting the total $28$ of the processed prefix.
Before repeating the same procedure, we merge these two sorted lists by Lemma~\ref{lem:pram-prefix-preservation}, which gives the next input
\[
(7,8,9,11,14)
\]
with prefix sums $(7,15,24,35,49)$.
Figure~\ref{fig:pram-iteration-example} illustrates this iteration.
No prefix sum is recomputed: each value is either read from an earlier prefix-sum array or obtained as the sum or difference of two such values.}
\end{example}

\input{Sections/pram-iteration-example-figure}

\paragraph{Solving a bounded-ratio instance.}
We parallelize the divide-and-conquer algorithm of
Section~\ref{sec:bounded-ratio-formal}, using the table-merging
procedure of Section~\ref{sec:layered-merging-formal}.
The parameters retain their earlier roles: $\lambda$ bounds the
frequency ratio, and $\PRAMFanout$ controls the number of positions
examined in each search step. For the PRAM implementation, we choose
$\lambda=\PRAMFanout=\lceil m^\delta\rceil$.
The computation has three steps: constructing the base tables,
merging them through a balanced hierarchy, and recovering an
optimal depth assignment.

We first recall the budget bound from
Section~\ref{sec:bounded-ratio-formal}.
By Corollary~\ref{cor:bounded-ratio-depth-range},
the maximum scaled Kraft value is
\[
\KraftScaleBound=2^{u-\ell}
=O(\lambda^{2\log_\varphi 2})
=O(\lambda^3).
\]
Consequently, the total number of table entries
at any level of the balanced hierarchy is
$O(m\KraftScaleBound)$.
If $\ell=u$, the optimum cost is computed directly.

\begin{itemize}
	
\item \textbf{Constructing the base tables.}
We place $f_1$ in a separate block and partition the remaining
frequencies into the intervals
$[\varphi^t f_2,\varphi^{t+1}f_2)$.
There are $O(1+\log\lambda)$ such intervals.
To find their boundaries, we first compute the thresholds
$f_2,\varphi f_2,\varphi^2 f_2,\ldots$.
Since we need only $O(\log m)$ powers of $\varphi$,
these thresholds can be computed in
$O(\log\log m)$ parallel time using repeated squaring.
We then locate the block boundaries by parallel searches.\footnote{
	The powers of $\varphi$ can be represented exactly as
	$(a+b\sqrt5)/2$, with integer coefficients of
	$O(\log m)$ bits. Comparisons with integer frequencies
	reduce to integer arithmetic, so no rounding of the
	thresholds is needed.}
Each nonempty block has frequency ratio smaller than
$\varphi$, so Lemma~\ref{lem:constant-time-base-case}
applies directly.

For each block $B$, we compute $\mu(B,b)$
for every relevant Kraft budget $b$, using
Lemma~\ref{lem:constant-time-base-case}.
Recall that $\mu(B,b)$ is the minimum encoded
length achievable with a total scaled Kraft
usage of at most $b$.
Since we maintain prefix sums of the sorted
frequencies, the sum over any consecutive interval
can be obtained in constant time.
Thus, each candidate depth assignment can be
evaluated in constant time.
We evaluate all candidates in parallel and use
prefix minima to obtain $\mu(B,b)$ for every budget $b$.

The lengths of all block tables sum to
$O(m\KraftScaleBound)$.
By the parallel prefix-minimum bound
in~\eqref{eq:pram-prefix-minimum-bound},
processing all these tables in parallel takes
$O(\delta^{-1})$ time and
$O(\delta^{-1}m\KraftScaleBound\PRAMFanout^2)$ work.
	
	\item \textbf{Merging the tables.}
	At each internal node of the balanced hierarchy, we combine
	two tables by minimizing over the ways to divide the budget.
	By Lemma~\ref{lem:layered-convexity}, the candidates for each
	entry form at most $\KraftScaleBound$ convex layers.
	
	Within each layer, we evaluate $\PRAMFanout$ approximately
	equally spaced positions, including the endpoints, and find
	the first minimum among them. By convexity, we retain only
	the interval between its neighboring sampled positions,
	using a one-sided interval at an endpoint.
	Once at most $\PRAMFanout$ positions remain, we evaluate
	all of them in parallel. The search takes $O(\delta^{-1})$
	steps, after which we take the minimum over the layers.
	
	There are $O(m\KraftScaleBound)$ table entries at each
	hierarchy level. Using the parallel minimum procedure for
	the sampled positions and the layers gives
	$O(\delta^{-1}m\KraftScaleBound^3\PRAMFanout^2)$ work
	per level. The hierarchy has $O(\log(2+\log\lambda))$ levels.
	Since $\lambda=m^{O(1)}$, table construction takes
	$O(\delta^{-1}\log\log m)$ parallel time.

\item \textbf{Recovering the depth assignment.}
During each merge, we store a budget split attaining
the minimum. Starting from the root, we recover the
budget of each block by following these stored splits.
All nodes at the same level are processed in parallel.
Since the balanced hierarchy has
$O(\log(2+\log\lambda))=O(\log\log m)$ levels,
this traceback takes $O(\log\log m)$ time.

By Lemma~\ref{lem:constant-time-base-case},
the depth assignment within each block consists
of at most three runs of equal depths.
Since there are $O(1+\log\lambda)$ blocks,
the recovered assignment has
$O(1+\log\lambda)$ runs in total.

The recovered depths may not be sorted across blocks.
We therefore reorder them so that smaller frequencies
receive larger depths, which preserves optimality.
To do this, we count the number of leaves at each depth
in parallel and compute prefix sums of these counts
in decreasing depth order. This determines which
consecutive frequencies receive each depth.
Using the maintained prefix sums of the sorted frequencies,
we also obtain the total frequency of each depth group
in constant time.

There are $O(1+\log\lambda)$ possible depths and
$O(1+\log\lambda)$ blocks. Thus, the reductions
and scan take $O(\log(2+\log\lambda))=O(\log\log m)$
parallel time. The resulting depth groups and their
frequency sums are used by Iterative-State Reconstruction.

\end{itemize}

Overall, Bounded-Ratio Huffman Coding takes
$O(\delta^{-1}\log\log m)$ parallel time and
\[
O(\delta^{-1}m\KraftScaleBound^3\PRAMFanout^2\log\log m)
\]
work, including table construction and recovery of the depth assignment.

\paragraph{{\color{black}Solving Iterative-State Reconstruction.}}
{\color{black}We now parallelize Iterative-State Reconstruction. Its main step computes the values $P(r)$,
using the method of Section~\ref{sec:computing-prefix-optima}.}
The algorithm first constructs the density order, identifies
the small residual instances around the fractional solutions,
and solves these instances in parallel. The resulting values
$P(r)$ are then used to reconstruct the stopping state.

Let $\DepthGroupCount$ be the number of nonempty depth groups.
As established above,
$\DepthGroupCount=O(1+\log\lambda)$.

\begin{itemize}
	
	\item \textbf{Constructing the density order.}
	Recall that the density of a leaf at depth $h$ is $f_i2^h$.
	Within each depth group, density order agrees with
	frequency order, so its leaves are already sorted.
	
	For each leaf, we use parallel searches to count how many
	leaves precede it in each depth group, according to density
	order. Summing these counts gives its position in the
	global density order. We use a fixed tie-breaking rule
	to make all positions unique.
	
	We also compute the total scaled Kraft value and total
	frequency of every prefix in this order. If a prefix
	contains $t_h$ leaves from the group of depth $h$,
	its scaled Kraft value is
	\[
	\sum_h t_h 2^{h_{\max}-h}.
	\]

For each position in the global density order, we know
how many leaves precede it in each depth group.
Since the leaves within each group are sorted by frequency,
these leaves form a prefix of that group.
We obtain the total frequency of each such prefix
from the maintained prefix sums of the sorted frequencies.
Summing over the $\DepthGroupCount$ groups gives the
frequency prefix sum at that position.

The same counts also determine the Kraft prefix sum,
since every leaf in a depth group has the same Kraft value.

We compute these sums by parallel reductions over the
depth groups, avoiding a prefix-sum scan over all leaves.
Together with the parallel searches, this takes
$O(\delta^{-1}\log\log m)$ time and
$
O(\delta^{-1}m\DepthGroupCount\PRAMFanout)
$
work.

	\item \textbf{Identifying the residual instances.}
	We construct the residual instances for all targets in parallel.
	For each target, we use parallel predecessor searches to locate
	the fractional cutoff in the Kraft prefix array and the
	corresponding position in each depth group. We then identify
	the windows given by Lemma~\ref{lem:fractional-proximity},
	storing their endpoints in each group. These windows contain
	$O(\DepthGroupCount\KraftScaleBound)$ leaves in total and
	suffice to recover an optimal integral solution.

	Leaves before the windows are fixed as selected, and leaves
	after them are fixed as unselected. Within each group, their
	total frequency is obtained from the maintained prefix sums,
	and their total scaled Kraft value from their count.
	Parallel reductions over the groups give the fixed cost
	and the Kraft value to subtract from the target.
	As in Section~\ref{sec:computing-prefix-optima}, the remaining
	exact budget is at most
	\[
	\ResidualBudgetBound=O(\DepthGroupCount\KraftScaleBound^2).
	\]
	The searches and reductions take
	$O(\delta^{-1}+\log(2+\DepthGroupCount))$ time and
	$O(\delta^{-1}\DepthGroupCount\PRAMFanout)$ work per target.
	All targets are processed simultaneously, so their number
	multiplies the work but not the parallel time.

\item \textbf{Solving the residual instances.}
We need to select a subset of the retained leaves
whose total scaled Kraft value equals the remaining
budget, while minimizing their total frequency.
We solve this problem using a balanced dynamic
program over the depth groups.

For any collection of depth groups, let $D(t)$
denote the minimum total frequency of a selection
with scaled Kraft value exactly $t$.
We first construct the table for each individual
depth group. Since all leaves in a group have
the same scaled Kraft value $a$, achieving
budget $t$ requires selecting exactly $t/a$ leaves.
If $t/a$ is not an integer between zero and the
number of retained leaves, we set $D(t)=+\infty$.
Otherwise, we select the $t/a$ leaves with
smallest frequencies. Their total frequency
is obtained in constant time from the maintained
prefix sums.

We then merge the tables in a balanced hierarchy.
For two disjoint collections of depth groups,
we compute
\[
D(t)=\min_{0\leq z\leq t}
\{D_1(z)+D_2(t-z)\},
\qquad
0\leq t\leq\ResidualBudgetBound.
\]
Here, $z$ is the budget assigned to the first
collection, and $t-z$ is assigned to the second.

To parallelize each merge, we evaluate all
candidate budget splits simultaneously.
Each table has $O(\ResidualBudgetBound+1)$ entries,
and each entry considers
$O(\ResidualBudgetBound+1)$ splits.
Using the parallel minimum procedure, we find
the best split for each entry in constant time
with $O((\ResidualBudgetBound+1)^2)$ work.
Thus, one merge takes constant time and
$O((\ResidualBudgetBound+1)^3)$ work.

There are $O(\DepthGroupCount)$ merges, arranged
in $O(\log(2+\DepthGroupCount))$ levels.
Since all merges at the same level run in parallel
and $\DepthGroupCount=O(1+\log\lambda)$,
the total time is $O(\log\log m)$.

Finally, since
$\ResidualBudgetBound=
O(\DepthGroupCount\KraftScaleBound^2)$,
the total work for one residual instance is
\[
O(\DepthGroupCount\ResidualBudgetBound^3)
=
O(\DepthGroupCount^4\KraftScaleBound^6).
\]
Adding the total frequency of the fixed selected
leaves to the minimum cost of the residual instance
gives the corresponding value $P(r)$.
	
	\item \textbf{Recovering the stopping state.}
We choose depth $\textcolor{black}{\AveragingDepth}$ as in Section~\ref{sec:iterative-state-formal}.
The total frequency is available from the maintained prefix sums;
we find $h_{\min}$ by a parallel reduction over the
$\DepthGroupCount$ depth groups in $O(\log\log m)$ time.
Since $2^{\textcolor{black}{\AveragingDepth}}=O(m)$, we compute all values $P(0),\ldots,P(2^{\textcolor{black}{\AveragingDepth}})$
simultaneously, multiplying the work per target by $O(m)$
without increasing the parallel time.
We then compute their consecutive differences and the consecutive
pair sums in parallel, using $O(1)$ time and $O(m)$ work.
These give the frequencies at depths $\AveragingDepth$ and $\AveragingDepth-1$,
whose prefix sums are already available as $P(j)$ and $P(2j)$, respectively.
{\color{black}As in Section~\ref{sec:iterative-state-formal}, we choose $h^*\in\{\AveragingDepth-1,\AveragingDepth\}$ by a parallel minimum, and write $w_1\leq\cdots\leq w_q$ for the frequencies at depth $h^*$ and $\PrefixSum{w}{t}$ for their prefix sums.
If $h^*=0$, the stopping state is the root, and the merge cost is the optimal encoded length of the prefix.
Otherwise, we find the number $s$ of frequencies at most $\tau$ by a parallel search, and let $k=\lfloor s/2\rfloor$.
The stopping state consists of the parents $V=(w_1+w_2,\ldots,w_{2k-1}+w_{2k})$, with prefix sums $\PrefixSum{V}{t}=\PrefixSum{w}{2t}$, and the unpaired frequencies $U=(w_{2k+1},\ldots,w_q)$, with prefix sums $\PrefixSum{U}{t}=\PrefixSum{w}{2k+t}-\PrefixSum{w}{2k}$.
Applying Lemma~\ref{lem:pram-prefix-preservation} to $U$ and $V$ gives the stopping state in sorted order together with its prefix sums, in $O(\delta^{-1})$ time and $O(\delta^{-1}m\PRAMFanout)$ work.
The merge cost is
\[
\sum_h (h-h^*)\sum_{i:\,h_i=h}f_i+\PrefixSum{w}{2k}.
\]
Each depth group's frequency sum is obtained from two entries of the maintained prefix-sum array, so a parallel reduction over the $\DepthGroupCount$ depth groups evaluates this cost in $O(\log\log m)$ time.}

\end{itemize}

{\color{black}Overall, Iterative-State Reconstruction returns the stopping state in sorted order, its prefix sums, and the merge cost in}
$O(\delta^{-1}\log\log m)$ parallel time and
\[
O\bigl(
m\DepthGroupCount^4\KraftScaleBound^6
+\delta^{-1}m\DepthGroupCount\PRAMFanout
\bigr)
\]
work.

We now combine these implementations to prove {\color{black}Theorem~\ref{thm:pram-prefix-input}}.

\begin{proof}[Proof of Theorem~\ref{thm:pram-prefix-input}]
Fix $0<\epsilon<1$ and set
\[
\delta=\epsilon/100,\qquad
\PRAMFanout=\lceil m^\delta\rceil,\qquad
\lambda=\PRAMFanout,
\]
where $m$ is the original number of frequencies.
We run the iterative algorithm of
Section~\ref{sec:putting-pieces-together} directly
on the frequencies.\footnote{We omit the small-frequency preprocessing
of Lemma~\ref{lem:small-weight-preprocessing}.}
Throughout the algorithm, we maintain the current frequencies
in sorted order together with their prefix sums.
{\color{black}Initially, we recover each sorted frequency by subtracting two
consecutive supplied prefix sums, in \(O(1)\) parallel time and \(O(m)\) work.}

We first analyze one iteration.
We find the maximal bounded-ratio prefix by parallel predecessor
search and solve it using the parallel bounded-ratio algorithm.
If this prefix contains the entire current list, we add its
optimal encoded length to the running cost and terminate.

{\color{black}Otherwise, we run Iterative-State Reconstruction on the prefix, as described above, with threshold $\tau$ equal to the smallest frequency outside the prefix. It returns the merge cost, which we add to the running cost, and the stopping state in sorted order together with its prefix sums. The untouched suffix is also sorted, and its prefix sums are the current ones minus the total frequency of the processed prefix. Lemma~\ref{lem:pram-prefix-preservation} merges these two lists in $O(\delta^{-1})$ time and $O(\delta^{-1}m\PRAMFanout)$ work, giving the sorted frequencies of the next iteration together with their prefix sums.}

We now bound the total parallel time.
In every nonterminal iteration, the second-smallest frequency
increases by a factor greater than $\lambda$.
Since all frequencies are positive integers bounded by $n$,
there are
\[
O(1+\log_\lambda n)
=O\left(\epsilon^{-1}\left(1+\frac{\log n}{\log m}\right)\right)
\]
iterations. Each iteration takes $O(\delta^{-1}\log\log m)$
time. The bounded-ratio and reconstruction algorithms have this
bound, and the remaining searches and updates take
$O(\delta^{-1})$ time. Therefore, the iterations take
\[
O\left(\epsilon^{-2}\left(1+\frac{\log n}{\log m}\right)
\log\log m\right)
\]
parallel time.

For the work bound, recall that
$\DepthGroupCount=O(\log m)$.
The algorithms described above use at most
\[
O\left(\delta^{-1}m\left[
\KraftScaleBound^3\PRAMFanout^2\log\log m
+\DepthGroupCount^4\KraftScaleBound^6
+\DepthGroupCount\PRAMFanout+\PRAMFanout^2
\right]\right)
\]
work per iteration. The first term accounts for the bounded-ratio
hierarchy, the second for the residual dynamic programs, and the
remaining terms for density ordering, searches, and updates.
Since $\PRAMFanout=O(m^\delta)$ and
$\KraftScaleBound=O(m^{3\delta})$, this is
\begin{align*}
&O\left(\delta^{-1}
\left(m^{1+11\delta}+m^{1+18\delta}+m^{1+2\delta}\right)
\log^4m\right)\\
&\qquad=O(\epsilon^{-1}m^{1+\epsilon}\log^4m).
\end{align*}
Multiplying by the number of iterations gives total work
\[
O\left(\epsilon^{-2}\left(1+\frac{\log n}{\log m}\right)
m^{1+\epsilon}\log^4m\right).
\]

The same per-iteration bound suffices for the number of processors
and memory words. The number of simultaneous operations cannot
exceed the total work of an iteration, and each operation uses
only a constant number of words. All temporary tables are discarded
after the next sorted list and its prefix sums have been formed.
Thus, the processor and memory bounds are
\[
O(\epsilon^{-1}m^{1+\epsilon}\log^4m),
\]
without an additional factor for the number of iterations.
The arrays can be indexed by group, target, and budget,
so no global compaction is needed.

All arithmetic fits in the stated word model.
Capacities, indices, and depths require $O(\log m)$ bits,
while frequencies and intermediate costs require $O(\log n)$ bits.
The depths used in the bounded-ratio instances are $O(\log m)$,
so multiplying frequencies by powers of two to obtain densities
adds only $O(\log m)$ bits.
The total encoded length is at most $n\lceil\log_2 m\rceil$,
which also fits in $O(\log n)$ bits since $m\leq n$.

{\color{black}Recovering the initial frequencies from the supplied prefix sums
is absorbed by these bounds, completing the proof.}
\end{proof}

{\color{black}It remains to prove Corollary~\ref{cor:pram-frequencies}, where the frequencies are given in arbitrary order. We first sort them and compute their prefix sums, and then apply Theorem~\ref{thm:pram-prefix-input}. Lemma~\ref{lem:pram-sorted-prefix-sums} bounds the cost of the first step. It builds on the prefix-sum algorithm of Cole and Vishkin~\cite{cole1989faster}, which computes the prefix sums of \(k\) integers with \(O(\log k)\) bits each in \(O(\log k/\log\log k)\) time and \(O(k)\) work on a common-CRCW \PRAM. Recall that we take logarithmic factors to be at least one.}

\begin{lemma}
\label{lem:pram-sorted-prefix-sums}
{\color{black}Let \(0<\epsilon<1\) be fixed, and let \(x_1,\ldots,x_m\) be nonnegative integers whose sum is at most \(n\), with \(m\geq2\).
On a common-CRCW \PRAM with \(O(\log n)\)-bit words, we can sort these integers and compute the prefix sums of the sorted sequence in
\[
O\left(\frac{\log m}{\log\log m}+\epsilon^{-2}+\log\left(1+\frac{\log n}{\log m}\right)\right)
\]
time and \(O\bigl(\epsilon^{-2}(1+\log n/\log m)m^{1+\epsilon}\bigr)\) work.}
\end{lemma}

\begin{proof}
{\color{black}We first show the following.
For any \(k\geq2\) and \(t\geq1\), the prefix sums of \(k\) integers in \([0,k^t]\) can be computed in \(O(\log k/\log\log k+\log t)\) time and \(O(kt)\) work.
Let \(c=\lceil\log_2 k\rceil+1\) and \(\ell=\lceil t\rceil+1=O(t)\).
We split each value \(y_i\) into \(\ell\) chunks of \(c\) bits, so that
\(y_i=\sum_{j=0}^{\ell-1}y_{i,j}2^{jc}\) with \(0\leq y_{i,j}<2^c\leq 4k\).
For each chunk position \(j\), the chunks \(y_{1,j},\ldots,y_{k,j}\) have \(O(\log k)\) bits,
so the algorithm of Cole and Vishkin computes their prefix sums
\(p_{i,j}=\sum_{i'\leq i}y_{i',j}\) in \(O(\log k/\log\log k)\) time and \(O(k)\) work.
We do this for all \(\ell\) chunk positions in parallel.
The \(i\)th prefix sum is then \(\sum_{i'\leq i}y_{i'}=\sum_{j=0}^{\ell-1}p_{i,j}2^{jc}\),
which a balanced-tree reduction over the \(\ell\) terms computes in \(O(\log \ell)\) time and \(O(\ell)\) work for each \(i\).

\emph{Sorting.}
We break ties by the original indices, so that all keys are distinct, and sort by multiway mergesort with fan-in
\(k=\lceil m^{\epsilon/2}\rceil\).
We start with \(m\) sorted lists of length one. In each level, we partition the current lists into groups of \(k\) consecutive lists and merge each group into one sorted list.
After \(\lceil\log m/\log k\rceil=O(\epsilon^{-1})\) levels, a single sorted list remains.
To merge a group, consider an element \(x\) of a list \(A\) in that group.
Its position in the merged list equals its rank in \(A\) plus, for each other list \(B\) in the group, the number of elements of \(B\) smaller than \(x\).
We find each of these \(k-1\) counts by a parallel search~\eqref{eq:pram-search-bound} with \(\delta=\epsilon/2\), in \(O(\epsilon^{-1})\) time and \(O(\epsilon^{-1}m^{\epsilon/2})\) work.
We then add the \(k\) values, each at most \(m\leq k^{2/\epsilon}\), using the bound above with \(t=2/\epsilon\),
in \(O(\log k/\log\log k+\log(1/\epsilon))\) time and \(O(\epsilon^{-1}k)\) work.
Finally, each element is written to its position.
Over all \(m\) elements, a level therefore takes
\(O(\epsilon^{-1}+\log k/\log\log k)\) time and \(O(\epsilon^{-1}m^{1+\epsilon})\) work.

Over the \(O(\epsilon^{-1})\) levels, the work is \(O(\epsilon^{-2}m^{1+\epsilon})\), and the time is
\[
O\left(\epsilon^{-2}+\left(\frac{\log m}{\log k}+1\right)\frac{\log k}{\log\log k}\right)
=O\left(\epsilon^{-2}+\frac{\log m}{\log\log k}\right).
\]
If \(\log m\geq\epsilon^{-2}\), then \(\log k\geq(\epsilon/2)\log m\geq\frac12\sqrt{\log m}\), so \(\log\log k=\Omega(\log\log m)\) and the last term is \(O(\log m/\log\log m)\).
Otherwise, \(\log m<\epsilon^{-2}\), and the last term is \(O(\epsilon^{-2})\).
In both cases, sorting takes \(O(\log m/\log\log m+\epsilon^{-2})\) time.

\emph{Prefix sums.}
We apply the bound above to the sorted sequence with \(k=m\) and \(t=\max\{1,\log_m n\}=O(1+\log n/\log m)\).
Since all prefix sums are at most \(n\), the \(O(\log n)\)-bit words suffice.
This takes
\(O\bigl(\log m/\log\log m+\log(1+\log n/\log m)\bigr)\) time and \(O\bigl(m(1+\log n/\log m)\bigr)\) work.
Combining the two steps proves the lemma.}
\end{proof}

\begin{proof}[Proof of Corollary~\ref{cor:pram-frequencies}]
{\color{black}By Lemma~\ref{lem:pram-sorted-prefix-sums}, we can sort the frequencies and compute the prefix sums of the sorted sequence in
\(O\bigl(\log m/\log\log m+\epsilon^{-2}+\log(1+\log n/\log m)\bigr)\) time and
\(O\bigl(\epsilon^{-2}(1+\log n/\log m)m^{1+\epsilon}\bigr)\) work.
We then apply Theorem~\ref{thm:pram-prefix-input} to these prefix sums.
Since \(\log(1+\log n/\log m)\leq\log n/\log m\), the additive \(\epsilon^{-2}+\log(1+\log n/\log m)\) term and the work of the first step are absorbed by the bounds of Theorem~\ref{thm:pram-prefix-input}, which proves the corollary.}
\end{proof}

%% file: Sections/pram-iteration-example-figure.tex
\begin{figure}[!t]
\centering
\begin{tikzpicture}[x=0.95cm,y=0.75cm,font=\small,
 cell/.style={minimum width=0.85cm,minimum height=0.5cm,draw=black,inner sep=0pt},
 pre/.style={cell,fill=black!25,draw=black},
 suf/.style={cell,fill=black!10,draw=black},
 stop/.style={cell,fill=black!25,draw=black},
 lab/.style={anchor=east,align=right}]
\node[lab] at (0.4,0) {Input};
\node[lab] at (0.4,-0.8) {Prefix sums};
\foreach \x/\v in {1/2,2/2,3/3,4/3,5/4,6/4,7/5,8/5} \node[pre] at (\x,0) {$\v$};
\foreach \x/\v in {9/7,10/14} \node[suf] at (\x,0) {$\v$};
\foreach \x/\v in {1/2,2/4,3/7,4/10,5/14,6/18,7/23,8/28,9/35,10/49} \node[cell] at (\x,-0.8) {$\v$};
\draw[decorate,decoration={brace,amplitude=4pt},black!85!black]
  (0.55,0.45)--(8.45,0.45) node[midway,above=5pt] {bounded-ratio prefix ($f_j\leq\lambda f_2=6$)};
\draw[decorate,decoration={brace,amplitude=4pt},black!65!black]
  (8.55,0.45)--(10.45,0.45) node[midway,above=5pt] {untouched};
\draw[-{Stealth},thick] (4,-1.35)--(4,-2.15) node[midway,right,align=left,font=\footnotesize] {Iterative-State\\Reconstruction, $\tau=7$};
\draw[-{Stealth},thick] (9.5,-1.35)--(9.5,-2.15) node[midway,right] {$-28$};
\node[lab] at (0.4,-2.7) {Two lists};
\node[lab] at (0.4,-3.5) {Prefix sums};
\foreach \x/\v in {3/8,4/9,5/11} \node[stop] at (\x,-2.7) {$\v$};
\foreach \x/\v in {3/8,4/17,5/28} \node[cell] at (\x,-3.5) {$\v$};
\foreach \x/\v in {9/7,10/14} \node[suf] at (\x,-2.7) {$\v$};
\foreach \x/\v in {9/7,10/21} \node[cell] at (\x,-3.5) {$\v$};
\draw[-{Stealth},thick] (6.5,-4.05)--(6.5,-4.85) node[midway,right] {Lemma~\ref{lem:pram-prefix-preservation}};
\node[lab] at (0.4,-5.4) {Next input};
\node[lab] at (0.4,-6.2) {Prefix sums};
\foreach \x/\v/\s in {3/7/suf,4/8/stop,5/9/stop,6/11/stop,7/14/suf} \node[\s] at (\x,-5.4) {$\v$};
\foreach \x/\v in {3/7,4/15,5/24,6/35,7/49} \node[cell] at (\x,-6.2) {$\v$};
\end{tikzpicture}
\caption{One iteration of Example~\ref{ex:pram-prefix-maintenance}.
Iterative-State Reconstruction turns the bounded-ratio prefix into the
stopping state $(8,9,11)$ together with its prefix sums.
The prefix sums of the untouched suffix are those of the input minus $28$,
the total of the processed prefix.
Lemma~\ref{lem:pram-prefix-preservation} merges the two lists, giving the input of
the next iteration together with its prefix sums.}
\label{fig:pram-iteration-example}
\end{figure}